%% file: paper_tit_v7.tex
\documentclass[journal,11pt, onecolumn,draftclsnofoot]{IEEEtran}

\usepackage{cite}
\usepackage{graphicx,color,epsfig,rotating}
\usepackage{tikz}
\usetikzlibrary{arrows.meta,positioning,calc}
\usepackage{amsfonts,amsmath,amssymb,bbm}
\usepackage{setspace}
\usepackage{algorithm}
\usepackage[algo2e]{algorithm2e} 
\usepackage{subcaption}
\usepackage{mdwtab}
\usepackage{placeins}
\usepackage{psfrag, graphicx}
\usepackage[latin1]{inputenc}
\usepackage{multirow}
\usepackage{stfloats}
\usepackage{tabularx} 
\usepackage{booktabs} 
\usepackage{url}
\usepackage{bm}
\usepackage{bbding}
\usepackage{pifont}
\usepackage{wasysym}
\usepackage{xspace}
\usepackage{tcolorbox}  
\usepackage{listings}
\usepackage{lipsum}
\usepackage{color}
\usepackage{tcolorbox}
\tcbuselibrary{skins,breakable}
\definecolor{lightpurple}{RGB}{245,243,255}
\newtcolorbox{findingbox}{
    enhanced,
    breakable,
    colback=gray!5,
    colframe=black!70,
    boxrule=1pt,
    arc=4mm,
    left=4pt,
    right=4pt,
    top=4pt,
    bottom=4pt,
     before skip=0pt,
    after skip=0pt
    drop shadow=black!25,
    fontupper=\normalsize
}

\usepackage{geometry}
\input{macros}  

\newcommand{\User}[1]{User~$#1$}
\newcommand{\user}[1]{user~$#1$}

\SetKwInput{KwData}{Input}
\SetKwInput{KwResult}{Output}

\allowdisplaybreaks 
\graphicspath{{./images/}} 

\begin{document}


\title{Taming Subpacketization without Sacrificing Communication: A Packet Type-based Framework for D2D Coded Caching}

\input{author_tit}

\maketitle

\begin{abstract}
Finite-length design is essential for making coded caching practical, as the optimal communication gains of existing schemes often require prohibitively large subpacketization. This paper studies rate-optimal device-to-device (D2D) coded caching with reduced subpacketization. We propose a packet type-based (PT) framework that exploits the geometric structure induced by user grouping. Under this structure, subfiles, packets, and multicast groups are classified into \tit{types}, allowing the symmetric design proposed by Ji, Caire, and Molisch (JCM) in 2016 to be systematically relaxed without sacrificing the optimal D2D communication rate.
The key feature of the PT framework is that subpacketization reduction is achieved through two complementary mechanisms: \emph{subfile saving}, by excluding redundant subfile types, and \emph{further-splitting saving}, by assigning type-dependent further-splitting factors to subfiles through transmitter selection.
The type-dependent splitting factors are then coordinated across multicast group types to produce a globally consistent file-splitting structure. 
Based on this framework, we construct several classes of rate-optimal D2D coded caching schemes that strictly improve upon the JCM subpacketization. The proposed schemes achieve either order-wise reductions in the number of users or constant-factor reductions over broad memory regimes, while preserving the optimal  rate. These results reveal a structural distinction between D2D and shared-link coded caching: unlike in the shared-link setting, full symmetric subpacketization is not necessary for rate-optimal D2D caching.
\end{abstract}



\section{Introduction}
\label{section:intro}

\subsection{Background: D2D \CoCa}
\label{subsec:D2D cc, intro}

Device-to-device (D2D) coded caching, proposed by Ji, Caire, and Molisch (JCM)~\cite{ji2016fundamental}, extends the Maddah-Ali and Niesen (MAN)~\cite{maddah2014fundamental} shared-link model to D2D networks where users help satisfy each other's demands through multicast transmissions. As in MAN, D2D coded caching consists of a cache placement phase, where users store parts of the files, and a \mtcst delivery phase, where coded messages are exchanged after the demands are revealed.
Let $K,N,M$ denote the number of users, files, and per-user cache size, and define $\mu \eqdef   M/N$ and $t\eqdef  K\mu $. The MAN scheme achieves  \comm rate (\ie, total number of transmitted bits)
$\Rman= \frac{K(1-\mu)}{t+1}$ at \sbp level\footnote{\Sbp level refers to the minimum number of packets that each file is partitioned into, in order to \achv the desired rate.} $\Fman=\binom{K}{t}$ \pkts per file.
The JCM D2D scheme uses the same combinatorial placement as MAN but achieves $\Rjcm=\frac{K(1-\mu)}{t}$
with higher \sbp
\be
\Fjcm=t\binom{K}{t}.
\ee
Under  uncoded placement and one-shot delivery, $\Rjcm$ was shown optimal when $K \le N$\cite{yapar2019optimality}.

Comparing JCM with MAN, the key difference lies in subpacketization. Both schemes use the same $\binom{K}{t}$    subfiles per file during placement, but 
 JCM \tit{further splits} each subfile into $t$ packets, yielding 
$\Fjcm=t\Fman$. 
This additional splitting enables symmetric packet exchange in D2D delivery: all $t+1$ users in each multicast group act as transmitters, so each requested subfile must provide $t$ distinct packets to support the $t$ coded messages received by its requester.
Thus, the optimal\footnote{We restrict our attention to one-shot \schms with uncoded prefetching in this paper.} D2D rate is achieved only after paying a \tit{two-layer} subpacketization cost: $\binom{K}{t}$ placement subfiles and an additional factor $t$ for packet exchange.

\subsection{Finite-Length \CoCa}
\label{subsec: FLA of coca, intro}

The communication gain of coded caching is achieved at the cost of large subpacketization. In both MAN and JCM schemes, the subpacketization level $F$ grows exponentially with  $K$, which creates two finite-length limitations. First, the optimal rate is exactly attainable only when $F\mid L$ ($F$ divides  $L$), where $L$ is the file length in bits. 
If $  F\nmid L$, zero-padding can preserve the scheme structure, but rate optimality is no longer guaranteed. Second, even when $F\mid L$, a large $F$ causes substantial implementation overhead for indexing, storing, and managing subfiles and packets.
Therefore, finite-length coded caching seeks to reduce subpacketization while preserving, or only mildly sacrificing, the communication rate. We next review representative finite-length designs for shared-link and D2D coded caching, with particular emphasis on the latter.

\subsubsection{Shared-Link}
\label{subsubsec: shared link cc fla}
\Finlen  coded caching~\cite{7539576, yan2017pda, 9536664,9284439,9037309,8080217,8651553,shangguan2018hpyergraph, tang2018subpacketization,8335418,8006726,8613527,9477627,9148593,9739200,9348098,8815564,zhang2019cache, 8950279,9839268, 9739199,salehi2020low,  sasi2021multi, brunero2022fundamental, cheng2024asymptotically, cheng2024coded,wei2024novel, rajan2024optimal,nt2025hierarchical, cheng2025new, parrinello2019fundamental, nt2025d2d,huang2026placement}
was first studied by Shanmugam \etal~\cite{7539576} in the decentralized setting, showing that if $F\le \exp(\frac{KM}{N})  $, the rate must scale linearly with $K$.
Yan \etal~\cite{yan2017pda} introduced the placement delivery array (PDA) \frmwk and constructed  \schms for $M/N\in \big\{ \frac{1}{q},  \frac{q-1}{q} \big \}(q\ge2)$, achieving exponentially smaller subpacketization than MAN but with reduced global caching gain.
Within the (linear) PDA framework, MAN was later shown to simultaneously achieve the minimum rate and subpacketization;
hence any $F< \Fman$ necessarily incurs  $R> \Rman$, revealing an inherent rate-subpacketization tradeoff.
Subsequent PDA-based designs further reduced $F$ at the cost of higher rates~\cite{9536664,9284439,9037309,8080217,8651553}.
Other constructions include hypergraph-based schemes with sub-exponential $F= O(\exp(\alpha \sqrt{K}))$\cite{shangguan2018hpyergraph}, 
resolvable designs based on linear block codes\cite{tang2018subpacketization}, 
Ruzsa-Szem$\acute{\rm e}$redi graph \schms achieving linear
$F\propto K$ with sub-linear rate $R=O(K^{\varepsilon}),\varepsilon\in(0,1)$, as well as line-graph~\cite{8613527} and projective geometry-based~\cite{9477627} approaches.
Extensions have also been developed for
multi-access, shared-cache, hierarchical, and multiple-input-multiple-output (MIMO) coded caching~\cite{9328826,sasi2021multi,brunero2022fundamental,yang2026low, 9839268,parrinello2019fundamental, rajan2024optimal,nt2025hierarchical,9148593,9739200,9348098,8815564,8950279,9739199,zhang2019cache, salehi2020low,9369971,lampiris2024adapt, huang2026placement}.

\subsubsection{\dd}
\label{subsubsec: d2d cc fla} 

Finite-length D2D coded caching has been studied in both caching~\cite{wang2017placement,chittoor2019low,8620232,9913463,9477627,woolsey2020d2d,wu2023coded,wang2025coded,nt2025d2d,rashid2026optimal} and coded distributed computing (CDC)~\cite{8437323,9448271,9103948} settings\footnote{The CDC scheme proposed by Li \etal~\cite{li2017fundamental}  is mathematically equivalent to  the \jsch when each Reduce function is assigned to a single  computing node. In this case, the finite-length analysis of CDC  can be directly translated to D2D coded caching, and vice versa.}.
Chittoor \etal\cite{chittoor2019low,9477627}  adapted projective geometry-based shared-link designs to D2D, reducing subpacketization but increasing the rate relative to JCM.
Wang \etal\cite{wang2017placement,8620232}  introduced the D2D placement delivery array (DPDA) framework and transformed shared-link PDAs into DPDAs.
In particular, \cite{wang2017placement}  established lower bounds on $F$ for
$ \frac{M}{N} \in\left\{\frac{1}{K},\frac{2}{K},\frac{K-2}{K},\frac{K-1}{K}\right\}$,
and constructed DPDAs attaining these bounds while preserving the JCM rate, but only at these limited memory points.
Other DPDA constructions~\cite{8620232,9913463,nt2025d2d,wang2025coded} achieve exponentially smaller $F$ than $\Fjcm$, but either increase the rate or apply only to a very  restricted set of \memo points.
Woolsey \etal\cite{woolsey2020d2d, 9448271,9492815 } proposed a hypercube-based scheme with a geometric interpretation where cached \pkts correspond to hyperplanes and coded messages arise from their intersections. 
For $\frac{M}{N} = \frac{1}{q} (q\in \mathbb{N}_{>0})$, it \achvs
$R=\frac{K(1-M/N)}{t-1}>R_{\rm JCM}$ and $F\in \{q^t, (t-1)q^t\}$. Implementation on Amazon Web Service (AWS) clusters demonstrates the  \schm's efficiency.
Related hypercube~\cite{9448271} and resolvable-design\cite{9103948} based constructions were also developed for CDC, reducing input files or Map functions but increasing the communication rate relative to the JCM-equivalent CDC baseline.

\subsection{Main Idea of this Paper}
\label{subsec:main idea}

The above discussion shows that existing finite-length D2D coded caching schemes remain \tit{fragmented}: most constructions either reduce subpacketization at the cost of higher communication rate, or preserve the JCM rate only at {isolated} memory points. Hence, despite its large subpacketization, the JCM scheme~\cite{ji2016fundamental} remains the benchmark state-of-the-art for rate-optimal D2D coded caching, while also leaving substantial room for systematic improvement.
This leads  to the first question:

\tbf{Question 1}:
\tit{Can the subpacketization of D2D coded caching be reduced while {i)} preserving the optimal communication rate, and  {ii)} maintaining broad applicability across system parameters?}

To identify room for improvement, we revisit the JCM scheme from first principles. 
In particular, the JCM scheme attains the optimal communication rate through a highly \tit{symmetric} design: 
\tit{every} $t$-subset of users labels a subfile, \tit{every} subfile is further split into $t$ packets, and \tit{every} user in each $(t+1)$-user multicast group $\Sc$ acts as a transmitter. This symmetry is sufficient for rate optimality, but not always necessary as  we will see  in Section \ref{section: PT design framework}.  
We  then ask

\tbf{Question 2}:
\tit{Which structural symmetries of the JCM scheme\cite{ji2016fundamental}, in both subpacketization and multicast transmission, are necessary for rate optimality, and which can be relaxed to enable \sbp reduction?}

The  central idea of this paper is  to expose and relax the redundant parts of JCM symmetry without violating the optimal communication-rate condition (see Fig. \ref{fig:two gains demo,intro}).
\begin{figure}[t]
    \centering
    \includegraphics[width=0.55\linewidth]{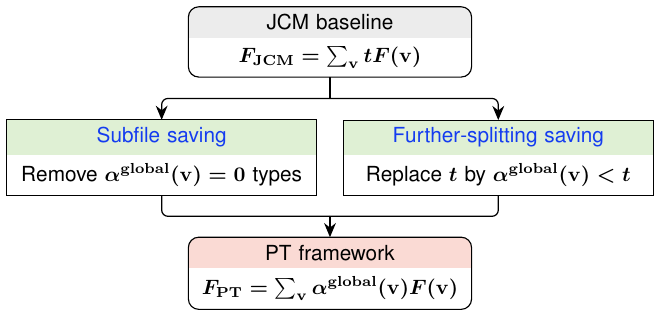}
    \caption{\small PT \frmwk reduces JCM \sbp through two reduction mechanisms, \ie, subfile saving and further-splitting saving.}
    \label{fig:two gains demo,intro}
\end{figure}
We do so by partitioning users into groups and assigning each subfile and multicast group a \tit{type} according to its group-occupancy profile, i.e., the sorted vector that counts how many of its associated users belong to each user group.
Once types are introduced, subfiles and multicast transmissions no longer need to be treated uniformly---some subfile types may be excluded, and the retained types can be assigned type-dependent further-splitting (FS) factors\footnote{The further-splitting factor of a subfile type is the number of packets into which each subfile of that type is divided  to accommodate multicast delivery.} through transmitter selection in each \mgrp  $\Sc$ (\ie, selecting a \tit{subset} of users in $\Sc$  to  serve as \txs).
This gives a compact interpretation of the subpacketization reduction. \Ip,  the \sbp in the proposed  packet type (PT)-based \frmwk is equal to 
\be 
\label{eq:Fpt,intro,main idea}
\Fpt= \sum_{\vv\in\Vc} \alphaglobalnb(\vv)F(\vv),
\ee 
where $F(\vv),  \vv  \in \Vc$  denotes the number of  type-$\vv$ \sbfs and $\alphaglobalnb(\vv)\le t$ is the  FS factor  assigned to type-$\vv$ \sbfs. As a special case, the JCM \sbp can be written as
\be
\label{eq:Fjcm,intro,main idea}
\Fjcm \eqdef   t\binom{K}{t} =\sum_{\vv \in\Vc_{\rm JCM}}tF(\vv)
\ee 
with $\alpha_{\rm JCM}=t$ for all potential types. 
Hence, PT reduces \sbp by \tit{choosing  $\alphaglobalnb(\vv)=0$ for redundant (excluded) \sbf types and by assigning $ \alphaglobalnb(\vv)<t $ to retained types} whenever possible.

The main technical challenge is that we need to coordinate the local FS requirements induced by transmitter selections within different types of multicast groups. 
\Specly, a subfile type $\vv$ may participate in several multicast group types, whose transmitter selections can impose different local FS factors. 
However, all subfiles of the same type must use a common final FS factor, regardless of the multicast groups in which they are delivered. These local requirements must therefore be reconciled across multicast group types. To this end, we introduce a \tit{vector least common multiple} (vLCM) operation, which combines compatible \tit{local} FS vectors into a \tit{global} FS vector
\be 
\alphaglobal \eqdef\big(\alphaglobalnb(\vv)\big)_{\vv\in \Vc} \in 
\mathbb{Z}_{\ge 0}^{1\times |\Vc|}
\ee 
where  $\Vc$ is the set of subfile types and $\alphaglobalnb(\vv)$ is the globally compatible FS factor for type-$\vv$ \sbfs.
Together with the corresponding transmitter selections, any global FS vector satisfying the vLCM compatibility and cache-memory constraints yields a rate-optimal D2D caching scheme whose \sbp is given by \eqref{eq:Fpt,intro,main idea}.

In  short, the PT framework converts subpacketization reduction into a {geometrically structured} combinatorial design problem: choose a user grouping and, for each induced multicast group type, choose transmitters so as to minimize the resulting global FS vector. The final scheme is compactly represented by $\alphaglobal$, whose entries specify which subfile types are excluded and how finely each retained type is split. The JCM scheme is recovered as the fully symmetric special case where all subfile types are retained and $\alphaglobalnb(\vv) =t$ for every possible type $\vv$.
\Aar, 
the PT framework provides affirmative answers to both  {Questions 1} and {2}.
By exploiting the geometric type structure induced by user grouping and transmitter selection, PT improves D2D \sbp over broad memory ranges and goes beyond existing fragmented designs. 
It  offers a general design methodology for reducing subpacketization while preserving the optimal D2D communication rate. 
The key differences between PT and JCM are briefly summarized in Table \ref{tab:central idea of pt,intro}.

\begin{table}[h]
    \centering
    \caption{Comparison of PT and JCM in terms of \sbp.}
    \begin{tabular}{|c|c|}
    \hline
      \tbf{JCM symmetry}   & \tbf{PT relaxation}  \\
      \hline
      All subfile types   & Redundant types excluded   \\
      \hline 
      Uniform FS factors & Type-dependent FS factors\\
    \hline
    Full Tx participation in $\Sc$ & Selective Tx participation\\
    \hline
    \end{tabular}
    \label{tab:central idea of pt,intro}
\end{table}

\subsection{Summary of Contributions}
\label{subsec:summary of contribution,intro}

\begin{itemize}
\item {We reveal a fundamental  distinction in \sbp between D2D and shared-link coded caching.}
In shared-link coded caching, prior results show that achieving the optimal MAN rate requires subpacketization no smaller than $\binom{K}{t}$ under the PDA framework~\cite{yan2017pda}. In contrast, we show that this limitation does not carry over to D2D coded caching: the optimal JCM  rate can be preserved with subpacketization strictly smaller than $F_{\rm JCM}=t\binom{K}{t}$. This demonstrates that the full symmetry of the JCM scheme is sufficient but not necessary for rate optimality, and exposes a new degree of freedom unique to the D2D setting.

\item  {We propose a PT-based framework that turns subpacketization reduction into a structured combinatorial design problem.}
The framework partitions users into groups, which induces geometric types for subfiles, packets, and multicast groups. Given this type structure, the design reduces to choosing the user grouping and the transmitter set for each multicast group type. These choices determine local further-splitting (FS) requirements, which are then coordinated through a vector least common multiple operation to produce a globally valid FS vector $\alphaglobal$.  
Each feasible $\alphaglobal$ specifies a valid rate-optimal D2D caching scheme, with \sbp $\Fpt= \sum_{\vv} \alphaglobalnb(\vv)F(\vv)$. The JCM scheme is recovered as  a special case where all subfile types are retained and $\alpha^{\rm global}(v)=t,\forall \vv$.

\item {We identify two operational subpacketization-reduction mechanisms (see Fig. \ref{fig:two gains demo,intro}).}
The first is the \tit{subfile saving gain}, obtained by excluding redundant subfile types, equivalently setting $\alpha^{\rm global}(\vv)=0$. This reduces the number of subfiles that need to be created and delivered. The second is the \tit{further-splitting saving gain}, obtained by assigning smaller-than-$t$ FS factors to retained subfile types through transmitter selection. Thus, PT reduces subpacketization not by sacrificing multicast gain, but by removing unnecessary subfile types and reducing the packet-per-subfile ratio whenever the type structure allows.

\item {We construct multiple classes of rate-optimal D2D coded caching schemes with reduced subpacketization.}
Based on the PT framework, we obtain several explicit scheme families that preserve the optimal JCM communication rate while achieving either order-wise or constant-factor subpacketization reductions. These schemes apply across broad memory regimes and strictly improve the rate-subpacketization tradeoff of existing finite-length D2D designs, many of which either increase the  rate or apply only to restricted memory points. To the best of our knowledge, the PT framework is the first general design methodology for reducing D2D subpacketization while preserving the optimal  rate.
\end{itemize}

\subsection{Paper Organization}
The remainder  of this paper is organized as follows.
Section~\ref{section: problem description} introduces the  problem setup and provides a brief overview of the proposed PT framework.
The main results, summarized in three theorems, are presented in Section~\ref{section: main results}. 
Section~\ref{section: PT design framework} presents the full PT design \frmwk.
The general caching schemes corresponding to the  three  main theorems are presented in Sections~\ref{section: gen ach schm, thm 1}, \ref{section: gen ach schm, thm 2} and \ref{section: gen ach schm, thm 3}, \resp, each following a common PT design-example-discussion structure.
Finally, Section~\ref{section: conclusion} concludes the paper.

\subsection{Notation Convention}
For integers $m\le n$, let 
$[m:n]\eqdef\{m,m+1,\cdots,n\}$ and write
$[1:n]$ as $[n]$.
$\mathbbm{1}\{\cdot\}$  denotes  the indicator function. 
Calligraphic letters $\Ac,\Bc,\cdots$ denote sets and bold  letters $\Am,\Bm,\av,\bv, \cdots $ denote matrices or vectors.
For sets $\Ac$ and $\Bc$, define
$\Ac\backslash \Bc\eqdef \{x\in \Ac:x\notin \Bc\}$, 
$\Ac\times \Bc \eqdef \{ \{a,b\}: a\in \Ac, b\in \Bc  \} $, and 
$\binom{\Ac}{n}\eqdef \{ \Sc \subseteq \Ac: |\Sc| = n   \}$.
For matrix $\Am$, 
$\Am_{i,:}$ and $\Am_{:,j }$ denote its $\ith$ row and $\jth$ column,  \resp.
$\underline{a}_{n}\eqdef (a,\cdots,a)$ denotes the length-$n$ row vector with identical entry $a$. $\oplus$ denotes bit-wise XOR. 
We write  $m\mid n$ if $m$ divides $n$, and $m \nmid n$ \otws. 
The standard  asymptotic notation $O$ and $\Theta$ is used.

\section{Problem Statement}
\label{section: problem description}
We first review the JCM D2D coded caching scheme~\cite{ji2016fundamental}, emphasizing the structural properties of its subpacketization and coded multicast delivery. We then formalize symmetric and asymmetric subpacketization and packet exchange, which form the basis of the PT framework. Finally, we present an illustrative example that highlights the key ideas behind the PT design.

\subsection{Two-Layer \Sbp of JCM}
\label{subsec: subpacketization structure of D2D}

In  the JCM D2D coded caching~\cite{ji2016fundamental}, there are $N$  files $W_1,\cdots,W_N$ each having $L$ bits,  and $K$ users, each with cache size $|Z_k|=ML$ bits.  
Let $t \eqdef \frac{KM}{N} \in \mathbb{Z}_{>0}  $. 
Each file is split into 
\be 
\label{eq:Fjcm def, problem statement}
\Fjcm\eqdef t\binom{K}{t}
\ee 
equal-size packets
\be 
\label{eq: JCM two-layer subpacketization structure}
W_n = \left\{W_{n,\Tc}^{(i)}:\Tc \in \binom{[K]}{t},i\in [t]\right\}, n \in [N]
\ee 
This subpacketization has \tit{two} layers.
First, each file is split into $\binom{K}{t}$ \tit{\sbfs}, 
\be
\label{eq: JCM first layer subpacketization}
\mrm{[File\;\, to\;\,subfiles]} \quad 
W_n=\lef\{W_{n, \Tc}:\Tc\in \binom{[K]}{t}\rig\}, 
\ee 
where the \sbf $W_{n, \Tc}$ is cached exactly  by the users in $\Tc$. Second, each subfile is further split into $t$ smaller \tit{packets},
\be 
\label{eq: JCM second layer subpacketization}
\mrm{[Subfile\;\, to\;\,packets]} \quad 
W_{n,\Tc}=\lef\{W_{n,\Tc}^{(i)}:i\in [t]\rig\},\; \forall \Tc \in \binom{[K]}{t}
\ee

The first layer realizes the combinatorial cache placement, while the second layer enables symmetric packet exchange in the \mtcst  delivery phase.
Specifically, in the placement phase, \user{k} stores
all subfiles with $ \Tc$ where $k\in \Tc$, \ie, 
$ 
Z_k=\big\{W_{n,\Tc}^{(i)}:\forall \Tc\ni k, i\in [t], n \in[N] \big\}
$. 
In the \deli phase, for each \mgrp  $\Sc \in \binom{[K]}{t+1}$, 
every user $k \in \Sc$ transmits one coded message to the remaining $t$ users:
\be 
\label{eq: JCM coded message, JCM delivery}
Y_{\Sc\backslash\{k\}}^k \eqdef \bigoplus_{k^{\prime}\in \Sc\backslash\{k\}}W_{d_{k^{\prime}}, \Sc\backslash \{k^{\prime}\}}^{(c_{k^{\prime}})}, 
\ee 
where the counters $\{c_k\}_{k \in \Sc}$  ensure that different packets of the same requested subfile are used across transmissions.
Each coded message is useful to exactly $t$ users, satisfying the optimal degrees-of-freedom (DoF) condition (see \Rmk \ref{rmk:optimal dof}) and yielding the JCM rate
\be 
\label{eq: JCM rate}
\Rjcm  \eqdef\frac{K(1-M/N)}{t}.
\ee 
Thus, JCM achieves rate optimality (or equivalently, the optimal \msg DoF) through a highly symmetric structure: every $t$-subset labels a subfile, every subfile is split into $t$ packets, and every user in every $(t+1)$-user multicast group transmits. This symmetry is sufficient for the optimal rate, but the PT framework shows that it is not always necessary.

\begin{remark}[Optimal DoF Condition]
\label{rmk:optimal dof} 
\tit{The optimal DoF condition means that each transmitted \msg in \eqref{eq: JCM coded message, JCM delivery} is  simultaneously useful to $t$ other users, which ensures the JCM rate \eqref{eq: JCM rate}. This condition is also enforced in the proposed PT \frmwk to guarantee rate optimality.}
\end{remark}

\subsection{Symmetric vs. Asymmetric \Sbp and Packet Exchange}
\label{subsec: defs, symmetric and asymmetiric subp, packet exchange, problem statement}
A generic two-layer \sbp structure for D2D \coca can be written as
\begin{align}
\label{eq:generic two-layer sbp, problem statement}
W_n &= \left  \{ W_{n, \Tc}: \Tc  \in \mathfrak{T}    \right\},\,\mathfrak{T} \subseteq \binom{[K]}{t},\notag\\
W_{n,\Tc} &= \left \{  W_{n, \Tc}^{(i)}: i \in  [\alpha_{\Tc}]   \right\},\; \alpha_\Tc\in  \mathbb{Z}_{\ge 0}. 
\end{align}
Here,   $\mathfrak{T} $ specifies the retained subfile labels among all $t$-subsets of users, and $\alpha_{\Tc}$ is the FS factor assigned to subfiles with label $\Tc$, which may vary with $\Tc$. Note that all retained subfiles have $\alpha_{\Tc} >0$, while  $\alpha_{\Tc}=0$ represents exclusion.

\begin{defn}[\Sym \Sbp]
\label{def: symmetric subpacketozation}
Under the generic structure \eqref{eq:generic two-layer sbp, problem statement}, the JCM scheme corresponds to the {symmetric}  choice: $\mathfrak{T}= \binom{[K]}{t}  $ and $\alpha_\Tc = t, \forall  \Tc  \in \mathfrak{T}$. Such a choice is referred to  as
\tit{\sym \sbp}, characterized by {1) \tit{Exhaustive coverage of  $\Tc$:}} $\mathfrak{T}=\binom{[K]}{t}$, and {2) \tit{Uniform subfile  splitting:}} $\alpha_\Tc=t$ for all $\Tc \in \binom{[K]}{t}$.
\end{defn}

 Besides \sym \sbp, the JCM \schm also uses  \sym \pkt  \exch during \mtcst delivery.

 \begin{defn}[\Sym Packet Exchange]
\label{def: symmetric data exchange} 
For each \mgrp $\Sc \in \binom{[K]}{t+1} $, each \user{k\in \Sc} broadcasts one  \cmsg to the remaining $t$ users, formed as the XOR of $t$ \pkts intended for users in $\Sc\bksl \{k\} $. 
Each receiver decodes one desired packet from each such message, so every user in $\Sc$ obtains $t$ packets forming its desired subfile
$W_{d_k, \Sc \bksl \{k\}}=\left( W_{d_k, \Sc \bksl \{k\}}    \right)_{i=1}^t$.
Equivalently, symmetric packet exchange has three properties:
\begin{itemize}
    \item[1)] \tit{Transmission balance:} all users in each multicast group transmit the same amount of data;
    \item[2)] \tit{Exhaustive coverage:} \mtcst delivery is performed over all $\Sc \in \binom{[K]}{t+1}$;
    \item[3)] \tit{Uniform \txn strategy:} the same transmitter selection and \msg  generation rule is used for every multicast group. 
\end{itemize}

\end{defn}

The following example illustrates why JCM needs the second layer of \sbp to 
\achv $\Rjcm$.

\begin{example}[JCM Scheme,\cite{ji2016fundamental}]
\label{example: necessity of FS}
Consider $K=3$ and $t=2$.  
\begin{figure}[t]
\centering
\begin{subfigure}[b]{0.37\linewidth}
\centering
\includegraphics[width=\linewidth]{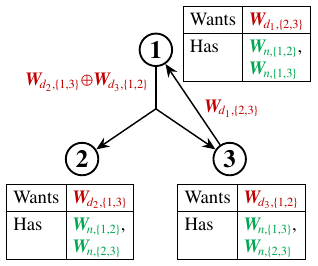}
\caption{\small Delivery without further splitting, $R=2/3$.}
\label{fig: fig_FS_necessity_1}
\end{subfigure}
\begin{subfigure}[b]{0.44\linewidth}
\centering
\includegraphics[width=\linewidth]{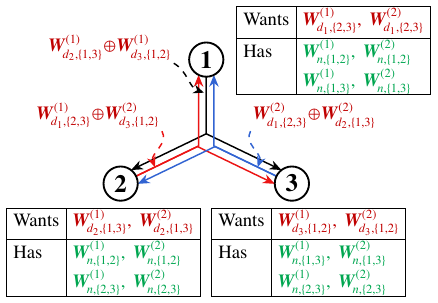}
\caption{\small Delivery with further splitting, $R=1/2$.}
\label{fig: fig_FS_necessity_2}
\end{subfigure}
\caption{\small Necessity of further splitting of subfiles into packets. The left  figure \achvs a rate of  $2/3$ with $F=3$ using only subfiles, whereas the right figure (\jsch) \achvs a smaller rate $1/2$ with $F=6$.} 
\label{fig: fig_FS_necessity}
\end{figure}
Each file is first split into 3 subfiles $W_{n}=\{W_{n,\{1,2\}},W_{n,\{1,3\}}  ,W_{n,\{2,3\}}\},\forall n$, and there is only  one \mgrp $\Sc=\{1,2,3\}$.
If \deli is performed directly at the subfile level, one message must be useful to only one user, yielding rate $R=2/3$ as in Fig.~\ref{fig: fig_FS_necessity_1}. 
By contrast, JCM further splits each subfile into two packets,
$W_{n,\Tc}=\big(W_{n,\Tc}^{(1)},W_{n,\Tc}^{(2)}\big)$, and sends
\begin{align}
Y^1_{\{2,3\}}  &= W_{d_2,\{1,3\}}^{(1)}  \oplus W_{d_3,\{1,2\}}^{(1)},\notag \\
Y^2_{\{1,3\}} &= W_{d_1,\{2,3\}}^{(1)} \oplus W_{d_3,\{1,2\}}^{(2)},\notag \\
Y^3_{\{1,2\}} & = W_{d_1,\{2,3\}}^{(2)} \oplus W_{d_2,\{1,3\}}^{(2)}.
\end{align}
Each \msg is useful to exactly two users, so the optimal DoF condition is satisfied and the rate becomes $\Rjcm=1/2$ (see Fig. \ref{fig: fig_FS_necessity_2}). Thus, the second subpacketization layer in JCM is needed to enable symmetric packet exchange and preserve the optimal rate.
\hfill $\lozenge$
\end{example}

\Ex \ref{example: necessity of FS} shows that symmetric subpacketization and  packet exchange are sufficient to achieve the optimal rate. However, they are not always necessary. The key requirement is  \tit{not symmetry itself, but that each transmitted bit remains simultaneously useful to $t$ users}. This observation motivates a more general design space in which both file splitting and packet exchange may be asymmetric while the optimal DoF condition is still preserved.

\begin{defn}[\Asym \Sbp]
\label{def: asymmetric subpacketization} 
\Asym \sbp allows the subfile labels to cover only a subset $\mathfrak{T} \subseteq \binom{[K]}{t}$ (non-exhaustive coverage),  thereby excluding some JCM subfiles, and allows the FS factor  $\alpha_\Tc$ to vary across subfiles (nonuniform further splitting).
\end{defn}

\begin{defn}[\Asym \Pkt Exchange]
\label{def: asymmetric data exchange}
Similarly, asymmetric packet exchange allows the delivery phase to deviate from the JCM rule in three ways:
\begin{itemize}
    \item[1)]\tit{Transmission imbalance:} users in the same multicast group may transmit different amounts of data;
    \item[2)]\tit{Non-exhaustive coverage:} only a subset $\mathfrak{S} \subseteq \binom{[K]}{t+1}  $ of multicast groups may be served;
    \item[3)]\tit{Nonuniform transmission strategy:} different multicast groups may use different transmitter selection and coding strategies.
\end{itemize}
\end{defn}

The JCM \sbp and \deli scheme can be  recovered from  Definitions  \ref{def: asymmetric subpacketization} and \ref{def: asymmetric data exchange} by choosing $ \mathfrak{T}=\binom{[K]}{t}$, $\alpha_\Tc=t ,\forall  \Tc \in\binom{[K]}{t}      $, $\mathfrak{S}=\binom{[K]}{t+1}$, and letting every user in $\Sc \in \binom{[K]}{t+1} $ transmit according to the same rule.

\subsection{PT Design: A Quick Example}
\label{subsec: a first look at PT, problem description}

We next present an illustrative example that showcases the key ideas behind the PT design. The formal description is deferred to Section~\ref{section: PT design framework}.

\begin{example}
\label{example: (4,2,1) PT first look}
Consider $(K,N,M)=(4,2,1)$, so that $t=2$. Let the two files be $A$ and $B$, and partition the users into two groups
$\Qc_1=\{1,2\},\Qc_2=\{3,4\}$. Instead of using all $\binom{4}{2}=6$ JCM  labels,  we keep only the cross-group labels 
$ \mathfrak{T}=\{\{1,3  \},\{1,4   \},\{2,3   \},\{2,4 \} \}$. Thus each file is split into  $4$ \sbfs, 
\begin{align}
\label{eq: file splitting, example (4,2,1) PT first look}
 A & =\left\{A_{\{1,3\}}, A_{\{1,4\}},A_{\{2,3\}},A_{\{2,4\}}\right\} ,\notag\\
 B &=\left\{B_{\{1,3\}}, B_{\{1,4\}},B_{\{2,3\}},B_{\{2,4\}}\right\},
\end{align}
and no further splitting is needed.
\User{k} caches all subfiles whose label contains $k$, for both files. 
Suppose the demand vector is $\dv=(1,1,1,2)$. 
The delivery \schm is given in Table~\ref{table: coded msg, example 1, PT first look}: 
in each multicast group $\Sc=\{k\}\cup \Rc_k$, only user $k$ transmits to the two users in the other group.
\begin{table*}[t]
\caption{Delivery Scheme of Example~\ref{example: (4,2,1) PT first look}}
\centering
\renewcommand{\arraystretch}{0.8}  
\small \begin{tabular}{|c"c|c|c|} 
\hline 
\textrm{Tx $k$} & \textrm{Rxs} $\Rc_k$  & $Y^k_{\Rc_k}$   &  $\Sc\eqdef \{k\} \cup \Rc_k   $  \\
 \thickhline
$1$ & $\{3,4\}$  & $A_{\{1,4\}} \oplus B_{\{1,3\}}$  & $\{1,3,4\}$ \\
 \hline
 $2$ & $\{3,4\}$  & $A_{\{2,4\}} \oplus B_{\{2,3\}}$ & $\{2,3,4\}$ \\
 \hline
$ 3$ & $\{1,2\}$  &$A_{\{1,3\}} \oplus A_{\{2,3\}}$   & $\{1,2,3\}$\\
 \hline
$ 4$ & $\{1,2\}$  & $A_{\{1,4\}} \oplus A_{\{2,4\}}$  & $\{1,2,4\}$ \\
 \hline
\end{tabular}
\label{table: coded msg, example 1, PT first look}
\end{table*}
Each user can recover its desired file. For example, \user{1} has cached 
$A_{\{1,3\}}$ and $A_{\{1,4\}}$, and decodes $A_{\{2,3\}}$ and $A_{\{2,4\}}$ from the transmissions of users  $3$ and $4$, \resp. Hence, \user{1} recovers file $A$. 
The \achvd rate is $R=1= \Rjcm$, while the \sbp is $\Fpt=4$, which is one third of the JCM \sbp $\Fjcm=2\binom{4}{2}=12$. 
\hfill $\lozenge$
\end{example}

\Ex \ref{example: (4,2,1) PT first look} highlights the two PT reduction mechanisms in their simplest form:
type-$(2,0)$ subfiles\footnote{A type-$(i,j), j \le i \le 2   $  \sbf $\Ac_{\Tc}$ or $\Bc_\Tc$ is one with $\{|\Tc \cap \{1,2\}|, |\Tc \cap \{3,4\}| \}=\{i,j\}$.}
$ \Ac_\Tc, \Bc_\Tc,  \Tc \in \{ \{1,2\},\{3,4\}\}$ 
are excluded, while retained type-$(1,1)$ subfiles $ \Ac_\Tc, \Bc_\Tc,  \Tc \in  \mathfrak{T}  $ require no further splitting.
Thus, the same optimal DoF condition is preserved with fewer subfiles and smaller FS factors.\footnote{Note that the same \sbp  can be also \achvd using the design for $M/N=1/K(t=2)$ from \cite{wang2017placement}. However, we use this simple example to highlight the type-based design  philosophy of the proposed PT \frmwk.}

\section{Main Results}
\label{section: main results}

The main results are given in three theorems, each providing a rate-optimal PT design with reduced subpacketization relative to the JCM scheme. 
Let  $ t \eqdef \frac{KM}{N}$, $\bar{t}\eqdef K-t=K(1-\mu)$, and $\mu \eqdef \frac{M}{N}$. 
\Thm~\ref{thm: order reduction}  achieves an order-wise reduction when $K$ and 
$ \tbar$ are even and 
$ \tbar \le K/2$, i.e., in the large-memory regime $\mu \ge 1/2$. 
\Thm~\ref{thm: constant reduction} guarantees at least a half reduction for all memory regimes when $K$ and $t$ are even. 
\Thm \ref{thm: PT design with only subfile saving gain} further provides a constant-factor reduction for a broad class of parameters $K=mq$ and 
$ \mu \le \min \big\{ 1/q,  {q}/{K} \big\}-{1}/{K} $ in the small-memory regime. 
Table~\ref{tab:summary of results} compares the applicable regimes, the resulting bounds on $\Fpt/\Fjcm$, and the corresponding reduction orders.

\begin{table}[t]
\centering
\caption{Comparison of Theorems}
\begin{tabular}{|c"c|c|c|}
\hline
 & \tbf{\Thm 1}    & \tbf{\Thm 2}   & \tbf{\Thm 3}   \\
 \hline  
 $K$    & even  &  even &  even \emph{or} odd    \\
 \hline
$\bar{t}\, (t)$ & even $(\bar{t} \le K/2)$ & even & even \emph{or} odd  \\ 
\hline 
Memory size  & large, $\mu \ge 1/2$ & full range, $0<\mu <1$ & small, $\mu \le \min \big\{   \frac{1}{q},  \frac{q}{K} \big\}-\frac{1}{K}  $  \\
\hline
${\Fpt}/{\Fjcm}$   & $\le \min\big\{ \frac{1}{t}\prod_{i=1}^{\bar{t}/2}(2i-1),1 \big\}     $   & $\le  \frac{1}{2}\left(1-\frac{1}{2^{t-1}} \right)  $ & $ \le    1-
\frac{ m\prod_{i=0}^{t-1}\left(q-i\right)! }
{\prod_{i=0}^{t-1}(K-i)}$  \\ 
\hline 
 Gain  &  order-wise, $ \Theta(1/K)$   & constant-factor, $ \le 1/2$ & constant-factor, $<1$  \\
\hline
\end{tabular}
\label{tab:summary of results}
\end{table}


\subsection{Main \Thms}
\label{subsec:main thms, main results}

\begin{theorem}[Order-Wise Reduction]
\label{thm: order reduction}
Suppose that
$K$ and $\overline{t}$ are even integers with $K\ge 4$ and $\bar{t}\le K/2$.
Then the D2D coded caching rate $R=N/M-1$  is achievable with \sbp $\Fpt$ satisfying
\be 
\label{eq: thm 1}
\frac{\Fpt}{\Fjcm} \le \min \left\{\frac{1}{t}\prod_{i=1}^{\overline{t}/2}(2i-1),1\right\}.
\ee
Moreover, the ratio ${F_{\rm PT}}/{F_{\rm JCM}}$
vanishes as $K\to\infty$ if $\mu = 1-O\big({\log_2\log_2 K}/{K}\big)$.
\end{theorem}

The proof of \Thm \ref{thm: order reduction} is given in Section \ref{section: gen ach schm, thm 1}. A comparison of the actual ratio $\Fpt/\Fjcm$ with the upper bound $ \frac{1}{t}\prod_{i=1}^{\overline{t}/2}(2i-1)  $ for various $\tbar$ is shown in Fig.~\ref{fig: subpacketization ratio, thm 1}.  
It  can be seen that the actual ratios and the derived upper bounds match closely.
We \highlt the implications of \Thm \ref{thm: order reduction} \af. 
\begin{figure}[ht]
    \centering
    \includegraphics[width = 0.52\linewidth]{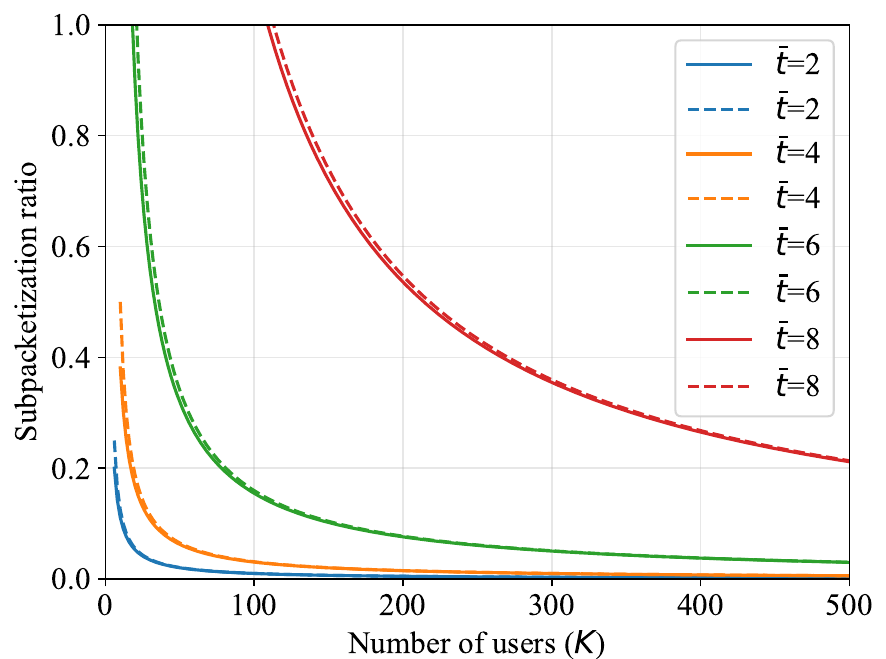}
    \caption{\small Actual \sbp ratios (solid lines) versus  upper bounds (dotted lines) for $\tbar\in \{2,4,6,8\}$.}
    \label{fig: subpacketization ratio, thm 1}
\end{figure}

\begin{enumerate}
    \item \tit{Order-wise reduction in large-\memo regime:} 
    For fixed $\tbar=K(1-\mu)$, \Thm \ref{thm: order reduction}  gives $ \Fpt/\Fjcm= \Theta(1/K)$, and thus the \sbp ratio vanishes as $K$ grows. Since fixed $\tbar$ implies $\mu =1- \Theta(1/K)$, this order-wise gain occurs in the large-memory regime $\mu \in  [1/2, 1]$. 

    \item  \tit{Mechanism:} The construction partitions users into $K/2$ groups of size two, \ie, $\qv=(\underline{2}_{K/2})$,  and selects only a subset of structurally distinct users as transmitters. \Ip, in the $\ith$ \mgrp type
    \be
    \label{eq: multicast group, implication of thm 1}
    \sv_i=\big( \underline{2}_{K/2-(\overline{t}/2+i)+1}, \underline{1}^{\dagger}_{2i-1},\underline{0}_{\overline{t}/2-i}\big),\; i\in \big[\overline{t}/2\big],
    \ee
    the $2i-1$ users at positions marked by dagger $\dagger$ are chosen as \txs. Since the local FS factor for each \invod subfile type in $\sv_i$ is proportional to the number of \txs therein (see (\ref{eq: FS factor})), it is bounded by  $2i-1\le  \tbar-1$. Hence, for fixed $\tbar$,  the FS factors remain \tit{constant} (\indept of $K$). In contrast, the \jsch has a uniform FS factor $\alphajcm=t=\Theta(K)$ linear in $K$. Since  the overall \sbp is determined by the \agg contribution of  \sbfs and  their FS factors (see \eqref{eq:Fpt,intro,main idea}),  the \asym  \tx selection in \eqref{eq: multicast group, implication of thm 1} results in an order-wise reduction compared to JCM. 
\end{enumerate}

\begin{theorem}[More-Than-Half Reduction]
\label{thm: constant reduction}
Suppose that both $K$ and $t$ are even.  The D2D coded caching rate $R =N/M -1$  is achievable with \sbp $ \Fpt$ satisfying  
\be 
\label{eq: thm 2}
\Fpt   \le  \frac{1}{2} \cdot \Fjcm  .  
\ee 
\end{theorem}

The proof of Theorem~\ref{thm: constant reduction}  is given  in Section~\ref{section: gen ach schm, thm 2}. We \highlt its implications \af.

\begin{enumerate}
    \item \tit{More-than-half reduction:}
    For even $K$ and $t$, the two-group user grouping $\qv=(K/2, K/2)$ achieves a global FS vector whose largest entry is at most $t/2$: 
    \begin{eqnarray}
    \label{eq: global FS vector, thm 2 implication}
    \alphaglobal  =\left\{\begin{array}{ll}
    \lef(0,1,2,\cdots,t/2  \rig),& \textrm{if}\; t\le K/2   -1\\
    \lef(t- K/2   ,t- K/2 +1,\cdots, t/2   \rig), &\textrm{if} \; t\ge  K/2 
    \end{array}\right.
    \end{eqnarray}
    Therefore, the total PT \sbp is at most half of the JCM \sbp.  
    When $t\le K/2-1$, the first subfile type is also excluded, yielding additional subfile saving.

    \item \tit{Complementarity with \Thm \ref{thm: order reduction}:} 
    Theorem \ref{thm: order reduction} gives an order-wise reduction in the \tit{large}-memory regime under parity conditions on $K$ and $\tbar$.  Theorem \ref{thm: constant reduction} instead guarantees a constant-factor reduction over the full memory range $0 < \mu < 1$ whenever $K$ and $t$ are even. Thus, different PT groupings and transmitter selections are useful in different parameter regimes.  
\end{enumerate}

\begin{theorem}[Constant-Factor Reduction]
\label{thm: PT design with only subfile saving gain}
Suppose $K=mq$, where $m \ge t+1$, $ q\ge t+1$ and $t\ge 2$. The D2D coded caching  rate $R=N/M-1$   is \achvb with \sbp $\Fpt$ satisfying
\be
\label{eq:thm,PT design with only subfile saving gain}
\frac{F_{\rm PT}}{F_{\rm JCM}}
 = 1-
 \frac{ m\prod_{i=0}^{t-1}\left(q-i\right)! }
 {\prod_{i=0}^{t-1}(K-i)}<1, \quad
 \lim_{q\to \infty} \frac{F_{\rm PT}}{F_{\rm JCM}} = 1 -  \frac{1}{m^{t-1}}.
\ee
\end{theorem}

The proof of \Thm  \ref{thm: PT design with only subfile saving gain} is given in  Section \ref{section: gen ach schm, thm 3}. We highlight its implications \af.

\begin{enumerate}
    \item \tit{Generic constant-factor reduction via subfile exclusion:}
    For $K = mq$ with $m, q \ge t+1$, the PT design excludes type-$(t,0^{m-1})$ subfiles and keeps the JCM FS factor $t$ for all remaining types. Hence, the gain comes solely from subfile saving. The resulting \sbp ratio is strictly smaller than one, and approaches $1 - 1/m^{t-1}$ as $q$ grows.

    \item \tit{Role among the three theorems:}
    \Thm \ref{thm: PT design with only subfile saving gain} gives a weaker but more broadly applicable reduction. Unlike Theorems \ref{thm: order reduction} and \ref{thm: constant reduction}, it does not require both $K$ and $t$ to be even, and therefore complements them in parameter regimes where the stronger PT designs do not apply.
\end{enumerate}

\section{Packet Type-based Design Framework}
\label{section: PT design framework}

We now formalize the PT \frmwk. Unlike Section \ref{subsec: a first look at PT, problem description}, which illustrates the idea through a small example,
this section specifies the design workflow with  step-by-step exposition of its key components.
The resulting cache placement and delivery phases are also specified. 
For readability, the notation used in this section is summarized in Table~\ref{tab:notations}.
\begin{table}[ht]
  \centering
  \caption{Summary of Notation}
  \label{tab:notations}
  \renewcommand{\arraystretch}{0.78}
  \begin{tabular}{| c | c " c | p{5.1cm}|}
    \thickhline
    Notation & Meaning  &     Notation & Meaning  \\
    \thickhline
    $K$ & $\#$ users  &     $\Ic  $ & \invod  \sbf  type set  \\
    \hline
    $N$ & $\#$ files  &   $ \nd(\qv) $ & $\#$ \usets in $\qv$ \\
    \hline
    $M$ & per-user \memo size  &    $ \nd(\sv) $ &  $\#$ \usets in $\sv$ \\
    \hline
    $ \mu \eqdef M/N $   & normalized  \memo size    &   $ \Dtx \subseteq \Sc $ & \Txs in multi.  \grp $\Sc$  \\
    \hline
    $ t \eqdef KM/N  $  & aggregate \memo size   & $ \alpha(\vv)$  & local FS  factor of type $\vv$ \\
    \hline
    $W_n$ & $\nth$  file   & $\alphaglobalnb(\vv)$    & \glb FS  factor of type $\vv$ \\
    \hline
    $W_{n,\Tc} $ & \sbf  & $F(\vv)$    & $\#$ type-$\vv$ \sbfs  \\
    \hline
    $W_{n,\Tc}^{(i)} $ &  $\ith$ \pkt  of \sbf  $W_{n,\Tc} $  & $F_i(\vv)$ &   $\#$ type-$\vv$ \sbfs cached by  each user in $\ith $ \uset of $\qv$ \\ 
    \hline
    $V $ & $\#$ \sbf types & $\Fm \eqdef [F(\vv_i)]_{i=1}^V $  &  \sbf  number vector \\ 
    \hline
    $S $ & $\#$ multicast \grp types  & $\Fm_i \eqdef [F_i(\vv_j)]_{j=1}^V $ & user cache vector \\  
    \hline
    $\qv $ & user  \grpg   & $\bm{\Delta}_i \eqdef \Fm_{i+1} - \Fm_{i}$  & $\ith$ cache \diffce vector \\
    \hline
    $\vv $ & \sbf type  & $ \bm{\Omega} \eqdef [\omega_{ij}]_{N_{\rm d}(\qv)\times V }  $    & \memoconst table \\
    \hline
    $\sv $ & multicast \grp types  & $\bm{\Lambda}\eqdef [ \alpha_{ij}  ]_{S\times V}$  & \fs  table \\
    \thickhline
  \end{tabular}
\end{table}

\subsection{Overview of PT Workflow}
\label{subsection: PT overview, PT framework}
\begin{figure}[ht]
    \centering
    \includegraphics[width=0.75\textwidth]{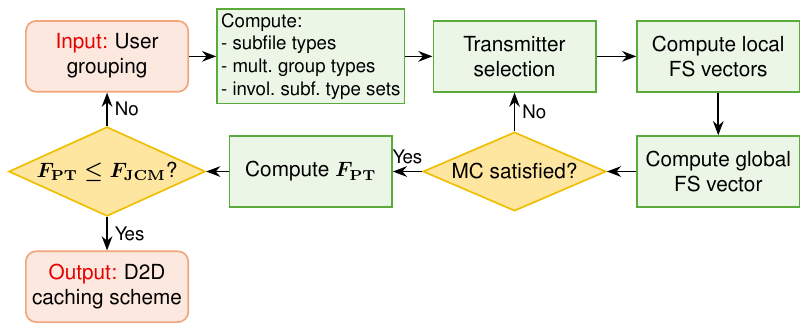}
    \caption{\small Workflow of the PT design framework.}
    \label{fig: PT flowchart}
\end{figure}

The PT framework follows a sequential design procedure, as illustrated in Fig. \ref{fig: PT flowchart}.
The steps include user grouping, calculation of subfile, packet, and multicast group types, transmitter selection, calculation of local and global FS vectors, and memory-constraint verification, as briefly explained below.

\subsubsection{User grouping \& subfile/\mgrp  types}
The starting point is a user grouping $\qv$, which partitions the  $K$ users into groups and thereby induces a geometric structure on the user set. This grouping determines the possible subfile/packet types and multicast group types. For each multicast group type $\sv_i$, the corresponding \tit{involved subfile type set} $\Ic_i$ is then identified; these are exactly the subfile types that participate in the coded multicast transmissions associated with groups of type  $\sv_i$.

\subsubsection{\Tx selection \& local FS factors}
Given these type structures, the next step is transmitter selection. For each multicast group type $\sv_i$, a \tit{subset} of users is selected as transmitters. This choice determines the local FS vector
\be
\label{eq:local FS vec, pt workflow}
\alphabm_i \eqdef [\alpha_i(\vv)]_{\vv \in \Ic_i} \in  \mathbb{Z}_{\ge 0}^{1\times |\Ic_i|}
\ee 
where $\alpha_i(\vv)$ is the number of \pkts \reqd for type-$\vv$ \sbfs when serving \mgrps of type-$\sv_i$.

\subsubsection{Global FS factors}
Since the same subfile type may appear in multiple multicast group types, the resulting local FS factors must be coordinated. This is done through the \tit{vector LCM operation}, which merges the local FS vectors into a globally consistent FS vector 
\be 
\alphaglobal=[\alphaglobalnb(\vv)]_{\vv\in \Vc}, \quad
\Vc \eqdef \bigcup_{i=1}^{\# \trm{ mult. group types}}  \Ic_i
\ee

\subsubsection{\Memoconst verification}
Once $\alphaglobal$ is obtained,
the memory constraint (MC) is verified under the corresponding file splitting and cache placement rule, where each user stores all packets of subfiles whose labels $\Tc$ contain that user. If MC is satisfied, the PT construction yields a valid D2D coded caching scheme: files are split according to $\alphaglobal$, cache placement is performed accordingly, and delivery follows the transmitter selections specified for all multicast group types.
The resulting \sbp level is 
\be
\label{eq:Fpt,pt overview}
\Fpt=\sum_{\vv\in \Vc}\alphaglobal(\vv)F(\vv),
\ee 
where $F(\vv)$ denotes the number of type-$\vv$ subfiles.

\subsection{Formal Description of PT Framework}
\label{subsection: PT framework, PT framework}

We first introduce two fundamental concepts, \ie, grouping over a set and the type of a set under a given grouping, which underpin the PT framework.
We  then  present the key components of  PT, which includes  \sbf and \mgrp types, further-splitting (FS) factors, \tx selection, local and global FS vectors, the vector \lcm (vLCM) operation that coordinates the local  FS vectors to generate universally compatible global FS vector,  and finally the  \memoconst table which ensures all users store the same amount data.

\begin{defn}[Grouping on a Set]
\label{def:set grouping}
Given a set $\Ac$ with at least $m$ elements, a  grouping is a partition of $\Ac$ into $m$ non-empty, disjoint  \grps $\{\Qc_{i,j}\}_{i,j }$ with sizes $|\Qc_{i,j}|=q_{i,j},\forall i,j$. Let  $N_{\rm d}$ denote the number of  distinct  \grp sizes. 
Then we can arrange the $m$  \grps in a nonincreasing order \accrdto  their sizes: 
$\Qc_{1,1}, \cdots ,\Qc_{1,\psi_1}, \cdots, \Qc_{N_{\rm d},1}, \cdots ,\Qc_{N_{\rm d},\psi_{N_{\rm d}}} $  where
$\psi_i$ denotes the number of size-$\beta_i \eqdef  |\Qc_{i,j}|,\forall j \in [\psi_i] $  \grps for all $i \in [N_{\rm d}]$ and  $\beta_1>\beta_2>\cdots> \beta_{N_{\rm d}}>0$. 
Hence,  $\Ac$ can be written as $\Ac= \cup_{i=1}^{N_{\rm d}} \cup_{j=1}^{\psi_i} \Qc_{i,j}$ which satisfies 
$|\Ac|= \sum_{i=1}^{N_{\rm d}}\psi_i\beta_i$, $ m=\sum_{i=1}^{N_{\rm d}}\psi_i$.
$\Ac= \cup_{i=1}^{N_{\rm d}} \cup_{j=1}^{\psi_i} \Qc_{i,j}$ is therefore called a 
 \tit{$\qv$-\grpg} with
\be 
\label{eq:def q}
\qv \eqdef \sort\left ( \{ |\Qc_{i,j}| \}_{i \in [N_{\rm d}], j\in [\psi_i]}\right)= \big(\underline{\beta_1}_{\psi_1},\cdots,\underline{\beta_{N_{\rm d}}}_{\psi_{N_{\rm d}}}\big)    \in     \Zpos^{1\times m}, 
\ee 
where the $\sort$ operator returns an ordered vector representation of the sizes of the subsets (in nonincreasing order).
See Fig. \ref{fig:q demon} for a demonstration.  
$\qv =(q_{i,j})_{i \in [N_{\rm d}], j \in [\psi_i]  } $ is called  an \emph{equal} \grpg if  $q_{i,j}=|\Ac|/m,\forall i,j $, and \tit{\uneq} \otws.
\Inadd, a \tit{unique set} $ \Uc_i \subseteq\Ac$ (\ie, a block of user \grps with identical sizes) is defined as the union of all of  all \grps with the same size, \ie,
\be
\label{def:unique set Ui}
\Uc_i \eqdef  \bigcup_{j=1 }^{\psi_j} \Qc_{i,j}, \; |\Uc_i|= \psi_i\beta_i,  \;\; \forall i\in [N_{\rm d}]
\ee 
\end{defn}
\begin{figure}[ht]
    \centering
    \includegraphics[width=0.75\linewidth]{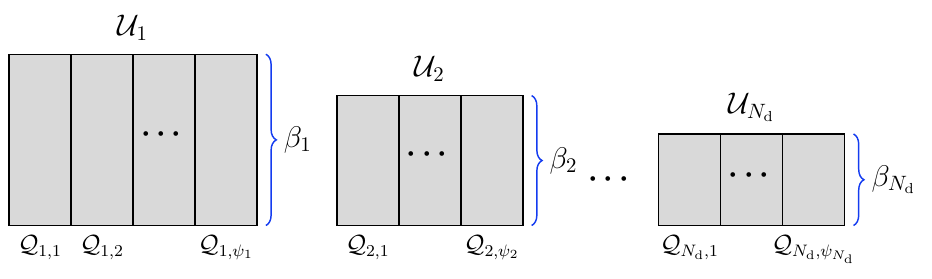}
    \caption{\small Demonstration of user \grpg. Each vertical bin represents a \uset $\Qc_{i,j}$ (containing  $\beta_i$ users) belonging   to the $\ith$ unique  \grp  $\Uc_i, i \in [N_{\rm d}]$.  }
    \label{fig:q demon}
\end{figure}

\begin{remark}
\tit{The \grpg type $\qv$ is invariant to  element permutations within each subset $\Qc_{i,j}$ and also to \perms among subsets $\{\Qc_{i,j}\}_{j\in[\psi_i]}$ belonging to each \uset $\Uc_i$.
\Fex, for $\Ac=\{1, \cdots, 6 \}$, $\qv=(2,2,1,1)$ includes all \grpgs $\{ \{i_1, i_2\},\{i_3, i_4\},\{i_5\},\{i_6\}\}$ where  $(i_1, \cdots, i_6)$ is a \perm of $(1,\cdots, 6)$.} 
\end{remark}

\begin{defn}[Type of a Set over $\qv$] 
\label{def:set type}
Given a grouping $\qv$ \eqref{eq:def q} of $\Ac$ and a subset $\Bc\subseteq \Ac$, 
the \tit{type} of $\Bc$  is defined as
\be
\label{eq:def type B on q}
\vv \eqdef  \texttt{concat}\left( \big\{
\sort\big(  \{ |\Bc \cap \Qc_{i,j}|  \}_{j=1}^{\psi_i}     \big) \big  \}_{i=1}^{N_{\rm d}}
\right)  \in  \Zplus^{1\times m},
\ee 
where the \texttt{concat} operator concatenates all sub-vectors $\sort(  \{ |\Bc \cap \Qc_{i,j}|  \}_{j=1}^{\psi_i}), i\in[\nd]$ into a single vector of length $m$.
$\vv$  can be interpreted as an ``\occu'' type of $\Bc$ among the  subsets $\{\Qc_{i,j}\}_{i,j}$. 
\end{defn}

\tit{Example:} 
Let $\Ac=\{1, \cdots, 6\} $. Consider two \grpgs, an equal \grpg 
$ \{\{1,2 \},\{3,4 \},\{5,6 \} \}$ with $\qv_1=(2,2,2)$, and an unequal \grpg   $ \{\{1,2 \},\{3,4 \},\{5 \},\{6 \} \}$ with $\qv_2=(2,2,1,1)$. Over $\qv_1$,  the subsets $\Bc=\{1,2\}, \{1,2,3\}$ and $\{1,2,3,5\}$ have types  $\vv=(2,0,0),(2,1,0)$ and $(2,1,1)$, \resp; Over $\qv_2$,  the subsets 
$\Bc=\{1,2\},\{5,6\},\{1,2,3,6\}$
have types  $\vv=(2,0,0,0),(0,0,1,0)$ and $(2,1,1,0)$, \resp.
\examend


Based on definitions \ref{def:set grouping} and \ref{def:set type}, we first define a user grouping on the D2D user set  $[K]$, and  then  specify  the  \grpg-induced subfile and \mgrp types.

\subsubsection{\tbf{User Grouping}}
\label{subsubsection: user grouping, PT framework}

The $K$ users are partitioned  into $m$ disjoint groups $\{\Qc_{i,j}\}_{ i\in [N_{\rm d}(\qv)], j \in [\psi_i]  }  $ 
with sizes 
$|\Qc_{i,j}|=\beta_i, \forall  j\in [\psi_i],\forall i\in  [\nd(\qv)]$ where $\nd(\qv)$ denotes the number of distinct \grp sizes. 
By \eqref{eq:def q},  the user \grpg can be written as
\be 
\label{eq: user grouping vector, PT framework} 
[K]= \bigcup_{i\in [N_{\rm d}(\qv)]} \bigcup_{ j\in [\psi_i] } \Qc_{i,j}  , \quad
\qv =\big(\underline{\beta_1}_{\psi_1},\cdots,\underline{\beta_{N_{\rm d}(\qv)}}_{\psi_{N_{\rm d}(\qv)}}\big) .
\ee
Since there  is only $N_{\rm d}(\qv)=1 $ unique set in equal \grpgs, we simplify the group notations as
$\Qc_1,\Qc_2, \cdots$ etc.

\subsubsection{\tbf{Subfile Type}} 
\label{subsubsection: subfile type, PT framework}
A subfile $W_{n,\Tc}$ is a portion  of file $W_n$ presumably\footnote{In the generic two-layer \sbp \eqref{eq:generic two-layer sbp, problem statement}, each file $W_n$ is split into subfiles $\{W_{n, \Tc}\}_{\Tc \in \mathfrak{T} }$ where the labels $\Tc$ are chosen from $\mathfrak{T}  \subseteq  \binom{[K]}{t}$. This means the labels in $\binom{[K]}{t} \bksl  \mathfrak{T}$ are excluded by design  and are \tit{never materialized}.  In the proposed approach,  we start  by assuming all possible $\binom{K}{t}$ labels, and then carefully exclude the redundant labels from our design.} stored  exclusively by the set of users in $\Tc \in \binom{[K]}{t}$. 
In two-layer \sbpzed D2D \coca, \sbfs are
intermediate design objects and are further split into smaller packets to  accommodate
multicast delivery.
Cache placement is nevertheless performed \tit{at the subfile level} since packets  with the same cache label $\Tc$ are  always  stored by the same set of users.
The \tit{type} of a subfile  $\WnT$  is defined as the type of the $t$-subset $\Tc$ \wrt  \ugrping  $\qv$ (see \eqref{eq: user grouping vector, PT framework})  and is denoted by
\be
\label{eq: subfile type def}
\vv \eqdef  \texttt{concat}\left( \big\{
\sort\big(  \{ |\Tc \cap \Qc_{i,j}|  \}_{j=1}^{\psi_i}     \big) \big  \}_{i=1}^{N_{\rm d}(\qv)}
\right). 
\ee 
Subfile type  $\vv$ \charzs the structure of $\Tc$ under grouping $\qv$. 
All \pkts $W_{n,\Tc}^{(i)}$ inherit the same type as $\WnT$; 
henceforth, the type of $W_{n,\Tc}$, its \pkts, and $\Tc$ are used interchangeably and denoted by $\ttt{type}(W_{n,\Tc})=\ttt{type}(W_{n,\Tc}^{(i)})=\ttt{type}(\Tc)=\vv$.

\begin{example}[Subfile Type]
\label{example: subfile type}
Let $(K,t)=(8,3)$ and $\qv=(4,4)$  with user  groups $\Qc_1=\{1,2,3,4\}$ and $\Qc_2=\{5,6,7,8\}$ as shown in Fig.~\ref{fig: example subfile type}.
Here  $\nd(\qv)=1$ so that $\vv=\sort( \{|\Tc  \cap \Qc_1|, |\Tc  \cap \Qc_2|\}$. 
Thus, the two subfile types are $\vv_1=(3,0)$ and $\vv_2=(2,1)$. 
\begin{figure}[ht]
\begin{center}
\includegraphics[width=0.31\linewidth]{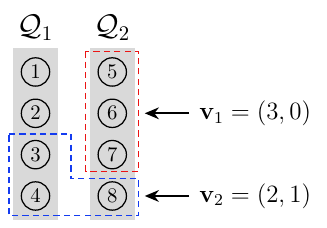}
\caption{\small Illustration of subfile types.}
\label{fig: example subfile type}
\end{center}
\end{figure}
Type-$\vv_1$ subfiles are cached by $3$ users from a single \grp, whereas type-$\vv_2$ \sbfs are cached by $2$ users from one \grp  and one from the other.
For instance,  $W_{n,\{5,6,7\}}$ is type-$(3,0)$ and $W_{n,\{3,4,8\}}$ is type-$(2,1)$.
Let $F(\vv)$ denote the number of type-$\vv$ subfiles. Then $F(\vv_1)=2\binom{4}{3}=8$ and  $F(\vv_2)=\binom{2}{1}\binom{4}{2}\binom{4}{1}=48$.
\hfill $\lozenge$
\end{example}

\subsubsection{\tbf{Multicast Group Type}}
\label{subsubsection: multi group type, PT framework}
A \mgrp 
$\Sc\in\binom{[K]}{t+1}$ involves
a set of $t+1$ users that exchange coded multicast messages in the delivery phase. The \tit{\mgrp type} 
is defined as the type of the $(t+1)$-subset $\Sc$ \wrt  \ugrping $\qv$ \eqref{eq: user grouping vector, PT framework}  and is denoted by
\be
\label{eq: multicast group type def}
\sv \eqdef  \texttt{concat}\left( \big\{
\sort\big(  \{ |\Sc \cap \Qc_{i,j}|  \}_{j=1}^{\psi_i}     \big) \big  \}_{i=1}^{N_{\rm d}(\qv)}
\right). 
\ee
For a type-$\sv$ \grp  $\Sc$, denote $\Sc_{i,j}\eqdef \Sc \cap \Qc_{i,j} $ and $\Sc_i  \eqdef \cup_{j\in [\psi_i]} \Sc_{i,j}$. It can be seen  that $ \{\Sc_{i,j}\}_{j\in [\psi_i], |\Sc_{i,j}|>0 } $  is a  \grpg of $\Sc_i$, \ie, $\Sc= \sum_{i\in [N_{\rm d}(\qv)]} \Sc_i$. 
Let  $N_{\trm{d},i}(\sv)$ denote the number of  unique sets in $\Sc_i$, and $ \hatu_{i,j} \subseteq \Sc_i $ as  the \jth, $j\in [N_{\trm{d},i}(\sv)]$ unique set of $\Sc_i$. Then there are in total $N_{\rm d}(\sv) \eqdef \sum_{i=1}^{N_{\rm d}(\qv)} N_{\trm{d},i}(\sv)  $ \usets in $\Sc$ which can be written as  
\be
\label{eq:S expressed in unique sets}
\Sc = \bigcup_{i\in [N_{\rm d}(\qv)]} \bigcup_{j \in  [N_{\trm{d},i}(\sv)]    }  \hatu_{i,j}. 
\ee 
An illustration of $\Sc_i, \Sc_{i,j}$ and $\widehat{\Uc}_{i,j}$ within the $\ith$ block \uset $\Uc_i$ of the user grouping \eqref{eq: user grouping vector, PT framework}  is shown in Fig. \ref{fig:s demo}. 
\begin{figure}[ht]
    \centering
    \includegraphics[width=0.38\linewidth]{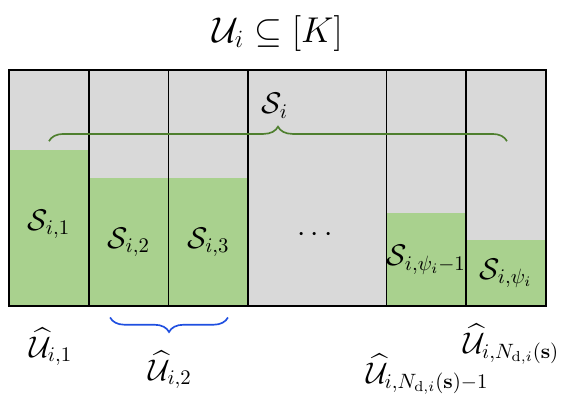}
    \caption{\small Illustration of \mgrp  $\Sc= \sum_{i\in [N_{\rm d}(\qv)]} \Sc_i$ and its occupancy within the $\ith$  \uset $\Uc_i$ of the user grouping \eqref{eq: user grouping vector, PT framework}. Green shadows represent users $ \Sc_{i,j}\subseteq \Sc$ that intersects with the $\jth$  user group $\Qc_{i,j}, j\in [\psi_i]$ in the $\ith$ \uset $\Uc_i, i\in [N_{\rm d}(\qv)]$ of $[K]=\cup_{i,j}\Qc_{i,j} $. $\widehat{\Uc}_{i,j},j\in [N_{\textrm{d}, i}(\sv)]$ is the $\jth$ \uset of $\Sc_i, i\in [N_{\rm d}(\qv)]$. Note that $N_{\textrm{d}, i}(\sv)\le  \psi_i$ and equality hold when all $\Sc_{i,j},j\in [\psi_i] $ have distinct sizes.      }
    \label{fig:s demo}
\end{figure}

\begin{example}[Multicast Group Type]
\label{example: toy multicast group type, PT framework}
Let $(K,t)=(9,5)$ so that each \mgrp is a $6$-subset $\Sc\in\binom{[K]}{6}$. Consider an (unequal)  \ugrping  $\qv=(3,2,2, 1,1)$   with $N_{\rm d}(\qv)=3$  \usets such that
\begin{align*} 
(\beta_1, \psi_1)=(3,1): & \quad   \Qc_{1,1} = \{1,2,3\}, \; \Uc_1=  \Qc_{1,1},    \\
(\beta_2, \psi_2)=(2,2): &  \quad  \Qc_{2,1} = \{4,5\}, \Qc_{2,2} = \{6,7\},\;  \Uc_2=  \Qc_{2,1} \cup   \Qc_{2,2}    \\
(\beta_3, \psi_3)=(1,2): &  \quad  \Qc_{3,1} = \{8\}, \Qc_{3,2} = \{9\}, \;  \Uc_3=  \Qc_{3,1} \cup \Qc_{3,2}. 
\end{align*}
Pick $\Sc=\{1,2,4,5,6, 8\}$. Then $\sort\big(\{|\Sc\cap\Qc_{1,1}|\}\big)=(2)$, $\sort\big( \{|\Sc\cap\Qc_{2,j}|\}_{j=1}^{2}\big)=(2,1)  $, and 
$\sort\big( \{|\Sc\cap\Qc_{3,j}|\}_{j=1}^{2}\big)=(1,0)  $,
so $\Sc$ has type $\sv=\texttt{concat}\big((2),(2,1), (1,0)\big)=(2,2,1,1,0)$.
Moreover,  the $N_{\rm d}(\sv)=4$ \usets in $\Sc$ are 
$\hatu_{1,1}=\{1,2\} (N_{{\rm d},1
}(\sv)=1), \hatu_{2,1}=\{4,5\},   \hatu_{2,2}=\{6\}(N_{{\rm d},2
}(\sv)=2)$, and $\hatu_{3,1}=\{8\}(N_{{\rm d},3
}(\sv)=1)$.
There are  $F(\sv)= \binom{|\Qc_{1,1}|}{2} 2! \binom{|\Qc_{2,1}|}{2} \binom{|\Qc_{2,2}|}{1} \binom{2}{1} \binom{|\Qc_{3,1}| }{1}= 24
$ type-$(2,2,1,1,0)$ \mgrps. 
In addition, for the equal user grouping $\Qc_1=\{1,2,3 ,4\}$, $\Qc_2=\{5,6,7,8\}$ with $t=3$ in Example~\ref{example: subfile type}, there are $3$  \mgrp types  
$\sv_1=(4,0),\sv_2=(3,1)$, and $\sv_3=(2,2)$, with total numbers of $F(\sv_1)=2\binom{|\Qc_1|}{4}=2 $, $F(\sv_2)=2!\binom{|\Qc_1|}{3}\binom{|\Qc_2|}{1}=32$ and 
$F(\sv_3)=\binom{|\Qc_1|}{2}\binom{|\Qc_2|}{2}=36$, \resp. 
\hfill $\lozenge$
\end{example}

\begin{defn}[Involved Subfile Type Set]
Given any \mgrp $\Sc$ with type $\sv$ (see  \eqref{eq:S expressed in unique sets},
\eqref{eq: multicast group type def}),
the set of subfile labels involved in $\Sc$ are $\{\Tc= \Sc\bksl \{k\}: k\in\Sc \}$. 
These $\Tc$ labels constitute  the following \sbf types:
\be
\label{eq:def involved subfile types}
\Ic_{\sv } \eqdef   \left\{\vv=\type(\Sc\bksl \{k_{i,j}\} ):
k_{i,j} \in \widehat{\Uc}_{i,j}, i \in [N_{\rm d}(\qv)], j\in [N_{\textrm{d},i}(\sv)]\right\}
\ee
which contains $|\Ic_\sv|=N_{\rm d}(\sv)=\sum_{i=1}^{N_{\rm d}(\qv)} N_{\textrm{d},i}(\sv)$ \diff subfile types. Note that
$\Ic_\sv$ is the same for all \mgrps with type $\sv$.
\end{defn}

Continue  with \Ex \ref{example: toy multicast group type, PT framework}. The type-$(2,2,1,1,0)$ \mgrp $\Sc =\{1,2,4,5,6,8\}$ contains four \diff subfile types: $\Ic_{\sv=(2,2,1,1,0)}=\{ 
(1,2,1,1,0), (2,1,1,1,0),(2,2,0,1,0), (2,2,1,0,0) 
\}$.

\subsubsection{\tbf{Further-Splitting  Factor \& \Tx Selection}}
\label{subsubsection: FS factor & Tx selection, PT framework}
For \sbf type  $\vv$\footnote{\label{footnote:fs factor same within type}In the PT \frmwk, we \tit{assume} subfiles of the same type must have an identical FS factor.   }, the \emph{further-splitting (FS) factor} $\alpha(\vv)\in \mathbb{Z}_{\ge 0} $ is the number of packets into which each type-$\vv$ subfile is divided to ensure the optimal  DoF condition (see \Rmk  \ref{rmk:optimal dof}).
While the JCM scheme uniformly sets $\alpha_{\rm JCM}=t$ for all   \sbfs, this choice is sufficient but not necessary for DoF optimality. In contrast, our PT framework allows \tit{nonuniform} FS factors across subfile types, enabling \sbp reduction via \tx selection.
For ease of presentation, assume from now no we deal with equal user grouping $\qv$ where $N_{\rm d}(\qv)=1$ so that \eqref{eq:S expressed in unique sets}   and \eqref{eq:def involved subfile types}  can be  simplified as 
$
\Sc = \cup_{j \in  [N_{\trm{d}}(\sv)]    }  \hatu_{j}  
$, 
$  \Ic_{\sv } \eqdef \big\{\vv=\type(\Sc\bksl \{k_{j}\} ):
k_{j} \in \widehat{\Uc}_{j}, j\in [N_{\textrm{d}}(\sv)]\big\}   $  with slight abuse of notation
$\nd(\sv)\eqdef  N_{\textrm{d},1}(\sv), \hatu_{j} \eqdef \widehat{\Uc}_{1,j},\forall j\in [\nd(\sv)]$.

Consider the coded \mtcst \txn within a
type-$\sv$ \mgrp $\Sc=\cup_{i\in [{N}_{\rm d}(\sv)]}\hatu_i$ with \usets  $\hatu_1,\cdots, \hatu_{N_{\rm d}(\sv)}$.
A subset of users $\Utx\subseteq \Sc  $ is selected as transmitters while all users in $\Sc \bksl  \Utx$ only receive coded \pkts.
$\Utx$ is restricted to be a \tit{union of unique sets} in $\Sc$, i.e.,
\be 
\label{eq:Utx as union of unique sets, Tx selection}
\mathrm{[Transmitter\;\,  selection]}\quad 
\Uc_{\rm Tx}=\bigcup_{i\in\mathcal{D}_{\rm Tx}}\widehat{\Uc}_{i}, \quad \Dtx \subseteq \big[{N}_{\rm d}(\sv)\big],
\ee 
which ensures consistency of further splitting among subfiles of the same type (see Remark~\ref{remark:reason of Utx as union of unique sets}  explanation).
The involved \sbf types are 
$\Ic_\sv=\{\vv_1, \cdots, \vv_{{N}_{\rm d}(\sv)}  \}$, 
where $\vv_i=\type(\{W_{n, \Sc\bksl \{k\}}\} ), k \in  \widehat{\Uc}_i$ denotes the type of the \sbfs desired by users in $\hatu_i,i\in [{N}_{\rm d}(\sv)] $. 
Under $\Utx$, the FS factor for $\vv_i$ is  given by 
\begin{eqnarray}
\label{eq: FS factor}
\mathrm{[Local \;\, FS\;\,factor]} \quad 
\alpha(\vv_i)=\left\{\begin{array}{ll}
\sum\limits_{k\in\mathcal{D}_{\rm Tx}}|\hatu_k  | -1,& \text{if}\; i\in\mathcal{D}_{\rm Tx}\\
\sum\limits_{k\in\mathcal{D}_{\rm Tx}}|\hatu_k  |, &\text{if} \; i\notin\mathcal{D}_{\rm Tx}
\end{array}\right.
\end{eqnarray}
which can be written more compactly as
\be
\label{eq:alpha(vi) compact}
\alpha(\vv_i)= |\Utx|-\mathbbm{1}\{i \in \Dtx  \}.
\ee
Note that the \tx selection \eqref{eq:Utx as union of unique sets, Tx selection} is performed \tit{at the type level} as $\alpha(\vv_i)$ depends on $\sv$   through $|\Utx|=\sum_{k\in\mathcal{D}_{\rm Tx}}|\hatu_k  | $ instead of the specific $\Sc$ within this type.
A  \tx selection strategy refers to a collection  of selections for all possible \mgrp types.
For any $\sv=(s_1^\dagger,s_2, \cdots,s_{m-1}^\dagger, s_m)$, we use  the superscript $\dagger$ to indicate the unique sets  that have been selected as \txs.

To make sense of \eqref{eq:alpha(vi) compact}, let us consider a user $k\in \hatu_i$ whose desired subfile has type $\vv_i$.
Excluding self-\txn (when applicable), the number of \txs observable to \user{k} is equal  to
$|\Utx \bksl \{ k \}  | = |\Utx|- \mathbbm{1}\{k \in \Utx \}= 
|\Utx|- \mathbbm{1}\{ i\in \Dtx \}$. The optimal DoF condition requires each \cmsg  benefit exactly  $t$ users. Consequently, \user{k} must decode a \diff packet of  the desired subfile $W_{d_k, \Sc \bksl \{k\}}$ from  the  \txn slot assigned to each \tx in $\Utx$.
\Aar, $W_{d_k, \Sc \bksl \{k\}}$ (and all type-$\vv_i$ \sbfs) should contain exactly $|\Utx|- \mathbbm{1}\{ i\in \Dtx \}$ \pkts, \ie, 
$\alpha(\vv_i)= |\Utx|-\mathbbm{1}\{i \in \Dtx  \} $ as  given in 
\eqref{eq:alpha(vi) compact}.
In other words, transmitter selection determines the local
FS factors through the number of transmission slots available
to each receiver. 
A summary of this  process is given in Fig. \ref{fig:tx2fs}.
\begin{figure}
    \centering
    \includegraphics[width=0.26\linewidth]{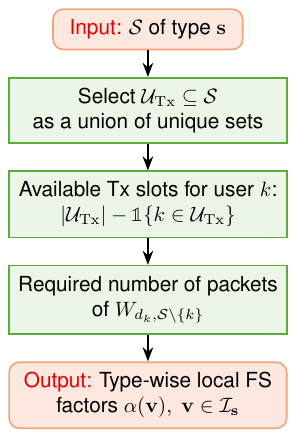}
    \caption{\small Transmitter selection determines local FS factor. Once the transmitter set $\mathcal U_{\rm Tx}$ is selected, the number of transmission slots available to each user $k$ determines the required number of packets of $W_{d_k,S\setminus\{k\}}$, and hence the local FS factor $\alpha(\vv)$ of the corresponding subfile type $\vv=\type(\Sc\bksl \{k\})$  .}
    \label{fig:tx2fs}
\end{figure}
\begin{remark}[Unions of Unique Sets as \Txs]
\label{remark:reason of Utx as union of unique sets}
\tit{The restriction in \eqref{eq:Utx as union of unique sets, Tx selection} is a consequence of the type-based further-splitting rule  in the PT \frmwk that \sbfs  of the same type must  have identical FS factors (see Footnote \ref{footnote:fs factor same within type}).
Hence, to ensure consistency of further splitting, users within the same unique set must be treated uniformly:
either all are selected as transmitters or none are.
Otherwise, two users from the same unique set would induce different FS factors under (\ref{eq: FS factor}), 
contradicting the requirement that subfiles of the same type share identical FS factors.
Specifically, consider two users $k,k'\in \hatu_i  $ belonging to the same \uset.
Suppose $k \in \Utx, k'\notin \Utx$. 
By (\ref{eq: FS factor}), the \sbfs $W_{d_k,\Sc\backslash \{k\}}$ and $W_{d_{k'},\Sc\backslash \{k'\}}$
have   FS factors $|\Utx|-1$ and $|\Utx|$, \resp.
Since they share the same type $\vv_i$, this leads to an inconsistency in further splitting.}
\end{remark}

Equation \eqref{eq: FS factor} reveals a core  insight of the PT \frmwk: the FS factors  differ from number of \txs by at most one.
\Thf, designating smaller number of \txs  in the  \mgrps directly reduces the FS factors and the overall \sbp  $\Fpt=\sum_{i=1}^V \alphaglobalnb(\vv_i)F(\vv_i)$ as shown in \Ex \ref{example: FS factor & Tx selection, PT framework}) ($V$ denotes the total number of \sbf types).
The FS factors  in (\ref{eq: FS factor}) are \tit{local}  as they are obtained only from a specific type $\sv$. 
When a subfile type appears in multiple \mgrps, \diff  local FS factors must be coordinated. 
This is achieved via the proposed  \vlcm (LCM) operation (see Definition \ref{def: LCM}), which maps local FS factors to a set of \tit{global} FS factors
\be
\label{eq:def alpha^global, PT framework}
\alphaglobal \eqdef\big (\alpha^{\rm global}(\vv_i)\big )_{i=1}^V,
\ee 
which govern the  overall  subfile splitting strategy  of the caching \schm.

\begin{example}[\Tx Selection \& Local FS Factor]
\label{example: FS factor & Tx selection, PT framework}
Consider Example~\ref{example: subfile type} ($t=3$) with  user \grpg  $\Qc_1=\{1,2,3,4\}$, $\Qc_2=\{5,6,7,8\}$.
The type-$(3,1)$ \mgrp $\Sc=\{1,2,3,5\}$ contains two unique sets $\hatu_1=\{1,2,3\}$ and $\hatu_2=\{5\}$.
Each nonempty $\Dtx \subseteq \{1,2\}$ yields a transmitter selection:

\begin{itemize}
    \item[] \tbf{Choice 1 (worst)}: $\Dc_{\rm Tx} =\{1\}$.
    The \txs are $\Utx=\{1,2,3\}$, giving \lfsfs  $\alpha(\vv_1)=3$  and $\alpha(\vv_2)=2$. 
    Jointly coordinating the local FS vectors of all multicast group types via the vector LCM operation gives
    $\alphaglobal=(9,6) $, yielding $\Fpt=9F(\vv_1)+ 6 F(\vv_2)= 360>\Fjcm=168$. Thus, this choice increases, rather than reduces, the \sbp relative to the JCM scheme.

    \item[] \tbf{Choice 2 (JCM)}: $\Dc_{\rm Tx} =\{1,2\}$.
    All users in $\Sc$ transmit, yielding $\alphaglobal=(3,3)  $ and 
    $\Fpt=\Fjcm=168$, matching the JCM \sbp.

    \item[] \tbf{Choice 3 (PT, best)}: $\Dc_{\rm Tx} =\{2\}$.
    The sole transmitter is user $5$. To ensure decodability, all type-$\vv_1$ \sbfs are excluded, giving $ \alphaglobalnb=(0,3)$. This yields 
    $\Fpt=144< \Fjcm$, demonstrating subpacketization reduction via clever \tx selection. \hfill $\lozenge$
\end{itemize}
\end{example}

\subsubsection{\tbf{Local FS vector \& FS Table}} 
\label{subsubsection: FS table, PT framework}

For a given transmitter selection, each multicast group type $\sv_i$ induces a set of local FS factors according to (\ref{eq: FS factor}), represented by the \tit{local FS vector}  
$ \bmalpha_i=  [ \alpha_i(\vv) ]_{\vv \in  \Ic_i}$ (here $\Ic_{\sv_i}$ is simplified as $\Ic_i$). 
Let  $V$ and  $S$ denote the total numbers of possible subfile types and multicast group types, respectively. The \tit{further-splitting (FS) table} is defined as a matrix
\be
\label{eq:FS table def, PT framework}
\mathbf{\Lambda}\eqdef[\alpha_{ij}]_{S\times V},
\ee 
which collects the local FS factors for all multicast group types. Specifically, the 
$\ith$ row corresponds to $\sv_i$, \ie, 
$\mbf{\Lambda}_{i,:} = \bmalpha_i,\forall  i\in[S] $. Entries corresponding to $\vv \notin \Ic_i$
are marked by $\star$ (meaning no appearance)  while excluded subfile types are assigned zero.

\subsubsection{\tbf{Global FS Vector}} 
Given a user \grpg $\qv$, 
the local FS vectors $\alphabm_1, \cdots, \alphabm_S$   induced by transmitter selections
(see \eqref{eq:Utx as union of unique sets, Tx selection}) for different
multicast group types  $\sv_1, \cdots, \sv_S$
must be coordinated, since the same subfile type
$\mathbf{v}$ may appear in multiple multicast group types with
different local FS factors. Consider, for example, a subfile type
$\mathbf{v}$ involved in $\mathbf{s}_1$
and $\mathbf{s}_2$, with local FS factors $\alpha_1(\mathbf{v})$
and $\alpha_2(\mathbf{v})$ where $\alpha_1(\mathbf{v})\ne \alpha_1(\mathbf{v})$.
Their consistency can be
achieved by setting the global FS factor to  their \lcm (LCM):
$$
\alpha^{\mathrm{global}}(\mathbf{v})
=\mrm{LCM}(\alpha_1(\mathbf{v}),\alpha_2(\mathbf{v}))
=z_1\alpha_1(\mathbf{v})=z_2\alpha_2(\mathbf{v}),$$
where $z_1$ and $z_2$ are the smallest positive integers satisfying
these equalities. Accordingly, each type-$\mathbf{v}$ subfile is
split into $\alpha^{\mathrm{global}}(\mathbf{v})$ equal-size ``atom" packets.
In multicast groups of types $\mathbf{s}_1$ and $\mathbf{s}_2$,
each \cmsg XOR component is then formed by stacking $z_1$
and $z_2$ consecutive packets, respectively, ensuring that both
realize the same total subfile size while preserving their
respective local FS structures. 
\Fex,  
imagine $ \alpha_1(\vv)=2, \alpha_2(\vv)=3$ so that  $z_1=3, z_2=2$ and $\alphaglobalnb=6$ as
shown  in Fig. \ref{fig:fs atom packets}. 
\begin{figure}[ht]
    \centering
    \includegraphics[width=0.56\textwidth]{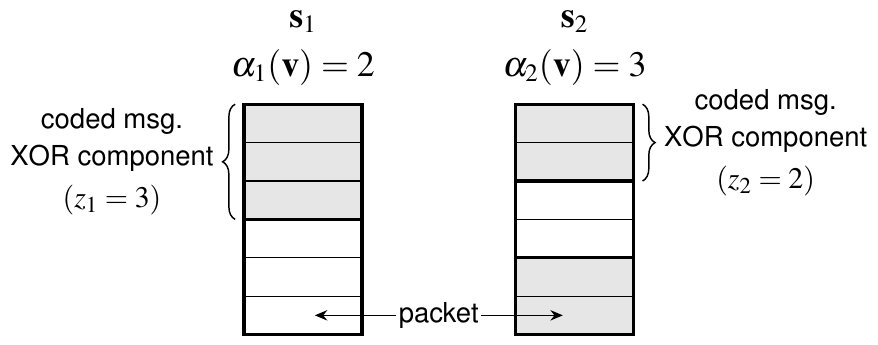}
    \caption{\small  Each type-$\mathbf v$ subfile
  is split into $\alpha^{\mathrm{global}}(\mathbf v)=6$
  atom packets (thin lines). In $\mathbf s_1$, each XOR component of a coded message
  stacks $z_1=3$ consecutive packets, and in $\mathbf s_2$ it stacks $z_2=2$
  (thick lines and shading mark the components). Both groups therefore carry the
  same total subfile size while keeping their local FS structures.}
    \label{fig:fs atom packets}
\end{figure}
Another concrete \ex is  given in Fig. \ref{fig: delivery, (9,3,2) example, PT first look} in \Ex \ref{example: (9,3,2) PT first look} where  the \sbf  type $\vv_2=(3,2,1)$ appears in both $\sv_1=(3,3,1)$ and $\sv_2=(3,2,2)$ with local FS factors $\alpha_1(\vv_2)=1$ and   $\alpha_2(\vv_2)=3$ \resp. 
Thus, each type-$\vv_2$ \sbf has $3$ packets.
In the type-$\sv_1$ \txn  in Fig. \ref{fig: S1 delivery, (9,3,2) example, PT first look},  the XOR  components are
$\big(W^{(1)}_{d_k, \Sc_1\bksl \{k\} }, W^{(2)}_{d_k, \Sc_1\bksl \{k\} } , W^{(3)}_{d_k, \Sc_1\bksl \{k\} }   \big)^T,k\in \{1,2,6\}$, each of which is a stack of $3$ \pkts;
In the type-$\sv_2$ \txn  in Fig. \ref{fig: S2 delivery, (9,3,2) example, PT first look},  the XOR  components are
$W^{(1)}_{d_k, \Sc_1\bksl \{k\} },k\in \{5,7,8\}$, each of which contains a single \pkt.

The local FS vectors of all subfile types  must be coordinated
across all \mgrp types to ensure global consistency.
To this end, we employ a vector least common multiple (vLCM) operation
to synchronize the local FS factors of each subfile type across all
multicast group types. In addition, a subfile-type exclusion procedure
is required to identify and eliminate redundant subfile types, since
the vLCM operation alone does not account for exclusions.
Together, these two operations yield a globally consistent FS vector
$\alphaglobal$, which determines the final subpacketization structure. 
The vLCM  and exclusion operations are defined \af.

\begin{defn}[Vector Least Common Multiple]
\label{def: LCM}
Given an FS table
$
\boldsymbol{\Lambda}
=
[\alpha_{ij}]
\in
\left(\mathbb{Z}_{\geq 0}\cup\{\star\}\right)^{S\times V}
$, define its positive
support as
$
\supp(\boldsymbol{\Lambda})
\eqdef
\{(i,j)\in[S]\times[V]:\alpha_{ij}>0\}.
$
The \tit{vector \lcm} (vLCM) $\alphalcm$ of $\Lambdam$  is  the vector
\be
\alphalcm 
\eqdef 
(\widetilde{\alpha}_1,\cdots,\widetilde{\alpha}_V)
\in\mathbb{Z}_{\geq 0}^{1\times V}
\ee
obtained from
\be 
\label{eq:def a^LCM, PT frmwk}
\begin{aligned}
(\alphalcm,
 \boldsymbol{z}^* )= 
\argmin_{\substack{\widetilde{\boldsymbol{\alpha}}\in\mathbb{Z}_{\geq0}^{1\times V},\\
\boldsymbol{z}\in\mathbb{Z}_{>0}^{1\times S}
}}
\quad&
\sum_{j=1}^{V}\widetilde{\alpha}_j
\\
\textrm{s.t.}\quad \;\;\; &
\widetilde{\alpha}_j
=
z_i\alpha_{ij},
\qquad
\forall(i,j)\in\supp(\boldsymbol{\Lambda}),
\\
&
\widetilde{\alpha}_j>0,
\qquad
\forall j\text{ appearing in }\supp(\boldsymbol{\Lambda}),
\\
&
\widetilde{\alpha}_j=0,
\qquad
\forall j\text{ not appearing in }
\supp(\boldsymbol{\Lambda}).
\end{aligned}
\ee
If the constraint set is empty,  then the vLCM does
not exist.
\end{defn}

Intuitively, the scaling factor $z_i$ specifies how many times the
local delivery pattern associated with  $\sv_i$ 
is replicated. The constraint
$\widetilde{\alpha}_j=z_i\alpha_{ij}$ ensures that whenever the same
subfile type appears in different multicast group types, all of its
local FS factors are synchronized to one common global FS factor $\widetilde{\alpha}_j$. The objective selects
the smallest such nonnegative integer vector, thereby removing any
unnecessary common multiplicative scaling. Entries equal to $0$ or
$\star$ impose no synchronization constraint and are therefore ignored
by the vLCM operation.
Together with the subfile exclusion rule defined in \Algo \ref{alg:exclusion_closure},  the vLCM operation takes  the local FS vectors $\Lambdam$  as inputs and output a globally  compatible FS vector $\alphaglobal$ with proper subfile type
exclusion.

\begin{example}[Vector LCM]
\label{example:lcm vector}
Consider two \diff FS tables
{\renewcommand{\arraystretch}{0.2}
\be
\Lambdam_1 =
\begin{bmatrix}
2 & \star  & \star \\
0 & 1 &  \star \\
\star & 1 &  0 \\
\star & \star &  2 \\
\end{bmatrix}, 
\quad
\Lambdam_2 =
\begin{bmatrix}
3 & \star    \\
0 & 1   \\
\star & 3 
\end{bmatrix}.
\ee}For $\Lambdam_1$ with  $\supp(\Lambdam_1)=\{(1,1),(2,2), (3,2), (4,3)\}$, the alignment constraints are
$  \alphatilde_1=2z_1, \alphatilde_2=z_2= z_3   $ and $\alphatilde_3=2z_4$. Hence, $ \alphatilde_1 +\alphatilde_2+\alphatilde_3= z_1+z_2+2z_4$ which is minimized by $z_1^*=z_2^*=z_3^*=z_4^*=1$ and $\alphalcm=(2,1,2)$.
For $\Lambdam_2$ with  $\supp(\Lambdam_2)=\{(1,1),(2,2), (3,2)\}$, the alignment constraints are
$  \alphatilde_1=3z_1, \alphatilde_2=z_2= 3z_3   $ so that $\alphatilde_1+\alphatilde_2= 3z_1+3z_3$ which is minimized by $z_1^*=z_3^*=1,z_2^*=3$ and $\alphalcm=(3,3)$. 
\hfill  $\lozenge$
\end{example}

Given the FS table $\Lambdam$, we want to identify \sbf  types that can be excluded before performing the vLCM \cdntn in order to maximize \sbp reduction. 
Lemma \ref{lemma:Condition for Subfile Exclusion} gives the necessary and sufficient conditions for subfile  exclusion. 

\begin{algorithm}[t]
\setstretch{0.8}
\caption{Zero-Triggered Exclusion Procedure}
\label{alg:exclusion_closure}

\KwData{Local FS table
$\bm{\Lambda}=[\alpha_{ij}]_{S\times V}$
and involved subfile type sets
$\{\Ic_i\}_{i=1}^S$.}

Initialize the exclusion sets:
\[
\Ec_v \gets \Ec_v^{(0)},
\quad
\Ec_s\gets\emptyset.
\]

\Repeat{
$\Ec_v=\Ec_v^{\rm old}$
\textbf{and}
$\Ec_s=\Ec_s^{\rm old}$
}{
    $\Ec_v^{\rm old}\gets\Ec_v$;
    $\Ec_s^{\rm old}\gets\Ec_s$\;

    \For{each $\sv_i\in \{\sv_k\}_{k=1}^S   \setminus\Ec_s$}{
        \If{$\exists\,\vv_j\in\Ec_v$
        such that $\alpha_{ij}>0$}{
            $\Ec_s\gets\Ec_s\cup\{\sv_i\}$\;
            $\Ec_v\gets\Ec_v\cup\Ic_i$\;
        }
    }
}

\eIf{$\Ec_v=\{\vv_k\}_{k=1}^V $ \textbf{or} $\Ec_s=\{\sv_k\}_{k=1}^S$}{
    \KwResult{Infeasible (trivial exclusion).}
}{
    \KwResult{Excluded subfile types
    $\Ec_v^\star=\Ec_v$
    and deactivated multicast group types
    $\Ec_s^\star=\Ec_s$.}
}
\end{algorithm}
\begin{lemma}[Subfile Type Exclusion Condition]
\label{lemma:Condition for Subfile Exclusion}
\tit{Given the FS table $\Lambdam=[\alpha_{ij}]_{S\times V}$, let
$\Ec_v\subsetneq\{\vv_j\}_{j=1}^V$ and
$\Ec_s\subsetneq\{\sv_i\}_{i=1}^S$
denote the excluded subfile types and deactivated\footnote{If all \sbf types in $\Ic_i$ are excluded, then the \mgrp type $\sv_i$ is deactivated, \ie, this type of \txn will not appear in the final caching scheme.} multicast group types, respectively.
An exclusion pattern $(\Ec_v,\Ec_s)$  is  valid
if it satisfies the following conditions:
\begin{itemize}
    \item[] Condition 1:
    If $\vv_j\in\Ec_v$ and $\alpha_{ij}>0$,
    then $\sv_i\in\Ec_s$.
    \item[] Condition 2:
    If $\sv_i\in\Ec_s$ and $\alpha_{ij}\neq\star$,
    then $\vv_j\in\Ec_v$.
\end{itemize}
Starting from some initial exclusions
$\Ec_v^{(0)}
=\{\vv_j:\alpha_{ij}=0\text{ for some }i\in[S]\}$,
Algorithm~\ref{alg:exclusion_closure} computes
the unique smallest exclusion-consistent $\Ec_v^\star$
containing $\Ec_v^{(0)}$.}
\end{lemma}

Lemma \ref{lemma:Condition for Subfile Exclusion} ensures consistency in \sbp because subfile-type
and multicast-group-type exclusions are interdependent.
Excluding a subfile type invalidates any multicast group
type requiring its packets (Condition 1), while deactivating
a multicast group type forces all its involved subfile
types to be excluded (Condition 2). Arbitrary exclusions
may violate these dependencies and thus disrupt the
prescribed delivery patterns (see Section \ref{subsec: subpacketization structure of D2D}). Starting from an zero-triggered
initial exclusion $\Ec_v^{(0)}$, Algorithm~\ref{alg:exclusion_closure}
repeatedly enforces both conditions until reaching a fixed
point, thereby obtaining the smallest exclusion-consistent
pattern.

\begin{remark}
\label{rmk:graph interpretation of subfile exclusion}
\tit{The exclusion conditions has a simple graph-theoretic
interpretation \af. Consider a bipartite graph whose vertices
represent subfile types and multicast group types, with
an edge $(\mathbf v_j,\mathbf s_i)$ whenever
$\alpha_{ij}>0$. Conditions 1 and 2 imply an
all-or-none property: all vertices within each connected
component must be either retained or excluded together.
Consequently, the zero-triggered exclusion procedure in Algorithm~\ref{alg:exclusion_closure}
removes precisely those connected components containing
at least one subfile type with a zero local FS factor,
yielding the smallest exclusion closure.}
\end{remark}

\begin{remark}[Interpretation of Subfile Exclusion] 
\tit{Excluding a subfile type does not imply discarding any portion of the original file. Rather, subfile labels represent possible cache-placement patterns, and excluding a type means that the corresponding labels are not used in the file partitioning. Each file is instead partitioned into \(F_{\mathrm{PT}}\) equal-size packets associated only with the retained subfile types and their respective further-splitting factors. These packets collectively constitute the entire original file. Thus, subfile saving eliminates redundant placement patterns, not file content.}
\end{remark}

The LCM vector may not exist; when it exists, it is unique. In the PT framework, the global FS vector is defined as the vector LCM of the local FS vectors with proper \sbf \excs:
\be 
\label{eq: global FS from local FS, LCM def}
\mathrm{[Global \;\, FS\;\,vector]} \quad 
\alphaglobal =   \textrm{vLCM}\left(\bmalpha_1, \cdots , \bmalpha_S\right) \in \Zplus^{1\times V}.
\ee 
\Ip, the zero entries in $\alphaglobal$ represent excluded \sbf types.
Fig. \ref{fig:local_to_global_FS_vector} summarizes the steps to obtain $\alphaglobal$ from a specific choice of \tx selections for the \mgrp types. 
 \begin{figure}[t]
     \centering
     \includegraphics[width=0.28\textwidth]{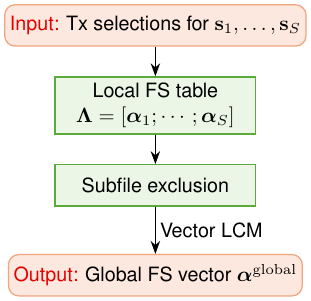}
     \caption{\small From transmitter selection to the global FS vector. The transmitter selections across multicast group types determine the local FS table $\Lambdam$.
    A subfile-type exclusion procedure removes \redun
    subfile types, after which the vLCM operation coordinates the
    remaining local FS factors to obtain the globally
    consistent FS vector $\alphaglobal$.}
     \label{fig:local_to_global_FS_vector}
 \end{figure}
Consequently, the overall \sbp level is given by
\be 
\label{eq: PT subp, PT frmwk components}
\mathrm{[PT \;\, \sbp]} \quad 
 \Fpt  = \alphaglobal \lef[F(\vv_1), \cdots, F(\vv_V) \rig ]^{\rm T},
\ee 
where $F(\vv_i)$  denotes the total number of type-$\vv_i$ \sbfs.
Note that $F(\vv_i)$'s are determined by the \ugrping $\qv$ while $\alphaglobal $ is determined by both $\qv$ and \tx selections.

\begin{example}[Calculation of $\alphaglobal$]
\label{example: compute global FS factors, PT framework}
Continuing with Example~\ref{example: FS factor & Tx selection, PT framework}, we show how $\alphaglobal$ can be calculated.
The subfile and \mgrp types are 
$ \vv_1=(3,0),\vv_2=(2,1)$, $\sv_1 = (4^{\dagger},0  ) ,
\sv_2 = (3,1^{\dagger} ),
\sv_3 = (2^{\dagger},2^{\dagger})$
with  $\Ic_1=\{\vv_1\}, \Ic_2=\{\vv_1,\vv_2\}$, and $\Ic_3=\{\vv_2\}$. 
Transmitters are indicated by $\dagger$.
The resulting local FS vectors are  $\bmalpha_1=(3,\star),\bmalpha_2=(0,1)   $,  and $\bmalpha_3=(\star,3)$, yielding $\alphaglobal=(0,3)$  shown in Table \ref{table: FS table, example global FS via LCM, PT framework}.
\begin{table*}[h]
\centering
\caption{FS Table and Global FS Vector}
\renewcommand{\arraystretch}{0.8}
 \begin{tabular}{|c"c|c|} 
\hline
  &  $\vv_1$ & $\vv_2$  \\
\thickhline
$\bm{\alpha}_1$ & 3  & $\star$  \\
\hline
$\bm{\alpha}_2$  & 0 &1 \\
 \hline
 $\bm{\alpha}_3 $  & $\star$ &3 \\
 \hline
 $\alphaglobal $  & 0 &3\\
 \hline
\end{tabular}  
\label{table: FS table, example global FS via LCM, PT framework}
\end{table*}
\Accrdto  \Algo \ref{alg:exclusion_closure},
the zero entry $\alpha_{21}=0$ triggers the exclusion
of $\mathbf v_1$, which further deactivates
$\mathbf s_1$ since $\alpha_{11}>0$.
Applying the vLCM operation to the remaining
FS table yields $\alphaglobal=(0,3)$.
Since $F(\vv_1)=8$  and $F(\vv_2)=48$, we have
$\Fpt =(0,3)(8,48)^{\rm T}=144$,
which is smaller than $\Fjcm=168$. The reduction is due to the exclusion of type-$\vv_1$ subfiles (\ie, \ssgain only).
\hfill $\lozenge$
\end{example}

\begin{example}[\Ex \ref{example: (4,2,1) PT first look} Revisited]
With the PT framework in place, we can now formalize Example~\ref{example: (4,2,1) PT first look} (Section \ref{subsec: a first look at PT, problem description}) as follows.
The \grpg $\qv=(2,2)$ is adopted, with user assignments $\Qc_1=\{1,2\}$ and $\Qc_2=\{3,4\}$.
There are two  \sbf types  $\vv_1=(2,0), \vv_2=(1,1)$, and one \mgrp type $\sv=(2,1^\dagger)$. For each $\Sc$, the user \corrspdg  to the second element of $\sv$ is  selected as the \tx, resulting in $ \alphaglobal= (0,1)$. Thus, type-$\vv_1$ \sbfs are excluded, and \fs  is not  needed for type-$\vv_2$ since $\alpha^{\rm \glb }(\vv_2)=1$. \Thf, each file is split into  $\Fpt=F(\vv_2)=4$ subfiles,
$ 
W_n=\lef \{W_{n, \Tc}\rig  \}_{\Tc  \in \mathfrak{T}},\;  \mathfrak{T}=\{ \{1,3\},\{1,4\},\{2,3\},\{2,4\}\},
$
and the cache placement and  \deli phases involve  only type-$\vv_2$ \pkts. 
\hfill $\lozenge$
\end{example}

\subsubsection{\tbf{Memory Constraint Table}} 
\label{subsubsection: MCT, PT framework}

Given a user grouping  $\qv$  with \usets $\Uc_1,\cdots,\Uc_{N_{\rm d}(\qv)}$, 
the \emph{memory constraint (MC) table} is defined as
\be 
\label{eq:def Omega MC matrix}
\mathbf{\Omega}\eqdef[\omega_{ij}]_{N_{\rm d}(\qv) \times V},
\ee
where $\omega_{ij}\eqdef F_i(\vv_j)$ denotes the number of type-$\vv_j$ \sbfs  presumably\footnote{This includes all subfiles, even those ultimately excluded; the final scheme removes such subfiles accordingly.} cached by each user in $\Uc_i$. 
Define the  \tit{user cache vector}
\be
\label{eq:def F_i, PT framework}
\mathbf{F}_i\eqdef \left[F_i(\vv_1),\cdots,F_i(\vv_V)\right],
\ee 
which equals the  $\ith$ row of $\mbf{\Omega}$, and the \tit{\sbf number vector} 
\be 
\label{eq:def F subfile nubmer vector, PT framework}
\mathbf{F}\eqdef[F(\vv_1),\cdots, F(\vv_V)],
\ee
where $F(\vv)$ is the total number of type-$\vv$ \sbfs. 
Also define the \tit{cache difference vector}
\be
\label{eq:def Delta_i, PT framework}
\mathbf{\Delta}_i\eqdef\mathbf{F}_{i+1}-\mathbf{F}_i,\;i\in [\nd(\qv) -1]
\ee 
whose $\jth$ entry $\Delta_{i}(j)=F_{i+1}(\vv_j)-F_{i}(\vv_j), j\in[V]$
represents the difference in the number of type-$\vv_j$ subfiles  cached per user in $\Uc_{i+1}$ and $\Uc_i$, \resp.
Assuming equal packet sizes, the \tit{memory constraint} (MC) is
\be 
\label{eq: MC}
\mathrm{[Memory\;\,constraint]} \quad 
\alphaglobal\mbf{\Delta}_i^{\rm T}=0,\;\forall i\in[\nd(\qv) -1],
\ee 
which is equivalent to $\alphaglobal \mbf{F}_{1}^{\rm T}=\alphaglobal\mbf{F}_i^{\rm T},\forall i\in [\nd(\qv)]$.
This ensures that all users store the same number of packets, consistent with equal cache sizes $H(Z_k)\le ML,\forall  k\in [K]$.

\begin{remark}
\label{rmk:MC under equal&unequal q, PT overview}
\tit{(\ref{eq: MC}) is automatically   satisfied under equal groupings since all users belong to the same \uset.
$H(Z_k)=ML$ can  be ensured by properly choosing the packet size. 
Under unequal groupings, \eqref{eq: MC} must hold simultaneously across all pairs of \usets $(\Uc_1, \Uc_i),i\in [\nd(\qv)]$, which generally requires careful design of $\qv$ and \tx selection.}
\end{remark}

\subsubsection{\tbf{Summary of PT Workflow}}
So far, we have introduced the main ingredients of the PT \frmwk. 
The workflow is summarized \af (see also Fig. \ref{fig: PT flowchart}).
Starting with the user grouping $\qv$ \eqref{eq: user grouping vector, PT framework},
\sbf types $\{\vv_i\}_{i=1}^V$  \eqref{eq: subfile type def} and \mgrp types $\{\sv_i\}_{i=1}^S$ \eqref{eq: multicast group type def} can be defined.
The \tx selections \eqref{eq:Utx as union of unique sets, Tx selection}  across the \mgrp types determine a set of local FS vectors  $\alphabm_1,\cdots ,\alphabm_S$ \eqref{eq: FS factor}. With proper subfile \exc,  the local FS vectors are \cdntd by the vector LCM operation \eqref{eq:def a^LCM, PT frmwk} to produce a global FS vector $\alphaglobal$ that determines the overall PT 
\sbp $\Fpt$ through \eqref{eq: PT subp, PT frmwk components}.
We aim to minimize $\Fpt$ by choosing proper \ugrping  and \tx selections while ensuring the \memoconsts \eqref{eq: MC} of the users.

\subsubsection{\tbf{Two Packet Reduction Mechanisms}}
The JCM \schm is a specical instance of our PT design.
\Specly, for any user grouping $\qv$, the subfile types collectively
partition all $\binom{K}{t}$ subfiles in the JCM scheme.
By retaining all subfile types and setting the global FS vector as
\be 
\label{eq:alpha_JCM global}
\alphabm_{\rm JCM} = (t, \cdots,t),
\ee 
the PT framework recovers the JCM cache placement.
Consequently, each user caches the same fraction
$t/K=M/N$ of every file, automatically satisfying
the memory constraint.
Moreover, selecting all users as transmitters in
each multicast group reproduces the JCM delivery,
establishing JCM as a special case of the PT framework.

Treating  JCM as a baseline, two packet reduction mechanisms emerge from the PT \frmwk: 
\begin{itemize}
        \item[] \tit{1) Subfile saving \mech}: Through careful transmitter selection and the subfile-type
    exclusion procedure (see \Algo \ref{alg:exclusion_closure}), the PT framework can eliminate certain
    subfile types that are redundant in the JCM scheme while
    preserving the required delivery performance. Excluding a
    subfile type $\vv$ removes all $F(\vv)$ subfiles of that type,
    thereby reducing the number of subfiles that need to be
    further split and delivered.
    \item[] \tit{2) Further-splitting saving \mech}: 
        Unlike the JCM scheme, which uniformly splits every subfile
    into $t$ packets, the PT framework allows different subfile
    types to have different global FS factors, determined by
    transmitter selection and coordinated through the vLCM
    operation. Consequently, a retained type-$\vv$ subfile
    may require fewer than $t$ packets, providing additional
    subpacketization savings.
\end{itemize}

Together, these two mechanisms reduce the total
subpacketization from $F_{\rm JCM}=  t\binom{K}{t}=\sum_{i=1}^V \alpha_{\rm JCM}(\vv_i)F(\vv_i)  $ to 
$ 
F_{\rm PT}
=\sum_{i=1}^V
\alpha^{\rm global}(\vv_i)F(\vv_i)
$ 
while maintaining the optimal communication rate.
The packet saving can be concretely written as 
\be
\label{eq:packet saving,PT two saving mechs}
\Fjcm-\Fpt= \sum_{i=1}^V\left(t- \alphaglobalnb(\vv_i)\right)F(\vv_i),
\ee
which  reveals that the packet saving of PT
is determined by the reduction in global FS factors
relative to the JCM baseline, weighted by the number
of subfiles of each type. Hence, smaller global FS
factors yield greater packet savings, with
$\alpha^{\rm global}(\vv)=0$ corresponding to the
complete exclusion of type-$\vv$ subfiles.

\subsection{Cache Placement and Delivery Phases under PT}
\label{subsection: placement & delivery under PT, PT framework} 

Given the  file splitting described in Section~\ref{subsection: PT framework, PT framework}, we specify the cache placement and delivery phases as follows.

\subsubsection{File splitting} 

Given the \gfsv 
$\alphaglobal=[\alpha^{\rm global}(\vv)]_{\vv \in  \Vc  }$, 
let $\Vc_{\rm nex} \subseteq \Vc  $ denote the set of non-excluded subfile types.  Each file  is partitioned into 
$\Fpt=\sum_{\vv\in\Vc_{\rm nex}}\alphaglobalnb(\vv)F(\vv)$ equal-sized packets, each of size $L/F_{\rm PT}$ bits.
Specifically,
\be 
\label{eq:PT file splitting, PT framework}
W_n=\bigcup_{\vv\in \Vc_{\rm nex}}\left\{W_{n,\Tc}^{(i)}: \forall \Tc ,  \texttt{type}(
\Tc)=\vv,\forall i\in \lef[\alphaglobalnb(\vv)\rig] \right\},\;\forall n\in[N]
\ee

\subsubsection{{Cache placement phase}}

Each user $k$ stores all packets $W_{n,\Tc}^{(i)}$ \suth  $k\in \Tc$, \ie, 
\be
\label{eq: Z_k, PT framework}
Z_k = \left\{ W_{n,\Tc}^{(i)}: \forall \Tc\ni k,\forall i \in  \left[ \alphaglobalnb(\ttt{type}(\Tc)) \right]  ,\forall n\in[N]   \right\}.
\ee
It can be shown that $H(Z_k)=ML$ bits, and hence the memory constraint is satisfied. A formal proof is given in Appendix~\ref{appendix: proof of MC}.

\subsubsection{Delivery phase}
Let the demand vector be $\dv=(d_1,\cdots,d_K)$.
Consider  transmission in an active type-$\sv$
multicast group $\Sc=\cup_{i\in [N_{\rm d}(\sv)]} \hatu_i $, with involved \sbf types $\Ic=\{\vv_i \}_{i=1}^{{N}_{\rm d}(\sv)}$,  where $\vv_i$ represents the type of the \sbfs
$\{W_{n,\Sc\backslash \{k\}}\}_{k\in \widehat{\Uc}_i}
$.
Let $\Utx =\cup_{i\in \Dc_{\rm Tx}(\sv)}\hatu_i \subseteq \Sc $ be the \tx set.
For each $\vv_i \in  \Ic $, let $\alpha_\sv(\vv_i)$ and $\alpha^{\rm global}(\vv_i)$ be the local and global FS factors, respectively.
By the vLCM
condition (\ref{eq:def a^LCM, PT frmwk}), there exists a common positive integer
$z(\sv)$ such that
\begin{equation}
z(\sv)
\eqdef 
\frac{\alpha^{\rm global}(\vv_i)}
{\alpha_{\sv}(\vv_i)}
\in\mathbb Z_{>0},
\quad \forall\,\vv_i\in\Ic_{\sv} \textrm{ with } \alpha_\sv(\vv_i)>0.
\label{eq: z, PT framework}
\end{equation}
Importantly, $z(\sv)$ is independent of $\vv_i$,
ensuring that all concatenated packet contributions
within the same coded multicast message have equal length.

Each \tx $k\in \Utx$ sends  a \cmsg
\be 
\label{eq: coded message, PT delivery, PT framework}
Y_{\Sc\backslash\{k\}}^k = \bigoplus_{k^{\prime}\in \Sc\backslash\{k\}} \widetilde{W}^{(c_{k^{\prime}})}_{d_{k^{\prime}}, \Sc\backslash \{k^{\prime}\} },
\ee 
where $\widetilde{W}^{(i)}_{d_{k^{\prime}}, \Sc\backslash \{k^{\prime}\}}$ denotes a \emph{concatenated} packet defined  as
\be 
\label{eq:def concatenated packet, PT delivery}
\widetilde{W}^{(i)}_{d_{k}, \Sc\backslash \{k\}}
\eqdef \left(
W^{\left((i-1)z(\sv)+1\right) }_{d_{k}, \Sc\backslash \{k\}},
W^{((i-1)z(\sv)+2) }_{d_{k}, \Sc\backslash \{k\}},\cdots,
W^{(iz(\sv)) }_{d_{k}, \Sc\backslash \{k\}}
\right ),\;\forall k \in \Sc.
\ee 
Equivalently, the coded message can be written as
$
Y_{\Sc\backslash\{k\}}^k = \big(Y_{\Sc\backslash\{k\}}^k(i)\big)_{i=1}^{z(\sv)}
$,
whose $\ith$ entry is
\be
\label{eq: def Y(i), PT delivery}
Y_{\Sc\backslash\{k\}}^k(i)=\bigoplus_{k^{\prime}\in \Sc\backslash \{k\}  } W_{d_{k^{\prime}}, \Sc\backslash \{k^{\prime}\}}^{((c_{k^{\prime}}-1)z(\sv)+i)}.
\ee
$c_k$ is a counter associated with user $k\in \Sc$, initialized to zero and incremented each time a concatenated packet with  subscript $(d_k,\Sc \backslash \{k\})$ appears in the \cmsgs $\{Y^i_{\Sc\backslash \{i\}}     \}_{i\in \Uc_{\rm Tx}}$ for $\Sc$.\footnote{
By (\ref{eq: FS factor}), $c_k\le |\Utx|-1$ for $k\in \Utx$ and $c_k\le |\Utx|$ for $ k\in \Sc\backslash \Utx$.} The counters are reset whenever a new multicast group 
$\Sc$ is considered.

Due to concatenation, each \cmsg length $\big|Y^k_{\Sc\backslash \{k\}}\big|= z(\sv)L/F_{\rm PT}$ bits.
From the above \txn, each \tx $k\in \Utx$ can decode $z(\sv)(|\Utx|-1)$ desired packets $\big\{\widetilde{W}_{d_k,\Sc\backslash\{k\}}^{(i)} \big\}_{i=1}^{|\Uc_{\rm Tx}|-1}$,
while each  non-\tx user $k\in \Sc\backslash \Utx$ can decode $z(\sv)|\Utx|$ desired packets 
$\big\{\widetilde{W}_{d_k,\Sc\backslash\{k\}}^{(i)} \big\}_{i=1}^{|\Uc_{\rm Tx}|}$.
This procedure is applied to every multicast group of every type and is summarized in \Algo~\ref{algorithm: PT delivery, PT framework}.
Since every coded multicast message is simultaneously
useful to $t$ users and all missing packets are delivered
exactly once, the resulting scheme achieves the optimal
communication rate
$R_{\rm JCM}=N/M-1$.
\begin{algorithm}[t]
\setstretch{0.75}
\caption{PT Delivery Scheme}
\label{algorithm: PT delivery, PT framework}
\KwData{$\{Z_k\}_{k\in [K]}$, $\dv=(d_1,\cdots,d_K)$.}
\For{each \mgrp type $\sv$}{
\For{each $\Sc$  of type $\sv$}{
\;\,Set $c_k=0,\forall k\in \Sc$.\\
\For{each user $k\in \Utx$}{
User $k$ sends the \cmsg 
\be 
Y_{\Sc\backslash \{k\}   }^k = \bigoplus_{k^{\prime}\in \Sc\backslash\{k\}} 
\left(  
W^{((c_{k^{\prime}}-1)z(\sv)+1) }_{d_{k^{\prime}}, \Sc\backslash \{k^{\prime}\}},\cdots, W^{(c_{k^{\prime}}z(\sv)) }_{d_{k^{\prime}}, \Sc\backslash \{k^{\prime}\}}
\right)
\ee 
to users in $\Sc\backslash \{k\}$. Update counters $\{c_k\}_{k\in \Sc}$.
}}}
\end{algorithm}

\begin{example}[Delivery Scheme under PT]
\label{example: (9,3,2) PT first look}
Consider $(K,N,M)=(9,3,2)$ with $t=6$.
Users are partitioned into $3$ groups $\Qc_1=\{1,2,3\}$, $\Qc_2=\{4,5,6\}$, and $\Qc_3=\{7,8,9\}$ under \grpg  $\qv=(3,3,3)$.
There are $V=3$ \sbf types and $S=2$ \mgrp types (see Fig. \ref{fig:PT delivery, PT frmwk}): 
\begin{figure}[t]
\centering
\begin{subfigure}[b]{0.5\linewidth}
\centering
\includegraphics[width=\linewidth]{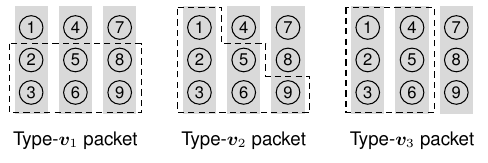}
\caption{Subfile and packet types.}
\label{fig: (9,3,2) example, subfile/packet type, PT first look}
\end{subfigure}
\hfill 
\begin{subfigure}[b]{0.39\linewidth}
\centering
\includegraphics[width=\linewidth]{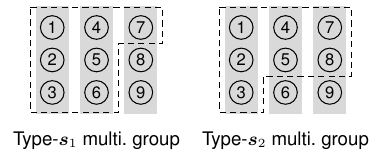}
\caption{\Mgrp types.}
\label{fig: (9,3,2) example, multicast group type, PT first look}
\end{subfigure}
\caption{\small Illustration of subfile and \mgrp types.} 
\label{fig:PT delivery, PT frmwk}
\end{figure}
\begin{align}
& \vv_1=(2,2,2),\; \vv_2=(3,2,1),\;\vv_3=(3,3,0),\notag\\
& \sv_1= (3,3,1^{\dagger} ),\; \sv_2= (3,2^{\dagger},2^{\dagger}),\; \Ic_1=\{\vv_2, \vv_3 \},\; \Ic_2=\{\vv_1, \vv_2 \},
\end{align}
where \tx selection is marked by $\dagger$ in each $\sv_i$. The local FS vectors  are $\bmalpha_1=(*,1,0) $ and $\bmalpha_2=(4,3,*) $, yielding $\alphaglobal=(4,3,0)$. Hence, type-$\vv_1$ and $\vv_2$ \sbfs  are split into $4$ and $3$
\pkts, \resp, while type-$\vv_3$ \sbfs are excluded.

\tbf{File splitting  \& cache placement.}  Each file is split into  $\Fpt=4F(\vv_1) +3F(\vv_2)=270$ \pkts,
\be 
W_n = \big\{W_{n,\Tc}^{(i)}: \Tc \in \Vc_1, i\in[4]\big\}  \cup  \big\{W_{n,\Tc}^{(i)}:\Tc\in \Vc_2, i\in[3]\big\}, \; \forall n \in[N]
\ee
User $k$ stores all packets $W_{n,\Tc}^{(i)}$ if $k\in \Tc$, i.e., 
\be 
Z_k = \lef\{ W_{n,\Tc}^{(i)} : \forall  \Tc \ni k ,\forall i \in [  \alphaglobalnb(\ttt{type}(\Tc))] , \forall n\in[3]\rig\}, \; k\in [9]
\ee 
Thus,  each user stores $216$ type-$\vv_1$ packets and $324$ type-$\vv_2$ packets, satisfying  the  memory constraint $M=\frac{216+324}{270}=2$.

\tbf{Delivery phase.}
Suppose $\dv=(d_1,\cdots,d_9)$. Consider delivery within two  \grps $\Sc_1=\{1,2,3,4,5,6,7\}$ and $\Sc_2=\{1,2,3,4,5,7,8\}$ of types $\sv_1$ and $\sv_2 $, \resp  as shown  in Fig. \ref{fig: delivery, (9,3,2) example, PT first look}.
\begin{figure}[t]
\centering
\begin{subfigure}[b]{0.47\linewidth}
\centering
\includegraphics[width=\linewidth]{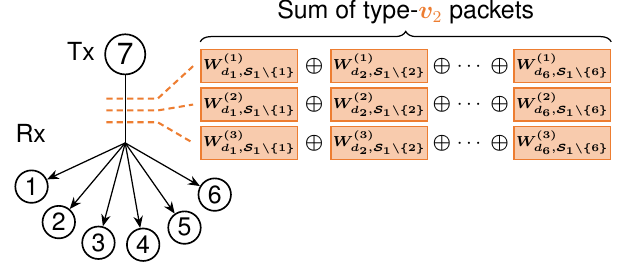}
\caption{Delivery within $\Sc_1=\{1,2,3,4,5,6,7\}$.}
\label{fig: S1 delivery, (9,3,2) example, PT first look}
\end{subfigure}
\qquad 
\begin{subfigure}[b]{0.47\linewidth}
\centering
\includegraphics[width=\linewidth]{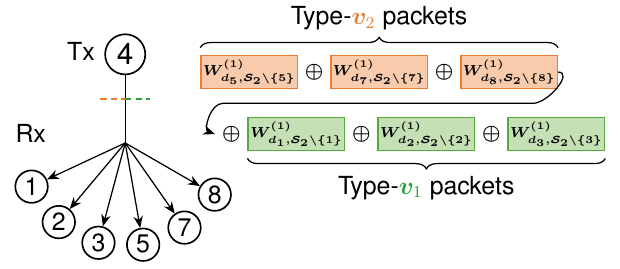}
\caption{\small Delivery within $\Sc_2=\{1,2,3,4,5,7,8\}$.}
\label{fig: S2 delivery, (9,3,2) example, PT first look}
\end{subfigure}
\caption{\small Delivery in two different \mgrps $\Sc_1$ and $\Sc_2$.} 
\label{fig: delivery, (9,3,2) example, PT first look}
\end{figure}

\tit{\textbullet \,   \Deli within $\Sc_1=\{1,\cdots, 7\}$:}
\User{7} is the only \tx and sends $3$ \cmsgs 
$$\lef\{\bigoplus\nolimits_{k\in\Sc_1\backslash \{7\}   } W_{d_k,\Sc_1\backslash \{k\}}^{(i)}\rig\}_{i=1}^3$$ to the remaining users in $\Sc_1  $ (see Fig.~\ref{fig: S1 delivery, (9,3,2) example, PT first look}).
Each \cmsg is an XOR of $ 6$ type-$\vv_2$ packets, enabling every user $k\in\Sc_1\backslash \{7\}$ to decode $3$ desired type-$\vv_2$ packets $ \big\{ W_{d_k,\Sc_1\backslash \{k\}}^{(i)}\big\}_{i=1}^3$, which constitute  a type-$\vv_2$ \sbf. 
Only type-$\vv_2$ packets are involved.
The same procedure applies to
all the remaining type-$\sv_1$ \mgrps.

\tit{\textbullet \,  \Deli within $\Sc_2=\{1,2,3,4,5,7,8\}$:} 
The \txs are $\Utx=\{4,5,7,8\}$.
The \deli involves type-$\vv_1$  packets $\big\{W_{d_k,\Sc_2\backslash \{k\}}^{(i)}:i\in [4],k\in [3]\big\}$   
and type-$\vv_2$  packets 
$\big\{W_{d_k,\Sc_2\backslash \{k\}}^{(i)}:i\in [3],k\in \Sc_2\backslash [3] \big\}$.
Defining $W^{(i)}  \eqdef \oplus_{k\in [3]} W_{d_k,\Sc_2\backslash \{k\} }^{(i)},\forall i\in[4]$, 
each \tx sends one \cmsg formed by XORing $3$ type-$\vv_1$ and $3$ type-$\vv_2$ packets (see Fig.~\ref{fig: S2 delivery, (9,3,2) example, PT first look}).
Consequently,  each user $k\in [3]$ decodes $ 4$ desired type-$\vv_1$ packets $\big\{W_{d_k,\Sc_2\backslash \{k\}}^{(i)}\big \}_{i=1}^4$, while each non-\tx user decodes $ 3 $ desired type-$\vv_2$ packets $\big\{W_{d_k,\Sc_2\backslash \{k\}}^{(i)} \big\}_{i=1}^3$. 
The same procedure applies to all type-$\sv_2$ \mgrps. \hfill $\lozenge$
\begin{table*}[ht]
\caption{Coded Messages  Associated with $\Sc_2$.}
\centering
\renewcommand{\arraystretch}{0.82}
\small \begin{tabular}{|c|c|c|} 
\hline 
\textrm{Tx $k$} & $Y^k_{\Rc_k}$ & Rx $\Rc_k$\\
 \thickhline
$4$ &  $W^{(1)}\oplus W_{d_5,\mathcal{S}_2\backslash\{5\}}^{(1)}\oplus W_{d_7,\mathcal{S}_2\backslash\{7\}}^{(1)} \oplus W_{d_8,\mathcal{S}_2\backslash\{8\}}^{(1)}$ & $\Sc_2\backslash \{4\} $ \\
 \hline
$ 5$ & $W^{(2)}\oplus W_{d_4,\mathcal{S}_2\backslash\{4\}}^{(1)}\oplus W_{d_7,\mathcal{S}_2\backslash\{7\}}^{(2)} \oplus W_{d_8,\mathcal{S}_2\backslash\{8\}}^{(2)}$ & $\Sc_2\backslash \{5\} $ \\
 \hline
$ 7$ & $W^{(3)}\oplus W_{d_4,\mathcal{S}_2\backslash\{4\}}^{(2)} \oplus W_{d_5,\mathcal{S}_2\backslash\{5\}}^{(2)} \oplus W_{d_8,\mathcal{S}_2\backslash\{8\}}^{(3)}$ & $\Sc_2\backslash \{7\} $ \\
 \hline
$ 8$ & $W^{(4)}\oplus  W_{d_4,\mathcal{S}_2\backslash\{4\}}^{(3)}\oplus W_{d_5,\mathcal{S}_2\backslash\{5\}}^{(3)}\oplus W_{d_7,\mathcal{S}_2\backslash\{7\}}^{(3)} $ & $ \Sc_2\backslash \{8\} $ \\
 \hline
\end{tabular}
\label{table: delivery schm, (9,3,2) example, PT first look}
\end{table*}
\end{example}

\begin{remark}[\Sbp  Analysis]
\tit{In Example~\ref{example: (9,3,2) PT first look}, \sbp reduction comes from both \sbf  saving and \fs saving. 
Excluding type-$\vv_3$  \sbfs   removes $tF(\vv_3)=18$  \pkts. Moreover,  type-$\vv_1$ and $\vv_2$  \sbfs are split into fewer than  $t=6 $  \pkts, yielding an additional  reduction  of $(t-4)F(\vv_1) + (t-3)F(\vv_2)=216$ packets. Overall, this results in a reduction of $234$ packets per file relative  to  the JCM  baseline $\Fjcm=t(F(\vv_1)+  F(\vv_2)+F(\vv_3))=504$.}
\end{remark}


\section{PT Design for Theorem~\ref{thm: order reduction}}
\label{section: gen ach schm, thm 1}

This section presents the PT design for Theorem~1, which partitions the users into \(K/2\) equal-size groups and achieves an order-wise subpacketization reduction over the JCM scheme. Sections
\ref{subsection: PT description, thm 1}
and \ref{subsec:subp analysis, thm 1}
give the construction and its subpacketization analysis, Section~\ref{subsection: examples, thm 1} presents illustrative examples, and Section~\ref{subsection: discussion, thm 1} discusses further implications of Theorem \ref{thm: order reduction}. Recall that \(\bar{t} \eqdef K-t\).

\subsection{General PT Design}
\label{subsection: PT description, thm 1}

For $(K,\overline{t})=(2m,2r)$ where $m\geq \overline{t},r\geq 1$,
consider the user grouping $\qv=(\underline{2}_{m})$ with $m=K/2$ groups each containing two users.
There are $V=r+1$ subfile types
\begin{align}
\label{eq: i-th subfile type, thm 1 scheme}
\vv_i & = \big(\underline{2}_{m-(r+i)+1},\underline{1}_{2(i-1)},\underline{0}_{r-i+1}\big),\; i\in [r+1]
\end{align}
Type $\vv_i$ \corrspds to a subset $\Tc$ consisting of $m-(r+i)+1$ full user groups and a single user from each of $2(i-1)$ additional  groups. There are $S=r$ \mgrp types 
\begin{align}
\label{eq: i-th multicast group type, thm 1}
\sv_i =\big(\underline{2}_{m-(r+i)+1},\underline{1}_{2i-1}^{\dagger},\underline{0}_{r-i}\big), \; \Ic_i =\{\vv_i, \vv_{i+1}\}, \; i\in [r]
\end{align}
The \txs  in $\sv_i$ are marked with $\dagger$. 
According to (\ref{eq: FS factor}), the local FS vectors can be computed as
\be 
\label{eq: local FS vector, gen schm, thm 1}
\bm{\alpha}_i \eqdef 
\left[\alpha_i(\vv_k)\right]_{k=1}^{r+1}=\left(\underline{\star}_{i-1},2(i-1),2i-1,\underline{\star}_{r-i}\right),\; i\in [r]
\ee
as shown in Table~\ref{table: FS table, thm 1 gen ach schm}\footnote{For brevity of notation, the $\star$ symbols are omitted.}.
\begin{table}[t]
\caption{FS Table for the PT Design of Theorem~\ref{thm: order reduction}.} 
\begin{center}
\renewcommand{\arraystretch}{0.8}
\begin{tabular}{|c" c| c |c| c| c |c| c| c| c|  c| c |}
\hline
& $\vv_1$ & $\vv_2$  & $\vv_3$ &$\vv_4$&$\cdots $&$\vv_i$&$\vv_{i+1}$ &$\vv_{i+2}$ &$\cdots $ &$\vv_{r}$ & $\vv_{r+1}$\\
\thickhline
$\bm{\alpha}_1  $ &$0$& $1$  &   &  &  & &   & & & & \\
\hline
$\bm{\alpha}_2$&& $2$ & $3$  &  &  &  & &  & &&  \\
\hline
$\bm{\alpha}_3   $& &  & $4 $ & $5$ &  &  &  &  && & \\
\hline
$\vdots$ & &  &   &   &$\ddots$  &    &&  & && \\
\hline
$\bm{\alpha}_i   $ & &  &   & &  & $2(i-1)$&$2i-1$  & &  &&  \\
\hline
$\bm{\alpha}_{i+1}   $ & &  &   &  &  &  &$2i$ &$2i+1$& & &  \\
\hline
$\vdots $ & &  &   &  &  &  &  &  &$\ddots$&&   \\
\hline
$\bm{\alpha}_r$ & &  &   &  &  &  &  &&&  $2(r-1)$& $2r-1$  \\
\hline 
\end{tabular}
\end{center}
\label{table: FS table, thm 1 gen ach schm}
\end{table}
The global FS vector $\alphaglobal=\left[\alphaglobalnb(\vv_i)\right]_{i=1}^{r+1}$ is then given by
\begin{eqnarray}
\label{eq: global FS factor, thm 1}
\alphaglobalnb(\vv_i)   =\left\{\begin{array}{ll}
0,&  i=1 \vspace{-.2cm}\\
2^{r-1}\prod_{k=1}^{r-1}k, &  i=2 \vspace{-.2cm} \\
\prod_{k=1}^{i-1}(2k-1)\left(\prod_{k=i-1}^{r-1}2k\right), &  i\in[3:r-1] \vspace{-.2cm}\\
2(r-1)\prod_{k=1}^{r-1}(2k-1) & i=r\vspace{-.2cm}\\
  \prod_{k=1}^{r}(2k-1)  &  i=r+1  
\end{array}\right.
\end{eqnarray}
Note that type-$\vv_1$ subfiles are excluded since $\alphaglobalnb(\vv_1)=0$. 
Hence, the \sbp is equal to
$\Fpt=\sum_{i=1}^{r+1}\alphaglobalnb(\vv_i) F(\vv_i)$, where the number of type-$\vv_i$ subfiles is given by
\be
\label{eq: F(v_i), thm 1}
F(\vv_i) =2^{2(i-1)}\binom{m}{2(i-1)}   \binom{m-2(i-1)}{r-i+1},\; i\in[r+1].
\ee

\subsection{\Sbp Analysis} 
\label{subsec:subp analysis, thm 1}
The above PT design attains an order-wise (in $K$) \sbp reduction over the \jsch, namely,  $\Fpt/\Fjcm=\Theta(1/K)$, which vanishes when  $K\to \infty$.  For finite $K$,  $\Fpt  \le  \Fjcm$ is guaranteed due to \sbf exclusion. This is proved \af.

\subsubsection{Asymptotic regime $K\to \infty$} Since the global FS factors $\{\alphaglobalnb(\vv_i)\}_{i=1}^{r+1}$
are nondecreasing in  $i$, 
the subpacketization ratio of PT relative to JCM satisfies 
\begin{align}
\label{eq: subp rartio upper bound, thm 1}
\frac{F_{\rm PT}}{F_{\rm JCM}} & =\frac{\sum_{i=1}^{r+1}\alphaglobalnb(\vv_i) F(\vv_i)}{t\binom{K}{t}} 
\le 
\frac{ \alphaglobalnb(\vv_{r+1})}{t}   \frac{\sum_{i=1}^{r+1} F(\vv_i)}{\binom{K} {t}}
= \frac{ \alphaglobalnb(\vv_{r+1})}{t},
\end{align}
where we used the identity $\binom{K}{t}=\sum_{i=1}^{V} F(\vv_i)$, which holds for any PT design  since all subfile types  collectively cover all  $t$-subsets of $[K]$.  
\Thf, for any fixed $\overline{t}$, 
\be 
\label{eq: Fpt/Fjcm ratio bound, sec. proof of thm 1}
\frac{F_{\rm PT}}{F_{\rm JCM}} \le  
\frac{ \prod_{i=1}^{\overline{t}/2}(2i-1)     }{K-\overline{t}} = \Theta\left(  \frac{1}{K} \right),
\ee 
implying an order-wise (in $K$) reduction  over JCM.
This gain  stems mainly from \fs savings: the largest FS factor 
\be 
\alphaglobalnb(\vv_{r+1})=\prod\nolimits_{k=1}^{\overline{t}/2}(2k-1)=\Theta(1) 
\ee 
does \tit{not scale with $K$} for fixed $\tbar$,  whereas the JCM FS factor  $\alphajcm= t=K-\tbar=\Theta(K)$ \tit{grows linearly with $K$}. 
Hence, the PT design \achvs an $\Theta(K)$ reduction over JCM. In addition, subfile saving is also available  due to the exclusion of type-$\vv_1$ \sbfs.
Since $\overline{t}=K(1-\mu)$ ($\mu$ is the per-user memory size), letting $K\to \infty$ with fixed $\overline{t}$ implies $\mu \to 1$. Therefore, Theorem~\ref{thm: order reduction} establishes an order-wise reduction in the large memory regime.
In fact, the vanishing ratio \sbp 
$ \Fpt/\Fjcm \overset{K\to \infty}{\to} 0 $ continues to hold even  when $\tbar$ is allowed to scale as $ \tbar=  O\left(\log_2\log_2 K\right)$,  as shown in the following lemma.
\begin{lemma}
\label{lemma: vanishing upper bound, thm 1 gen schm}
\emph{The \sbp ratio bound $\big( \prod_{i=1}^{\overline{t}/2}(2i-1) \big) /(K-\overline{t})$ on the RHS of (\ref{eq: Fpt/Fjcm ratio bound, sec. proof of thm 1}) vanishes as $K\to\infty$ if  $\overline{t}= O\left(\log_2\log_2 K\right)$.}
\end{lemma}
\begin{IEEEproof} 
See Appendix \ref{sec:proof of vanishing subp ratio, app}.
\end{IEEEproof}

\subsubsection{Finite $K$ regime}
For large values of $\overline{t}$, the upper bound
$ \frac{\alpha^{\rm  global}(\vv_{r+1})}{t} = \frac{ \prod_{i=1}^{\overline{t}/2}(2i-1) }{K-\overline{t}} $ in (\ref{eq: subp rartio upper bound, thm 1}) may exceed $1$.
Nevertheless, $\Fpt \le  \Fjcm$ can always be ensured by using a \diff \tx selection from~(\ref{eq: i-th multicast group type, thm 1}).
\Ip,  the same selection of~(\ref{eq: i-th multicast group type, thm 1}) is still applied to $\vv_1$; however, all $t+1$ users are chosen as \txs  in the  remaining   \mgrp types. This results  in the \lfsvs 
$\bmalpha_1=\left[\alpha(\vv_1),\alpha(\vv_2)\right]=(0,1)$, $\bmalpha_i=\left[\alpha(\vv_i),\alpha(\vv_{i+1})\right]=(t,t),i\ge 2$, yielding  $\alphaglobal =(0,\underline{t}_r)$.  Thus,  only subfile saving gain is available and 
\be 
\label{eq: subp ratio upper bound 1, thm 1}
\frac{F_{\rm PT}}{F_{\rm JCM}}= \frac{ \sum_{i=2}^{r+1}tF(\vv_i)       }{\sum_{i=1}^{r+1}tF(\vv_i)} = 1 -\frac{F(\vv_1)}{\sum_{i=1}^{r+1}F(\vv_i)} \overset{(\ref{eq: F(v_i), thm 1})}{=}   1  -
\frac{ \binom{K/2}{\overline{t}/2}}{\binom{K}{t}} <1.
\ee 
Combining (\ref{eq: Fpt/Fjcm ratio bound, sec. proof of thm 1}) and (\ref{eq: subp ratio upper bound 1, thm 1}), we conclude that 
\be 
\frac{F_{\rm PT}}{F_{\rm JCM}} \le \min\left\{ \frac{\prod_{i=1}^{\overline{t}/2}(2i-1)}{K-\overline{t}}, 1\right\},
\ee
completing the proof of \Thm \ref{thm: order reduction}.

\subsection{Concise Examples}
\label{subsection: examples, thm 1}

\begin{example}
\label{example 1: subsection examples, thm 1}
Consider $(K,N,M)=(6,3,2)$ with $(t,\overline{t})=(4,2)$. Let  the user assignment  be $\Qc_1=\{1,2\}$, $\Qc_2=\{3,4\}$, and $\Qc_3=\{5,6\}$ under $\qv=(2,2,2)$. The \sbf and \mgrp types are $
 \vv_1=(2,2,0),  \vv_2=(2,1,1), \sv=(2,2,1^{\dagger}), \Ic= \{\vv_1, \vv_2\}$ 
as illustrated in Fig.~\ref{fig: (6,2) example thm 1, subfile n multicast group type}.
\begin{figure}[t]
    \centering
    \includegraphics[width = 0.65\linewidth]{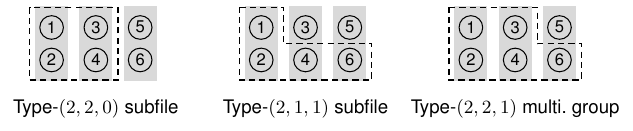}
    \caption{\small Subfile and \mgrp types for $(K,\overline{t})=(6,2)$.}
    \label{fig: (6,2) example thm 1, subfile n multicast group type}
\end{figure}
(\ref{eq: global FS factor, thm 1}) yields    
$\alphaglobal=(0,1)$, so type-$\vv_1$ subfiles are excluded and no \fs is required for type-$\vv_2$.  Hence, $\Fpt=F(\vv_2)=12$, which is only one fifth of $\Fjcm=60$. 

\tbf{File splitting \&  cache placement.}  Each file is split into $12$ \sbfs:
\be
W_n =\lef\{ W_{n,\Tc}: \forall \Tc, \type(\Tc)=(2,1,1)     \rig\}, \; n =1,2,3
\ee

\tbf{\Deli phase.} 
Consider \deli within $\Sc=\{1,2,3,4, 6\}$ with a single \tx $\Utx=\{6\}$. \User{6} sends $ Y_{  \{1,2,3,4\}   } ^6 = \bigoplus_{k\in[4]} W_{d_k, \Sc \bksl \{k\}   } $ from  which each remaining user decodes one desired (type-$\vv_2$) subfile. 
\Deli for other \mgrps  follows analogously. \hfill $\lozenge$
\end{example}

\subsection{Discussion}
\label{subsection: discussion, thm 1}

This section provides further insights into the PT design of Theorem 1, including 1) subpacketization optimality within a class of equal-grouping PT designs, 2)  nonuniqueness of PT designs achieving order-wise reduction, 3) extension to broader \param regimes, 
and 4) connection to existing methods.

\subsubsection{\Sbp optimality among equal groupings}

An interesting observation is that 
the \grpg  $\qv=(\underline{2}_{K/2})$ used in Theorem~\ref{thm: order reduction}
actually achieves the minimum \sbp among all equal-grouping PT designs, as shown in Lemma~\ref{lemma: min sbp series grouping, thm 1 gen schm}.
\begin{lemma}
\label{lemma: min sbp series grouping, thm 1 gen schm}
\emph{For $(K,\overline{t})=(mq,2)$, the  grouping $\qv=\left(\underline{2}_m\right)$ achieves the minimum \sbp among all groupings 
$ 
\big\{\qv=\big(\underline{q}_m\big): mq=K,  m,q \ge2\big\}
$. The minimum \sbp is  $F_{\rm PT}^*=\frac{K(K-2)}{2}$.}
\end{lemma}
\begin{IEEEproof} 
For a general \grpg  $\qv=\big(\underline{q}_m\big)$ with $m,q\ge 2$, there are two subfile types $ 
\vv_1=\big(\underline{q}_{m-1},q-2\big)$,   $\vv_2=\big(\underline{q}_{m-2},\underline{(q-1)}_{2}\big)
$, and one \mgrp  type 
$\sv =\big( \underline{q}_{m-1}, q-1^{\dagger}\big)$ with $\Ic=\{\vv_1,\vv_2\}$. The global FS vector is
$\alphaglobal=(q-2,q-1)$, yielding 
\begin{align}
\label{eq: F(q), lemma min sbp series grouping, thm 1 gen schm}
\Fpt(q) & =\alphaglobal [F(\vv_1), F(\vv_2)]^{\rm T}=
(q-2)F(\vv_1) + (q-1)F(\vv_2)= \frac{K(q-1)(K-2)}{2},
\end{align}
which is minimized at $q=2$ with $\Fpt(2)=\frac{K(K-2)}{2}$.
\end{IEEEproof}

\subsubsection{Nonuniqueness of PT designs}
The PT designs  achieving \ow \sbp reduction are \tit{not} unique. For instance, \Ex \ref{example: q=(4...) for t_bar=4, discussion thm 1} shows that such a  reduction is attainable   under \grpg $\qv=(\underline{4}_m)$ for $K=4m$. Such nonuniqueness 
is also demonstrated by
Examples \ref{example: discussion of Thm 1, PT design not unique} and \ref{example: example 3m equal grouping, continued} in \subseq  sections.

\begin{example}
\label{example: q=(4...) for t_bar=4, discussion thm 1}
Consider $(K,\bar{t})=(4m, 4)$ with $m\ge 5$. Under the grouping $\qv=(\underline{4}_m)$, there are $5$  \sbf types and $3$ \mgrp  types, 
\begin{align}
& \vv_1=\left(\underline{4}_{m-1},0 \right ), \;
\vv_2=\left(\underline{4}_{m-2},3,1 \right ),\;
\vv_3=\left(\underline{4}_{m-2}, 2,2 \right ),\;
\vv_4=\left(\underline{4}_{m-3},3,3,2 \right ),\;
\vv_5=\left(\underline{4}_{m-4},\underline{3}_4 \right ),\notag\\
 & \sv_1 =\big( \underline{4}_{m-1},1^{\dagger}   \big),\;
 \sv_2 =\big( \underline{4}_{m-2},3,2^{\dagger}   \big),\;
 \sv_3 =\big( \underline{4}_{m-3},\underline{3}_3^{\dagger}  \big),\notag\\
 & \Ic_1=\{\vv_1,\vv_2\},\;
 \Ic_2=\{\vv_2,\vv_3,\vv_4\},\;
 \Ic_3=\{\vv_4,\vv_5\}.
\end{align}
The \lfsvs are $\bm{\alpha}_1=(0,1, \star, \star, \star)$, $\bm{\alpha}_2=(\star,1,2, 2, \star)$, and $\bm{\alpha}_3=(\star, \star, \star,8,9)$,
yielding $\alphaglobal=(0,4,8,8,9)$. 
Hence, 
$
F_{\rm PT}(\qv) = \alphaglobal[F(\vv_1), \cdots, F(\vv_5)  ]^{\rm T} = \frac{K(K-4)(3K(K-4)+8)}{8}
$, and
\be
\frac{F_{\rm PT}(\qv)}{F_{\rm JCM}} = 
\frac{9K(K-4)+24}{ (K-1)(K-2)(K-3)} =  \Theta \lef( \frac{9}{K}\rig),
\ee 
which confirms an order-wise reduction.
By contrast, using  the grouping $\qv^{\prime}=(\underline{2}_{2m})$ from \Thm~\ref{thm: order reduction} yields   $\Fpt(\qv^{\prime})=\frac{K(K-2)^2(K-4)}{8}$, and 
\be
\frac{ \Fpt(\qv^{\prime})}{F_{\rm PT}(\qv)}= 
\frac{K(K-2)^2(K-4)}{K(K-4)(3K(K-4)) +8)} \overset{K\to \infty}{\to } \frac{1}{3}. 
\ee
Although both $\qv$ and $\qv^\prime$
achieve order-wise reductions, the \grpg
$\qv^\prime$   attains an asymptotically three-fold smaller \sbp than $\qv$, demonstrating an advantage of \Thm \ref{thm: order reduction}.
\hfill $\lozenge$
\end{example}

\subsubsection{Broader regimes---beyond even $K$ and $\tbar$}
The  PT design in \Thm \ref{thm: order reduction} requires
 both  $K$ and $\tbar$ to be even.
We relax  these  conditions  and present several  PT designs  that  apply  to odd $K$  or $\tbar$, complementing \Thm \ref{thm: order reduction}.

\tbf{Design 1.} 
First, Lemma \ref{lemma: min sbp series grouping, thm 1 gen schm} shows that for $\tbar=2$ and \tit{\arbi} $K=mq$,   an \ow  reduction is \achvb using 
$\qv=\big(\underline{q}_m\big)$,
 without requiring $K$ to be even.
\Specly, by  (\ref{eq: F(q), lemma min sbp series grouping, thm 1 gen schm}),  
\be 
\frac{\Fpt(q)}{F_{\rm JCM}}= \frac{q-1}{K-1}= \Theta \lef( \frac{q}{K}\rig), 
\ee 
implying an \ow reduction for fixed  \grp size $q$ as  $K\to \infty$.

\tbf{Design 2} (Odd $K$ and $\tbar=2$)\tbf{.} 
The next \ex  shows that for $\tbar=2$  and odd $K$, an \ow reduction can be  \achvd  using an  \tit{\uneq} \grpg.

\begin{example}[Order-wise Reduction with \Uneq  Grouping]
\label{example: discussion of Thm 1, PT design not unique} 
Consider $(K,\tbar)=(2m+1,2)$  with $m\ge 3$ and grouping $\qv=(3,\underline{2}_{m-1})$. The subfile and \mgrp types are 
\begin{align}
&\vv_1=\big(1,\underline{2}_{m-1} \big), \vv_2=\big(2,\underline{2}_{m-2},1 \big),
\vv_3=\big(3,\underline{2}_{m-3},\underline{1}_2 \big), \vv_4=\big(3,\underline{2}_{m-2},0 \big), \notag\\
& \sv_1=\big(2^{\dagger},\underline{2}_{m-1}\big),
\sv_2=\big(3,\underline{2}_{m-2}, 1^{\dagger}\big), \;
\Ic_1=\{\vv_1,\vv_2\}, \Ic_2=\{\vv_2,\vv_3,\vv_4\}. 
\end{align}
The local and global FS vectors are $\bm{\alpha}_1=(1,2, \star, \star)$, $\bm{\alpha}_2=(\star, 1,1,0)$, and $\alphaglobal=(1,2,2,0)$.
The memory constraint (\ref{eq: MC}) is satisfied since $\alphaglobal \bm{\Delta}_1^{\rm T}=0$ (see Table~\ref{table: MC table, example, nonunique design, discussion, thm 1}).
\begin{table}[t]
\caption{MC Table for Example~\ref{example: discussion of Thm 1, PT design not unique}. }
\begin{center}
\renewcommand{\arraystretch}{0.8}
\begin{tabular}{|c" c| c| c| c|}
\hline
 & $\vv_{1}$   & $\vv_{2}$ &  $\vv_{3}$ & $\vv_{4}$ \\
\thickhline
 $\mbf{F}_1$ & $1$   & $4(m-1)$ &  $2(m-1)(m-2)$ & $m-1$ \\
\hline
$\mbf{F}_2$ & $3$   & $3(2m-3)$ &  $2(m-2)^2$ & $m-2$ \\
\hline
$\bm{\Delta}_1$ & $2$   & $2m-5$ &  $-2(m-2)$ & $-1$ \\
\hline
\end{tabular}
\end{center}\label{table: MC table, example, nonunique design, discussion, thm 1}
\end{table}
Hence, 
\begin{align}
 F_{\rm PT}  =\alphaglobal\Fm^{\rm T}
 =(1,2,2,0)\left[  3,6(m-1),2(m-1)(m-2),m-1 \right]^{\rm T} =K(K-2),
\end{align} 
and  ${F_{\rm PT}}/{F_{\rm JCM}}=2/(K-1)$, 
implying an order-wise reduction. \hfill $\lozenge$
\end{example}

\tbf{Design 3.} (Odd $\tbar$)\tbf{.}
\Ex \ref{example: example 3m equal grouping, continued} shows  that an \ow reduction is \achvb for $\tbar=3$ using the  equal \grpg $\qv=(\udl{3}_m)$.
The same approach extends to larger odd $\tbar$, which we omit for brevity.

\begin{example}[PT Design for $\tbar=3$]
\label{example: example 3m equal grouping, continued}
Consider $(K, \tbar)=(3m, 3)$ with $m\ge 3$, where $K$ can  be  even or odd. With $\qv=(\udl{3}_m)$, the  \sbf and \mgrp types are 
\begin{align}
 & \vv_1=( \underline{3}_{m-3} ,\underline{2}_3     ), 
\vv_2=( \underline{3}_{m-2} ,2,1  ), \vv_3= ( \underline{3}_{m-1} ,0  ), \notag\\
& \sv_1  =\big(  \underline{3}_{m-1}, 1^{\dagger} \big), \sv_2 =\big(  \underline{3}_{m-2}, \underline{2}_{2}^{\dagger} \big), \; \Ic_1=\{\vv_2, \vv_3\},   \Ic_2=\{\vv_1, \vv_3\}. 
\end{align}
The local and global \fs vectors are $\bm{\alpha}_1=(\star, 1,0)$,  $\bm{\alpha}_2=(4,3,\star)$, and $\alphaglobal=(4,3,0)$. Hence,
\begin{align}
F_{\rm PT} &= (4,3,0)\left[ F(\vv_1), F(\vv_2),F(\vv_3)\right]^{\rm T}\notag\\
&  = (4,3,0)\left[  \binom{m}{3}\binom{3}{1}^3, \binom{m}{1}\binom{3}{1}\binom{m-1}{1}\binom{3}{2},\binom{m}{1}\binom{3}{0} \right]^{\rm T} =  \frac{2K(K-3)(K-3/2)}{3},
\end{align}
and 
${F_{\rm PT}}/{F_{\rm JCM}}= \frac{4(K-3/2)}{(K-1)(K-2)}=\Theta(4/K)$, indicating an \owr. \hfill $\lozenge$ 
\end{example}

\subsubsection{Connection to Existing  DPDA Designs~\cite{wang2017placement}}
\label{subsubsection: connection to DPDA, discussion thm 1}
Wang \etal~\cite{wang2017placement} showed that lower-than-JCM \sbps  are \achvb  for $t=2,K-2$ under the DPDA \frmwk while ensuring the  optimal rate:
i) For $t=2$, $F_{\rm DPDA}=K^2/4$ (Lemma 3, \cite{wang2017placement}); 
ii) For $t=K-2$, $\Fdpda=K(K-2)$ for odd $K$, and $\Fdpda= \frac{ K(K-2)}{2}$ for even $K$ (Lemma 4, \cite{wang2017placement}). 
These results  are also \achvb by PT and
the corresponding designs are given in
Appendix \ref{app:PT for DPDA}.

\section{PT Design for \Thm~\ref{thm: constant reduction}}
\label{section: gen ach schm, thm 2}

This section presents the  PT design for \Thm~\ref{thm: constant reduction} under a two-\grp equal user \grpg $\qv=(K/2,K/2)$, where $K$ is even.
The design \achvs  a \sbp at most half that of the JCM scheme. 
Unlike \Thm \ref{thm: order reduction}, which 
only applies to the large-\memo regime $ \mu \ge 1/2$,\footnote{From Section \ref{subsection: PT description, thm 1}, for $K=2m$, it is required that $\tbar \le  m$, \ie,  $K-t \le K/2 \Rightarrow  \mu \ge  1/2$.} this design  covers the full memory range $ 0< \mu <1$. 
Section~\ref{subsection: general PT design, thm 2} gives the construction and its  \sbp analysis, 
Section~\ref{subsec:examples, thm 2 gen PT design} presents examples, and 
Section~\ref{subsection: discussion of thm 2} provides further 
discussion.

\subsection{General PT Design}
\label{subsection: general PT design, thm 2}

Let $(K,t)=(2q, 2r)$ with $r \le q$, and consider the  equal grouping $\qv=(q,q)$.  
The number of subfile and \mgrp types are given by
\be 
\label{eq: (V,S) general ach. scheme thm 2}
V=\frac{t}{2} +  1 -\max\left\{t-\frac{K}{2},0\right\},\quad 
S=\frac{t}{2}   + 1 -\max\left\{t-\frac{K}{2}+1,0\right\},
\ee 
\soth $V=S=t/2+1$ if $ t\le K/2-1$, and $V=S+1=\frac{K-t}{2}  +1$  if $t\ge K/2$.
The \sbf and \mgrp types are given by
\begin{subequations}
\label{eq: subfile and multi. group types, thm 2}
\begin{align}
\vv_i &=(t-X-i+1, X+i-1),\;  i\in [V]
\label{eq: i-th subfile type, thm 2} \\
\sv_i &=\left(t-Y-i+2, Y+i-1\right), \; i\in [S]\label{eq: j-th multi. group type, thm 2} 
\end{align}
\end{subequations}
where $X\eqdef \max\{t-K/2,0\}$ and $Y\eqdef \max\{t+1-K/2,0\}$.
We consider the cases $t\le K/2-1$ and $t\ge K/2$ separately.

\begin{table*}[t]\caption{FS Table for Case 1} 
\label{table: FS table gen. ach. case 1, thm 2} 
\begin{center}
\renewcommand{\arraystretch}{0.8}
\begin{tabular}{|c" c| c| c |c |c |c |c| c| c| c|c |}
\hline
& $\vv_1$ & $\vv_2$ & $\vv_3$ &$\cdots $ &$\vv_{i-1}$ &$\vv_i$ &$\vv_{i+1}$ &$\cdots $ &$\vv_{r-1}$&$\vv_r$ & $\vv_{r+1}$\\
\thickhline
  $\bmalpha_1$ &$t$&   &   &  &  & &   & & & & \\
 \hline
   $\bmalpha_2$& $0$ & 1 &   &  &  &  & &  & &&  \\
 \hline
   $\bmalpha_3$& & 1 & 2  &  &  &  &  &  && & \\
 \hline
 $\vdots$ & &  &   & $\ddots$  &  &    &&  & && \\
 \hline
  $\bmalpha_i$ & &  &   & & $i-2$ & $i-1$&  & &  &&  \\
 \hline
   $\bmalpha_{i+1}$ & &  &   &  &  &$i-1$  & $i$ && & &  \\
 \hline
   $\vdots $ & &  &   &  &  &  &  & $\ddots$ &&&   \\
 \hline
   $\bmalpha_{r}$ & &  &   &  &  &  &  &&$r-2$&  $r-1$&  \\
 \hline 
  $\bmalpha_{r+1}$ & &  &   &  &  &  &  &  &&$r-1$ &$r$\\
  \thickhline
    $\alphaglobal$ &0 & 1 & 2  & $\cdots$ & $i-2$ &$i-1$ &$i$  &$\cdots$ &$r-2$ & $r-1$ &$r$  \\
    \hline
\end{tabular}
\end{center} 
\end{table*}

\subsubsection{{Case 1:} $t\le K/2-1$}
\Itc, $V=S=r+1$. The subfile and \mgrp types are  
$\vv_i =\left( t-i+1, i-1 \right)$ and
$\sv_i=(t-i+2, i-1),i\in[r+1]$. \Tx selection is given by 
\begin{align}
\label{eq:tx select, case 1, thm 2}
 & \sv_1=\big(t+1,0 \big), \;
 \sv_i =\big(  t-i+2, (i-1)^{\dagger} \big),\;\forall i\in[2:r+1]
\end{align}
with involved \sbf types  $\Ic_1=\{\vv_1\},  \Ic_i=\{\vv_{i-1},\vv_i\},i\in [2:r+1]$.
The \lfsvs are
$\bmalpha_1=(t, \udl{\star}_{V-1})$ and
$\bmalpha_i=(\udl{\star}_{i-2}, i-2, i-1, \udl{\star}_{V-i} ), \forall i\in[2:V]$ (see  Table~\ref{table: FS table gen. ach. case 1, thm 2}), yielding 
\be
\label{eq: global FS vector, case 1, thm 2}
\alphaglobal=(0,1,2,\cdots,r-1,r).
\ee
Thus, type-$\vv_1$ \sbfs are excluded, and all remaining types have FS factors  at most $r=t/2$, achieving both \sbf saving and \fs saving gains.
The resulting \sbp is
\begin{align}
\Fpt & = \alphaglobal \lef(\big[F(\vv_i)\big]_{i=1}^{r+1}\rig)^{\rm T}=\sum_{i=1}^{r+1}(i-1)F(\vv_i),
\end{align}
where $F(\vv_i) =2\binom{q}{i-1}\binom{q}{t-i+1},i\in [r]$ and $F(\vv_{r+1})=\binom{q}{r}^2$.
\Conseqtly,
\be
\frac{F_{\rm PT}}{F_{\rm JCM}}
=\frac{\sum_{i=1}^{r+1}(i-1)F(\vv_i)}{t  \binom{K}{t}   }
=\frac{\sum_{i=1}^{r+1}(i-1)F(\vv_i)}{t\sum_{i=1}^{r+1}F(\vv_i)}
\le 
r/t=1/2  \Rightarrow \frac{F_{\rm PT}}{F_{\rm JCM}} \le  \frac{1}{2},
\ee
where  the actual ratio is typically smaller than  $1/2$ and increasing with $t$ (see Fig.~\ref{fig: actual sub. ratio, case 1, thm 2}). 
\begin{figure*}[ht]
\centering
\begin{subfigure}[b]{0.45\textwidth}
\centering
\includegraphics[width=\linewidth]{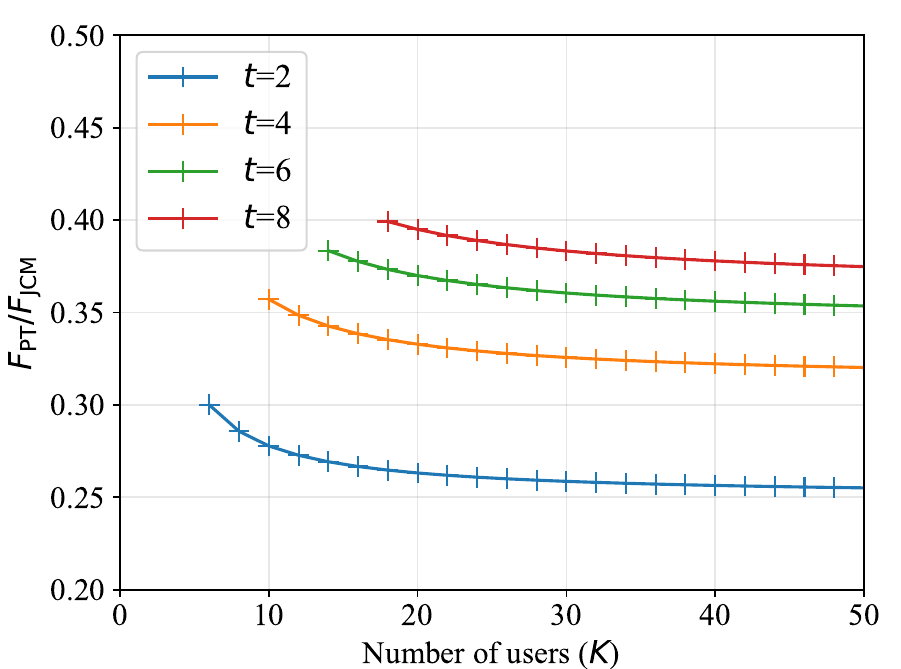}
\caption{\Sbp ratios for $t\in \{2,4,6,8\}$, Case 1.}
\label{fig: actual sub. ratio, case 1, thm 2}
\end{subfigure}
\begin{subfigure}[b]{0.45\textwidth}
\centering
\includegraphics[width=\linewidth]{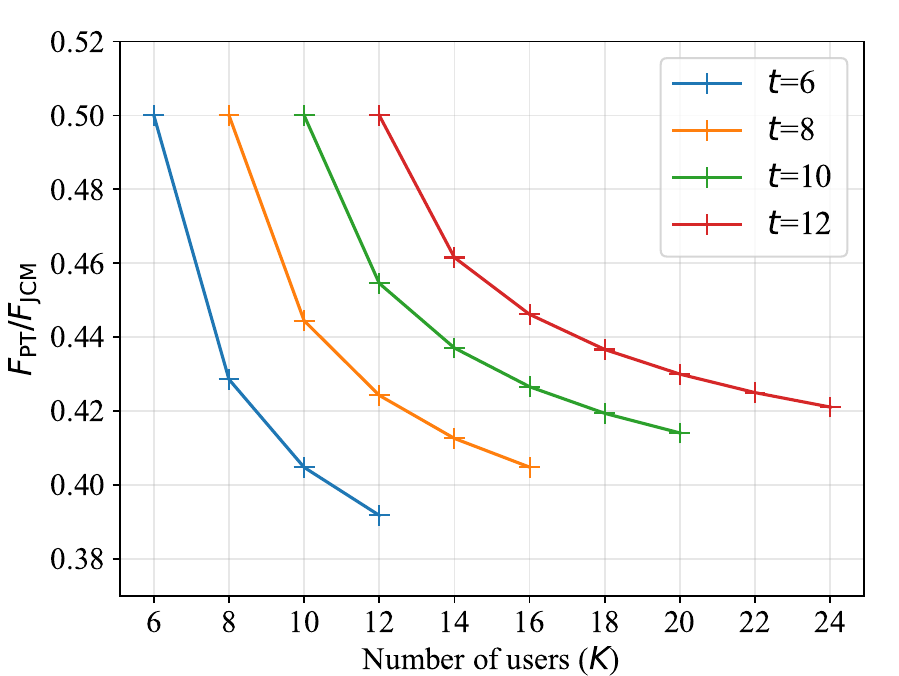}
\caption{\Sbp ratios for $t\in \{6,8,10,12\}$, Case 2.}
\label{fig: actual sub. ratio, case 2, thm 2}
\end{subfigure}
\caption{\Sbp ratio  $F_{\rm PT}/F_{\rm JCM}$  vs. $K$.} 
\label{fig: actual sub. ratio, thm 2 }
\end{figure*}

\subsubsection{{Case 2:} $t\ge K/2$}
In this case, 
$V=S+1=(K-t)/2+1$.
The \sbf  and \mgrp types are given in (\ref{eq: subfile and multi. group types, thm 2}) where
$x= t-K/2$, $y=t+1-K/2$.  
The transmitter selection is 
\be
\label{eq:tx select, case 2, thm 2}
\sv_i=\big( t-y-i+2,(y+i-1)^{\dagger} \big),\;  i\in\left[\frac{K-t}{2} \right]
\ee 
with  \istss $\Ic_i=\{\vv_i,\vv_{i+1}\},i\in [(K-t)/2]$.
\begin{table*}[ht]
\caption{FS Table for Case 2.} 
\begin{center}
\renewcommand{\arraystretch}{0.8}
\begin{tabular}{|c" c| c| c |c |c |c | c| c| c|c |}
\hline
& $\vv_1$ & $\vv_2$ & $\vv_3$ &$\cdots $ &$\vv_j$ &$\vv_{j+1}$ &$\cdots $ &$\vv_{\frac{K}{2}-r-1}$&$\vv_{\frac{K}{2}-r}$ & $\vv_{\frac{K}{2}-r+1}$\\
 \thickhline
  $\bmalpha_1$ &$y-1$& $y$  &   &  &  &    & & & & \\
 \hline
   $\bmalpha_2$& & $y$ &$y+1$   &  &  &   &  & &&  \\
 \hline
 $\vdots$ & &  &   & $\ddots$  &  &    &  & && \\
 \hline
  $\bmalpha_j$ & &  &   & & $y+j-2$ &   & &  &&  \\
 \hline
   $\bmalpha_{j+1}$ & &  &   &   &$y+j-2$  & $y+j-1$ && & &  \\
 \hline
   $\vdots $ & &  &   &  &  &  &   $\ddots$ &&&   \\
 \hline
   $\bmalpha_{\frac{K}{2}-r-1 }$ & &  &   &  &  &    &&$r-2$& $r-1$&  \\
 \hline 
  $\bmalpha_{ \frac{K}{2}-r}$ & &  &   &  &    &  &  &&$r-1$ &$r$\\
  \thickhline
    $\alphaglobal$ &$y-1$ & $y$   & $y+1$ & $\cdots$ &$y+j-2$ &$y+j-1$  &$\cdots$ &$r-2$ & $r-1$ &$r$  \\
    \hline
\end{tabular}
\end{center}
\label{table: FS table gen. ach. case 2, thm 2}
\end{table*}
The  \lfsvs are given in Table~\ref{table: FS table gen. ach. case 2, thm 2}, from which the \gfsv can be computed as 
\be
\label{eq: global FS vector, case 2, thm 2}
\alphaglobal=(y-1,y,y+1,\cdots,r-1,r),
\ee 
where $y=t+1-K/2$. 
Thus, 
\be 
\Fpt=\sum_{i=1}^{K/2-r+1}\lef(t+i-1-\frac{K}{2} \rig)  F(\vv_i),
\ee
where
$F(\vv_i) =2\binom{q}{t-q+i-1} \binom{q}{q-i+1},  i\in[q-r]$ and $F(\vv_{q-r+1})=\binom{q}{r}^2$. 
Note that $y\ge 1$ since $t \ge K/2$. When $t = K/2$, type-$\vv_1$ is excluded because $\alphaglobalnb(\vv_1) = y - 1 = 0$. As the largest \gfsf is $r = t/2$, it follows that $\Fpt \le \Fjcm/2$, similarly to Case 1. A demonstration of the \sbp ratios for various $t$ is given in Fig.~\ref{fig: actual sub. ratio, case 2, thm 2}.

Combining both cases, we conclude that $\Fpt/\Fjcm\le 1/2$ for all $(K,t)=(2q,2r)$ with $r\le q$,  thereby completing the proof of Theorem~\ref{thm: constant reduction}.

\begin{remark}
\tit{Comparing \tx selections in (\ref{eq:tx select, case 1, thm 2}) and (\ref{eq:tx select, case 2, thm 2}),
the two cases are distinguished because the transmitter is always chosen from the smaller user group in each type-$\sv_i$  \mgrp, which yields smaller FS factors according to (\ref{eq: FS factor}).}
\end{remark}

\subsection{Concise Examples}
\label{subsec:examples, thm 2 gen PT design}

\begin{example}[$t=2$]
Consider $(K,t)=(2q,2 )$ with $q\ge 3$ and \grpg $\qv=(q,q)$,  extending \Ex~\ref{example: (4,2,1) PT first look} in Section \ref{subsec: a first look at PT, problem description}.
The  \sbf and \mgrp types are   $\vv_1=(2,0), \vv_2=(1,1),\sv_1=(3^\dagger,0)$, and $\sv_2=(2,1^\dagger)$, yielding
$\alphaglobal=(0,1)$ and $\Fpt=(0,1)[F(\vv_1),F(\vv_2)]=K^2/4< \Fjcm/2= \frac{K(K-1)}{2}$.\hfill $\lozenge$
\end{example}

\begin{example}[$t=4$]
Let $t=4$ and $\qv=
(K/2,K/2)$ with $K\ge 10$ so  that the PT design in Case 1 applies.
Specifically, the \sbf  and \mgrp types are 
$\vv_1 =(4,0)$, $\vv_2 =(3,1)$, $\vv_3 =(2,2)$,
$\sv_1 =(5,0)$, $\sv_2 = (4,1^{\dagger} )$, and
$\sv_3 = (3,2^{\dagger} )$. 
The \gfsv is $\alphaglobal=(0,1,2)$ so that type-$\vv_1$  is excluded. 
Since $F(\vv_1)=2\binom{K/2}{4}$, $F(\vv_2)=2\binom{K/2}{3}\binom{K/2}{1}$, and $F(\vv_3)=\binom{K/2}{2}^2$,  
\be
F_{\rm PT}=(0,1,2)\left[F(\vv_1),F(\vv_2),F(\vv_3)\right]^{\rm T}=\frac{K^2(K-2)(5K-14)}{96}.
\ee
\Conseqtly,
\be 
\frac{F_{\rm PT}}{F_{\rm JCM}}=\frac{0.3125K(K-2.8)}{(K-1)(K-3)}\leq 0.3571,\;\forall K\geq 10
\ee 
\ie, the PT \sbp is about one-third of the \jsch asymptotically.
\hfill $\lozenge$
\end{example}

\subsection{Discussion}
\label{subsection: discussion of thm 2}

\subsubsection{The Case of Odd $t$ (Even $K$)}
\label{subsec:the case of odd t, discussion thm 2}
Let $K$ be even.  For odd $t$, the \morethanhalf \red is unachievable  because no global FS vector can have all entries bounded by $t/2$ (cf. (\ref{eq: global FS vector, case 1, thm 2}), (\ref{eq: global FS vector, case 2, thm 2})). The obstruction is the multicast group type 
$ \sv=\lef( \frac{t+1}{2},   \frac{t+1}{2}   \rig)$, which forces all users in $\sv$ to transmit, yielding a local FS factor $\alpha=t$ for the involved subfile type 
$ \vv=\lef( \frac{t+1}{2},   \frac{t}{2}   \rig)$. Moreover, due  to the misalignment\footnote{In Tables \ref{table: FS table gen. ach. case 1, thm 2} and \ref{table: FS table gen. ach. case 2, thm 2}, the  local FS factors are perfectly aligned (\ie, identical) across 
each column of the FS table. \Conseqtly,  the global FS factors coincide  with the local ones, and all entries are bounded by $t/2$.}  of  the local FS factors, the global FS factors are amplified by the vector LCM operation  and may exceed $t$, nullifying the more-than-half reduction achievable for even $t$.
Nevertheless, constant-factor reduction remains \achvb for odd $t$ under $\qv=(K/2, K/2)$; the next example shows a $1/4$ reduction for $t=3$.

\begin{example}[$t=3$]
\label{example: dicussion of thm 2}
Let $ (K,t)=(2q, 3)$  with $q\ge 4$ and   $\qv=(q,q)$. The subfile and \mgrp types are  
\begin{align}
 & \vv_1 =(3,0),\vv_2=(2,1),\;
 \sv_1=(4,0), \sv_2=(3,1^{\dagger}),\sv_3=(2^{\dagger},2^{\dagger}),
\end{align}
with $ \Ic_1=\{\vv_1\},\Ic_2=\{\vv_1,\vv_2\}$, and $\Ic_3=\{\vv_2\}$.
The \lfsvs are $\alphabm_1=(3, \star), \alphabm_2=(0, 1)$, and $\alphabm_3=(\star, 3)$, yielding $\alphaglobal=(0,3)$.
Hence, $F_{\rm PT}=(0,3)\left[F(\vv_1),F(\vv_2)\right]^{\rm T}
=\frac{3K^2(K-2)}{8}$. \Conseqtly, 
\be 
\frac{F_{\rm PT}}{F_{\rm JCM}} = \frac{3}{4}\cdot \frac{K}{K-1}\le \frac{6}{7}, \;\forall K\ge 8, \quad \lim_{K\to \infty}\frac{F_{\rm PT}}{F_{\rm JCM}} =\frac{3}{4},
\ee 
 suggesting  a $1/4$ asymptotic \red.
\hfill $\lozenge$
\end{example}

\subsubsection{The Case of Odd $K$}
\label{subsec:the case of odd K, discussion thm 2}
For general odd $K=mq$ with $ m \ge  t+1,q \ge  t+1 $, the \PTds proposed in \Thm \ref{thm: PT design with only subfile saving gain} can be applied to \achv constant-factor \reds. 
However, the more-than-half \red is no longer guaranteed.

\section{PT Design for \Thm~\ref{thm: PT design with only subfile saving gain}}
\label{section: gen ach schm, thm 3}

\subsection{General PT Design}
\label{subsec:gen PT design,thm 3}
Suppose $K=mq$ with $t+1 \le \min \{m,q\}$, implying  $ \mu \le \min\{ 1/q, {q}/{K}\}-{1}/{K} $. Consider the $m$-\grp equal \grpg $\qv=\big( \underline{q}_m\big)$. 
Let $P(n)$ denote the {partition function}\footnote{In number theory, a \tit{partition} of a positive integer $n$ is  a way of writing $n$ as a sum of positive integers, where the order of summands is irrelevant. The partition function $P(n)$ counts the number of distinct partitions of $n$.
\Fex, $P(4)=5$, \corrspdg to 
$4,3+1,2+2,2+1+1$, and $1+1+1+1$.} of $n \in \mathbb{Z}_{> 0}$. 
Then, there are $V=P(t)$ subfile types and $S=P(t+1)$ multicast group types, 
\begin{align}
& \vv_1=\left(t,\underline{0}_{m-1}\right), \vv_2=\left(t-1,1,\underline{0}_{m-2}\right),\cdots, \vv_{P(t)}=\left(\underline{1}_{t},\underline{0}_{m-t}\right),\\
& \sv_1=\left(t+1,\underline{0}_{m-1}\right),
\sv_2=\left(t,1,\underline{0}_{m-2}\right)\cdots, 
\sv_{P(t+1)}=\left(\underline{1}_{t+1}, \underline{0}_{m-t-1}\right),
\end{align}
with \invod \sbf types $\Ic_1=\{\vv_1\}$, $\Ic_2=\{\vv_1,\vv_2\},\cdots,\Ic_{P(t+1)}=\{\vv_{P(t)}\}$. 
Note that type $\vv_1$ is only \invod in $\Ic_1$ and $\Ic_2$.
The \tx selection is \af:
\begin{align}
\sv_1  &=\big((t+1)^{\dagger},\underline{0}_{m-1} \big ) \Rightarrow \alphabm_1=(t, \underline{ \star}_{P(t)-1}     ),  \notag\\
\sv_2  &  =\big(t,1^{\dagger},\underline{0}_{m-2} \big ) \Rightarrow \alphabm_2=(0,1, \underline{ \star}_{P(t)-2} ), 
\end{align}
and all $t+1$ users in each $\sv_i,\forall  i\in [3:P(t+1)]$ are chosen as \txs, leading to $  \alphabm_i=(\underline{t}_{|\Ic_i|})$ ($\star$ ignored). 
The \gfsv  is  
\be 
\alphaglobal= \big(0, \underline{t}_{P(t)-1}\big),
\ee
so that type-$\vv_1$ \sbfs are excluded while all other types have a uniform FS factor of $t$. 
The resulting \sbp is
\begin{subequations}
\label{eq:Fpt, thm 3, gen scheme}
\begin{align}
F_{\rm PT} & =\alphaglobal \Fm^{\rm T}=  \left(\bm{\alpha}_{\rm JCM} - \left(\bm{\alpha}_{\rm JCM}-\alphaglobal  \right)    \right) \mathbf{F}^{\rm T}\\
& = \bm{\alpha}_{\rm JCM}\mathbf{F}^{\rm T}-
\big(t,\underline{0}_{P(t)-1}\big)\mathbf{F}^{\rm T}\\
& = F_{\rm JCM} - tF(\vv_1)\\
& = t\binom{K}{t} -  mt   \binom{q}{t}, 
\end{align}
\end{subequations}
where $\bm{\alpha}_{\rm JCM}=(\underline{t}_{P(t)} )$ denotes the  JCM FS vector.
Thus, 
\be 
\label{eq: Fpt/Fjcm exact/asymp, thm 3}
\frac{F_{\rm PT}}{F_{\rm JCM}} = 1-
\frac{m\prod_{i=0}^{t-1}\left(q-i\right)!}{\prod_{i=0}^{t-1}(mq-i)}, \quad \lim_{q\to \infty}\frac{F_{\rm PT}}{F_{\rm JCM}} =1-\frac{1}{m^{t-1}},
\ee
completing the proof of \Thm~\ref{thm: PT design with only subfile saving gain}.

\subsection{Discussion}
\label{subsec: discussion of Thm 3}

\subsubsection{Broader coverage regime} 
The asymptotic reduction gain $1/m^{t-1}$ in (\ref{eq: Fpt/Fjcm exact/asymp, thm 3}) vanishes  when $m \to \infty$ with $q$ and $t$ fixed. This occurs \bcuz the \sbp \red in (\ref{eq:Fpt, thm 3, gen scheme}) arises solely from excluding \sbfs of type $\vv_1=(t,\underline{0}_{m-1})$, whose number becomes negligible when $m\to \infty$ (with $q$ fixed), \ie, 
\be   
\lim_{m\to\infty} \frac{F(\vv_1)}{\sum_{i=1}^{P(t)}F(\vv_i)  } =\lim_{m\to\infty}  \frac{m\binom{q}{t} }{\binom{mq}{t}} =0.
\ee 
Nevertheless, the key \advtg  of the \PTd in \Thm \ref{thm: PT design with only subfile saving gain} lies in its \tit{genericity}; that is,
it applies to  almost \tit{any} pair $(K,t)$, beyond the coverage regimes of both \Thms \ref{thm: order reduction} and \ref{thm: constant reduction}.
While the \sbp  \red is weaker, it complements \Thms \ref{thm: order reduction} and \ref{thm: constant reduction} for parameter pairs they do not cover.

\subsubsection{Minimal \sbp among equal-grouping designs}
It is interesting to observe that $\qv=(K/2,K/2)$ actually \achvs the minimum \sbp among all equal-grouping designs
$\big\{\big( \underline{K/m}_m\big)\big\}_{m=2}^K$. 
\Specly, under $\qv=\big(\underline{q}_{m}\big)$, 
the PT design in Section \ref{subsec:gen PT design,thm 3} \achvs  \pkt \red 
\be 
\label{eq: m=2 maximizes trial PT design saving}
F_{\rm JCM}-F_{\rm PT} = mt \binom{q}{t} = \frac{K}{(t-1)!}\prod_{i=1}^{t-1}(q-i)\overset{\trm{(a)}}{\le}
\frac{K}{(t-1)!}\prod_{i=1}^{t-1}\left({K}/{2}-1  \right),
\ee
where $\trm{(a)}$ follows from $q\le K/2$ since $m\ge 2$. \Conseqtly, $\qv=(K/2,K/2)$  yields the maximal packet \red relative to $\Fjcm$.

\section{Concluding Remarks}
\label{section: conclusion}

In this paper, we studied finite-length D2D coded caching with the goal of preserving the optimal communication rate. We proposed a packet type (PT)-based design framework that uses user grouping to induce a geometric classification of subfiles, packets, and multicast groups. This type-based viewpoint enables structured asymmetry in file splitting and multicast delivery. In particular, PT reduces subpacketization through two complementary mechanisms: subfile saving, by excluding redundant subfile types, and further-splitting saving, by assigning type-dependent splitting factors through transmitter selection. Based on this framework, we constructed several classes of rate-optimal D2D coded caching schemes with strictly smaller subpacketization than the JCM scheme, achieving either order-wise or constant-factor reductions depending on the system parameters. 
These results show that, unlike in shared-link coded caching, the symmetric subpacketization structure of JCM is not necessary for rate-optimal D2D caching. This reveals a fundamental structural distinction between the two models and provides a systematic methodology for finite-length D2D coded caching design.

For future work, several directions merit further investigation. First, although
the PT framework provides a systematic approach to reducing
subpacketization, developing explicit and scalable constructions
for arbitrary $K$ and $t$ remains an open challenge. In
particular, jointly optimizing user grouping, transmitter
selection, subfile-type exclusion, and global FS factors
is a promising direction toward general low-subpacketization
designs.
Second, the current framework preserves the optimal
communication rate by enforcing the optimal DoF condition
in every multicast transmission. Relaxing this requirement
may enable more aggressive transmitter selection and
subfile exclusion, potentially yielding further
subpacketization reductions at the cost of a moderate
rate increase. Characterizing this rate--subpacketization
tradeoff is an interesting extension.
Finally, establishing fundamental lower bounds on the
subpacketization required to achieve the optimal
communication rate remains an important theoretical
problem. Such bounds would help determine whether
the reductions achieved by PT constructions are
fundamentally optimal and identify the potential
for further improvements.

\section*{Acknowledgment}
The authors  would like to thank Prof. Antti T\"olli and Prof. Petros Elia for their valuable discussions and suggestions during the preparation of this paper.

\appendices

\section{Proof of Memory Constraint}
\label{appendix: proof of MC}

Recall that $F_i(\vv)$ denotes the number of type-$\vv$ \sbfs stored by user $i$.
Without loss of generality,  consider an equal grouping  $\qv=(\underline{q}_m)$ (i.e., $K=qm$) and
$\vv=(v_1,\cdots,v_m)$ where {$v_1  \ge \cdots\ge v_m> 0$ (\ie, there are $m$ unique sets in $\vv$)}.
Suppose $i\in \Qc_{\tilde{i}}$ for some $\tilde{i}\in [m]$.  
Define $\Vc_{i}(j)\eqdef\{\Tc:\texttt{type}(\Tc)=\vv, |\Tc\cap \Qc_{\tilde{i}}|=v_j\} $ where
$ 
|\Vc_{i}(j)|=\binom{q-1}{v_j-1}\prod_{j^{\prime}\ne j}\binom{q}{v_{j^{\prime}}}(m-1)!
$.
Since $ \Qc_{\tilde{i}}$ could correspond to any $v_j,j\in [m]$, the number of type-$\vv$ subfiles stored by user $i$ is equal to 
$ 
F_i(\vv) =\sum_{j=1}^m|\Vc_{i}(j)|
$, \ie, 
\begin{subequations}
\label{eq: Z verify, PT framework}
\begin{align}
F_i(\vv) &= \sum_{j=1}^m \binom{q-1}{v_j-1}\prod_{j^{\prime}\in[m]\backslash \{j\} }\binom{q}{v_{j^{\prime}}}(m-1)!\\
& = \sum_{j=1}^m\frac{v_j}{q}\binom{q}{v_j}   \prod_{j^{\prime}\in[m]\backslash \{j\} }\binom{q}{v_{j^{\prime}}}(m-1)!\\
&=\frac{m(m-1)!}{K}\sum_{j=1}^m v_j \binom{q}{v_j}   \prod_{j^{\prime}\in[m]\backslash \{j\} }\binom{q}{v_{j^{\prime}}}\\
&=\frac{m!}{K}\sum_{j=1}^m v_j\prod_{j^{\prime}=1}^m\binom{q}{v_{j^{\prime}}}\\
&= \frac{\sum_{j=1}^m v_j   }{K} m!\prod_{j}^m\binom{q}{v_j}\\
&= \lef(\sum_{j=1}^m v_j\rig) \frac{F(\vv)}{K}   \label{eq: step 0, Z verify, PT framework}\\
&=   {tF(\vv)}/{K} ,\label{eq: step 1, Z verify, PT framework}
\end{align}
\end{subequations}
where (\ref{eq: step 0, Z verify, PT framework}) is due to $F(\vv)=m!\prod_{j=1}^m\binom{q}{v_j}$; (\ref{eq: step 1, Z verify, PT framework}) is due to $ \sum_{j=1}^m v_j =t$. (\ref{eq: Z verify, PT framework}) shows that $F_i(\vv)=(t/K)F(\vv),\forall \vv\in \Vc$. 
Therefore,
\be 
H(Z_k) =\frac{NL}{F_{\rm PT}}\sum_{\vv\in \Vc}\alpha^{\rm \glb}(\vv)F_i(\vv)
= \frac{NL}{F_{\rm PT}}\frac{t}{K}\sum_{\vv\in \Vc}\alpha^{\rm \glb}(\vv)F(\vv)
= \frac{tNL}{K}=ML, 
\ee 
where note  that $\Fpt=\sum_{\vv\in 
\Vc}\alpha^{\rm \glb}(\vv)F(\vv) $.
The same argument extends to arbitrary user groupings; the details are however omitted to avoid excessive technicality.

\section{Proof of Lemma \ref{lemma: vanishing upper bound, thm 1 gen schm}}
\label{sec:proof of vanishing subp ratio, app}

Suppose $K=2^{2^k}$ and $\overline{t}=2C\log_2\log_2 K$ for some constant $C$. Then
\begin{subequations}
\label{eq: vanishing bound, thm 1 gen schm}
\begin{align}
\lim\limits_{K\to\infty}\frac{F_{\rm PT}}{F_{\rm JCM}} &=\lim\limits_{K\to\infty}\frac{\prod_{i=1}^{C\log_2\log_2 K}(2i-1)}{K-2C\log_2\log_2 K}\\
&=\lim\limits_{k\to\infty}\frac{\prod_{i=1}^{k C}(2i-1)}{2^{2^k}-2k C}\\
& \le  \lim\limits_{k\to\infty}\frac{2\prod_{i=1}^{k C}(2i-1)}{2^{2^k}}\label{eq: step 0, vanishing bound, thm 1 gen schm}   \\
&\le \lim\limits_{k\to\infty}\frac{2(2kC)^{kC}}{2^{2^k}}\\
& =  \lim\limits_{k\to\infty} 2^{\log_2\left(2(2kC)^{kC} \right) - \log_2(2^{2^k})   } \\
& = \lim\limits_{k\to\infty} 2^{ 1+ C\log_2(2C)k  + Ck\log_2k     - 2^k   } \\
& = \lim\limits_{k\to\infty} 2^{ \Theta(k\log_2k) - 2^k   } =0, 
\end{align}
\end{subequations}
where (\ref{eq: step 0, vanishing bound, thm 1 gen schm}) is due to the fact that there exists $ k_0>0$ such that $K-2kC \ge K/2, \forall k \ge k_0$ since $K=2^{2^k}$ grows doubly exponentially with $k$. 
This completes the proof of Lemma~\ref{lemma: vanishing upper bound, thm 1 gen schm}.

\section{\Achvblty of DPDA \Sbp  under PT}
\label{app:PT for DPDA} 
Within the PT \frmwk, we show that  $F=K^2/4$ is \achvb for $t=2$. For $t=K-2$,  we  further show  that
$F=K(K-2)$  is \achvb when $K$ is odd, and 
$F=\frac{K(K-2)}{2} $ is \achvb when $K$ is even. 
Hence,  the DPDA results  \corrspdg to $t=2$ and $K-2$ are fully  recovered under PT.

\subsubsection{$t=2$}
Consider \grpg $\qv=(K/2,K/2)$.
The \sbf   types are  $\vv_1=(2,0), \vv_2=(1,1)$. Depending on whether $K\ge 6$, the \mgrps may differ. Specifically, if $K=4$, there is one \mgrp type $
\sv=(2,1^\dagger)$;
if $K\ge 6$, there  are  two \mgrp types $
\sv_1=(3^\dagger,0), \sv_2=(2,1^\dagger)$ with local FS vectors $\alphabm_1=(2, \star), \alphabm_2=(0, 1)  $. Both cases yield the same \gfsv $\alphaglobal=(0,1)$, implying $\Fpt=(0,1)[F(\vv_1),F(\vv_2)]=F(\vv_2)=K^2/4$.
In fact, this is exactly the PT design presented in \Ex \ref{example: (4,2,1) PT first look} (Section \ref{subsec: a first look at PT, problem description}).

\subsubsection{$t=K-2$}

First consider the case of even $K$, say $K=2m,m \ge 2$, and equal grouping $\qv =\left(\underline{2}_m \right)$.  The  \sbf and \mgrp types are $
\vv_1=\left(\underline{2}_{m-1}\right),
\vv_2=\left(\underline{2}_{m-2},\underline{1}_2\right), 
\sv=\big(\underline{2}_{m-1},1^{\dagger} \big),\;
\Ic_1=\{\vv_1,\vv_2\}
$,
yielding $\alphaglobal=(0,1)$. Hence, 
$
F_{\rm PT}=\alphaglobal[F(\vv_1),F(\vv_2)  ]^{\rm T}   =
(0,1)\big[\binom{m}{1},\binom{m}{2}\binom{2}{1}^2 \big]^{\rm T}=\frac{K(K-2)}{2}
$.

Next, consider the case of odd $K$, say $K=2m+1$.
When $K\ge 7$, the PT design in Example~\ref{example: discussion of Thm 1, PT design not unique} with  grouping $\qv=\left(
3,\underline{2}_m\right)$ can be used to \achv   $F_{\rm PT}=K(K-2)$.
When $K=3$, the \jsch itself \achvs $F=3$.
When $K=5$, under $\qv =\left(3,2 \right)$, the \sbf and  \mgrp types are $
\vv_1=\left( 1,2 \right),
\vv_2=\left( 2,1\right),
\vv_3=\left(3,0\right)$, 
$\sv_1= ( 2^{\dagger},2  ),
\sv_2= (3, 1^{\dagger} ),
\Ic_1=\{\vv_1,\vv_2\},
\Ic_2=\{\vv_2,\vv_3\}$. 
The local and global FS vectors  are $\alphabm_1=(1,2,\star)   ,  \alphabm_2=(\star,1,0)     $, and $\alphaglobal=(1,2,0)$. The memory constraint is satisfied as 
$\alphaglobal  \bm{\Delta}_1=(1,2,0)(2,-1,-1)^{\rm T}=0$.
Hence, 
$ 
\Fpt =\alphaglobal\big[F(\vv_1), F(\vv_2),F(\vv_3) \big]=15=K(K-2) $.


\bibliographystyle{IEEEtran}
\bibliography{references_newest}
\end{document}

%% file: macros.tex
\long\def\comment#1{}

\newtheorem{defn}{Definition}
\newtheorem{example}{Example}
\newtheorem{theorem}{Theorem}
\newtheorem{lemma}{Lemma}

\newtheorem{remark}{Remark}

\makeatletter
\newcommand{\subalign}[1]{%
  \vcenter{%
    \Let@ \restore@math@cr \default@tag
    \baselineskip\fontdimen10 \scriptfont\tw@
    \advance\baselineskip\fontdimen12 \scriptfont\tw@
    \lineskip\thr@@\fontdimen8 \scriptfont\thr@@
    \lineskiplimit\lineskip
    \ialign{\hfil$\m@th\scriptstyle##$&$\m@th\scriptstyle{}##$\crcr
      #1\crcr
    }%
  }
}
\newcommand{\thickhline}{%
    \noalign {\ifnum 0=`}\fi \hrule height 1pt
    \futurelet \reserved@a \@xhline
}
\newcolumntype{"}{@{\hskip\tabcolsep\vrule width 1pt\hskip\tabcolsep}}

\makeatother

\newfont{\bbb}{msbm10 scaled 700}

\newfont{\bb}{msbm10 scaled 1100}

\let \bksl\backslash

\usepackage{xcolor} 

\newcommand{\mech}{mechanism\xspace}

\newcommand{\supp}{\mathrm{supp}}

\newcommand{\occu}{occupancy\xspace}

\newcommand{\arbi}{arbitrary\xspace}

\newcommand{\deli}{delivery\xspace}
\newcommand{\Deli}{Delivery\xspace}
\newcommand{\suth}{such that\xspace}
\newcommand{\soth}{so that\xspace}
\newcommand{\sbpzed}{subpacketized\xspace}

\newcommand{\uneq}{unequal\xspace}
\newcommand{\Uneq}{Unequal\xspace}
\newcommand{\ugrping}{user grouping\xspace}

\newcommand{\glb}{global\xspace}

\newcommand{\redun}{redundant\xspace}

\newcommand{\sort}{\texttt{sort}\xspace}

\newcommand{\highlt}{highlight\xspace}
\newcommand{\memo}{memory\xspace}

\newcommand{\subseq}{subsequent\xspace}

\newcommand{\accrdto}{according to\xspace}
\newcommand{\Accrdto}{According to\xspace}

\newcommand{\Rmk}{Remark\xspace}

\newcommand{\uset}{unique set\xspace}
\newcommand{\usets}{unique sets\xspace}

\newcommand{\charzs}{characterizes\xspace}

\newcommand{\Fex}{For example\xspace}

\newcommand{\Ex}{Example\xspace}
\newcommand{\ex}{example\xspace}

\newcommand{\morethanhalf}{more-than-half\xspace}

\newcommand{\PTd}{PT design\xspace}
\newcommand{\PTds}{PT designs\xspace}

\newcommand{\achvs}{achieves\xspace}
\newcommand{\achvd}{achieved\xspace}
\newcommand{\achv}{achieve\xspace}

\newcommand{\Achvblty}{Achievability\xspace}
\newcommand{\achvb}{achievable\xspace}

\newcommand{\Itc}{In this case\xspace}

\newcommand{\indept}{independent\xspace}

\newcommand{\otws}{otherwise\xspace}
\newcommand{\grpg}{grouping\xspace}

\newcommand{\grpgs}{groupings\xspace}

\newcommand{\cdntd}{coordinated\xspace}
\newcommand{\cdntn}{coordination\xspace}

\newcommand{\schms}{schemes\xspace}
\newcommand{\schm}{scheme\xspace}

\newcommand{\Conseqtly}{Consequently\xspace}

\newcommand{\corrspdg}{corresponding\xspace}

\newcommand{\corrspds}{corresponds\xspace}

\newcommand{\advtg}{advantage\xspace}

\newcommand{\bcuz}{because\xspace}

\newcommand{\Aar}{As a result\xspace}
\newcommand{\sbfs}{subfiles\xspace}
\newcommand{\sbf}{subfile\xspace}

\newcommand{\exch}{exchange\xspace}
\newcommand{\exc}{exclusion\xspace}
\newcommand{\excs}{exclusions\xspace}

\newcommand{\perm}{permutation\xspace}
\newcommand{\perms}{permutations\xspace}
\newcommand{\grps}{groups\xspace}
\newcommand{\grp}{group\xspace}

\newcommand{\frmwk}{framework\xspace}

\newcommand{\Pkt}{Packet\xspace}
\newcommand{\pkt}{packet\xspace}
\newcommand{\pkts}{packets\xspace}

\newcommand{\reqd}{required\xspace}

\newcommand{\msg}{message\xspace}
\newcommand{\ie}{i.e.\xspace}
\newcommand{\agg}{aggregate\xspace}
\newcommand{\mtcst}{multicast\xspace}

\newcommand{\CoCa}{Coded Caching\xspace}

\newcommand{\coca}{coded caching\xspace}

\newcommand{\af}{as follows\xspace}

\newcommand{\Specly}{Specifically\xspace}

\newcommand{\resp}{respectively\xspace}

\newcommand{\wrt}{with respect to\xspace}
\newcommand{\diffce}{difference\xspace}

\newcommand{\diff}{different\xspace}

\newcommand{\jsch}{JCM scheme\xspace}

\newcommand{\Inadd}{In addition\xspace}

\newcommand{\Ip}{In particular\xspace}

\newcommand{\tx}{transmitter\xspace}
\newcommand{\txs}{transmitters\xspace}

\newcommand{\txn}{transmission\xspace}

\newcommand{\Tx}{Transmitter\xspace}
\newcommand{\Txs}{Transmitters\xspace}

\newcommand{\ssgain}{subfile saving gain\xspace}

\newcommand{\mgrp}{multicast group\xspace}
\newcommand{\mgrps}{multicast groups\xspace}
\newcommand{\Mgrp}{Multicast group\xspace}

\newcommand{\sym}{symmetric\xspace}
\newcommand{\Sym}{Symmetric\xspace}
\newcommand{\asym}{asymmetric\xspace}
\newcommand{\Asym}{Asymmetric\xspace}

\newcommand{\cmsg}{coded message\xspace}
\newcommand{\cmsgs}{coded messages\xspace}

\newcommand{\Algo}{Algorithm\xspace}
 
\newcommand{\comm}{communication\xspace}

\newcommand{\etal}{et al.\xspace}

\newcommand{\Finlen}{Finite-length\xspace}

\newcommand{\Sbp}{Subpacketization\xspace}
\newcommand{\sbp}{subpacketization\xspace}
\newcommand{\sbps}{subpacketizations\xspace}

\newcommand{\Thm}{Theorem\xspace}
\newcommand{\Thms}{Theorems\xspace}
\newcommand{\red}{reduction\xspace}
\newcommand{\reds}{reductions\xspace}

\newcommand{\param}{parameter\xspace}

\newcommand{\Thf}{Therefore\xspace}

\newcommand{\fs}{further-splitting\xspace}

\newcommand{\ow}{order-wise\xspace}
\newcommand{\owr}{order-wise reduction\xspace}

\newcommand{\WnT}{\mbox{$W_{n, \Tc}$}}
\newcommand{\tbar}{\mbox{$\overline{t}$}}
\newcommand{\alphaglobal}{\mbox{$\bm{\alpha}^{\rm  global}$}}
\newcommand{\alphatilde}{\mbox{$\widetilde{\alpha}$}}

\newcommand{\alphaglobalnb}{\mbox{$\alpha^{\rm  global}$}}

\def\Zplus{{\mathbb Z}_{\ge 0}}
\def\Zpos{{\mathbb Z}_{> 0}}

\newcommand{\examend}{\hfill\mbox{$\lozenge$}}

\newcommand{\bmalpha}{\mbox{$\bm{\alpha}$}}

\DeclareMathOperator*{\argmin}{arg\,min}
\newcommand{\alphalcm}{\mbox{$\bm{\alpha}^{\rm LCM}$}}
\newcommand{\alphajcm}{\mbox{$\alpha_{\rm JCM}$}}
\newcommand{\alphabm}{\mbox{$\bm{\alpha}$}}
\newcommand{\ith}{\mbox{$i^{\rm th}$}}
\newcommand{\jth}{\mbox{$j^{\rm th}$}}
\newcommand{\nth}{\mbox{$n^{\rm th}$}}
\newcommand{\Fpt}{\mbox{$F_{\rm PT}$}}
\newcommand{\Fjcm}{\mbox{$F_{\rm JCM}$}}
\newcommand{\Fdpda}{\mbox{$F_{\rm DPDA}$}}

\newcommand{\Fman}{\mbox{$F_{\rm MAN}$}}
\newcommand{\Rjcm}{\mbox{$R_{\rm JCM}$}}
\newcommand{\Rman}{\mbox{$R_{\rm MAN}$}}
\newcommand{\nd}{\mbox{$N_{\rm d}$}} 
\newcommand{\hatu}{\mbox{$\widehat{\Uc}$}}
\newcommand{\Utx}{\mbox{$\Uc_{\rm Tx}$}}
\newcommand{\Dtx}{\mbox{$\Dc_{\rm Tx}$}}

\newcommand{\av}{{\bf a}}
\newcommand{\bv}{{\bf b}}

\newcommand{\dv}{{\bf d}}

\newcommand{\qv}{{\bf q}}

\newcommand{\sv}{{\bf s}}

\newcommand{\vv}{{\bf v}}

\newcommand{\Am}{{\bf A}}
\newcommand{\Bm}{{\bf B}}

\newcommand{\Fm}{{\bf F}}

\newcommand{\Ac}{{\cal A}}
\newcommand{\Bc}{{\cal B}}

\newcommand{\Dc}{{\cal D}}
\newcommand{\Ec}{{\cal E}}

\newcommand{\Ic}{{\cal I}}

\newcommand{\Qc}{{\cal Q}}
\newcommand{\Rc}{{\cal R}}
\newcommand{\Sc}{{\cal S}}
\newcommand{\Tc}{{\cal T}}
\newcommand{\Uc}{{\cal U}}

\newcommand{\Vc}{{\cal V}}

\newcommand{\dd}{D2D\xspace}

\newcommand{\invod}{involved\xspace}

\newcommand{\istss}{involved subfile type sets\xspace}

\newcommand{\memoconsts}{memory constraints\xspace}
\newcommand{\memoconst}{memory constraint\xspace}
\newcommand{\Memoconst}{Memory constraint\xspace}

\newcommand{\gfsf}{global FS factor\xspace}
\newcommand{\gfsv}{global FS vector\xspace}
\newcommand{\lfsfs}{local FS factors\xspace}

\newcommand{\lfsvs}{local FS vectors\xspace}

\newcommand{\lcm}{least common multiple\xspace}
\newcommand{\vlcm}{vector least common multiple\xspace}

\newcommand{\Lambdam}{\hbox{\boldmath$\Lambda$}}

\newcommand{\eqdef}{\stackrel{\Delta}{=}}

\newcommand{\be}{\begin{equation}}
\newcommand{\ee}{\end{equation}}
\newcommand{\bea}{\begin{eqnarray}}
\newcommand{\eea}{\end{eqnarray}}

\def \type{\ttt{type}}

\let \rig\right
\let\lef \left 

\let\ttt\texttt
\let\trm\textrm
\let\tbf\textbf
\let\tit\textit

\let\mbf\mathbf

\let\udl\underline
\let \mrm \mathrm

%% file: author_tit.tex
\author{Xiang Zhang,~\IEEEmembership{Member,~IEEE},
Giuseppe Caire,~\IEEEmembership{Fellow,~IEEE},
and Mingyue Ji,~\IEEEmembership{Member,~IEEE}

\thanks{Part of this work was presented in the conference papers~\cite{9174215,9834012}.
X. Zhang and G. Caire are with the Department of Electrical Engineering and Computer Science, Technical University of Berlin, 10623 Berlin, Germany (e-mail: \{xiang.zhang, caire\}@tu-berlin.de).
M. Ji is with the Department of Electrical and Computer Engineering, University of Florida, Gainesville, FL 32611, USA
(e-mail: mingyueji@ufl.edu).
This work was done while X. Zhang and M. Ji were with the Department of Electrical and Computer Engineering, University of Utah, Salt Lake City, UT 84112, USA.}
\thanks{Corresponding author: Xiang Zhang.}
}

%% file: references_newest.bib
@STRING{globecom    = {Proc. {IEEE} Global Commun. Conf. (GLOBECOM)}}

@STRING{icc         = {Proc. {IEEE} Int. Conf. on Commun. (ICC)}}

@STRING{isit        = {Proc. {IEEE} Int. Symp. on Inform. Theory (ISIT)}}

@STRING{itw         = {Proc. {IEEE} Inform. Theory Workshop (ITW)}}

@ARTICLE{9492815,
  author={Woolsey, Nicholas and Wang, Xingyue and Chen, Rong-Rong and Ji, Mingyue},
  journal={IEEE Transactions on Cloud Computing}, 
  title={FLCD: A Flexible Low Complexity Design of Coded Distributed Computing}, 
  year={2023},
  volume={11},
  number={1},
  pages={470-483},
  doi={10.1109/TCC.2021.3098593}}

@inproceedings{chittoor2019low,
  title={Low subpacketization coded caching via projective geometry for broadcast and d2d networks},
  author={Chittoor, Hari Hara Suthan and Krishnan, Prasad},
  booktitle={2019 IEEE Global Communications Conference (GLOBECOM)},
  pages={1--6},
  year={2019},
  organization={IEEE}
}

@article{wang2025coded,
  title={Coded Caching for D2D Multi-access Networks with Linear Subpacketization via Vector Set},
  author={Wang, Jinyu and Cheng, Minquan and Wu, Youlong and Wang, Pengfei},
  journal={IEEE Transactions on Communications},
  year={2025},
  publisher={IEEE}
}

@inproceedings{nt2025d2d,
  title={D2D Coded Caching from Two Classes of Optimal DPDAs Using Cross Resolvable Designs},
  author={NT, Rashid Ummer and Rajan, B Sundar},
  booktitle={2025 IEEE Wireless Communications and Networking Conference (WCNC)},
  pages={1--6},
  year={2025},
  organization={IEEE}
}

@article{rajan2024optimal,
  title={Optimal Placement Delivery Arrays from $ t $-Designs with Application to Hierarchical Coded Caching},
  author={Rajan, B Sundar and others},
  journal={arXiv preprint arXiv:2402.07188},
  year={2024}
}

@article{nt2025hierarchical,
  title={Hierarchical Coded Caching With Low Subpacketization and Coding Delay Using Combinatorial t-Designs},
  author={NT, Rashid Ummer and Rajan, B Sundar},
  journal={IEEE Internet of Things Journal},
  year={2025},
  publisher={IEEE}
}

@article{rashid2026optimal,
  title={Optimal Device-to-Device Placement Delivery Arrays using Combinatorial Designs},
  author={Rashid Ummer, NT and Rajan, B Sundar},
  journal={Discrete Mathematics, Algorithms and Applications},
  year={2026},
  publisher={World Scientific}
}

@article{sasi2021multi,
  title={Multi-access coded caching scheme with linear sub-packetization using PDAs},
  author={Sasi, Shanuja and Rajan, B Sundar},
  journal={IEEE Transactions on Communications},
  volume={69},
  number={12},
  pages={7974--7985},
  year={2021},
  publisher={IEEE}
}

@inproceedings{zhang2019cache,
  title={Cache-aided interference management using hypercube combinatorial cache designs},
  author={Zhang, Xiang and Woolsey, Nicholas and Ji, Mingyue},
  booktitle={ICC 2019-2019 IEEE International Conference on Communications (ICC)},
  pages={1--6},
  year={2019},
  organization={IEEE}
}

@article{brunero2022fundamental,
  title={Fundamental limits of combinatorial multi-access caching},
  author={Brunero, Federico and Elia, Petros},
  journal={IEEE Transactions on Information Theory},
  volume={69},
  number={2},
  pages={1037--1056},
  year={2022},
  publisher={IEEE}
}

@article{parrinello2019fundamental,
  title={Fundamental limits of coded caching with multiple antennas, shared caches and uncoded prefetching},
  author={Parrinello, Emanuele and {\"U}nsal, Ay{\c{s}}e and Elia, Petros},
  journal={IEEE Transactions on Information Theory},
  volume={66},
  number={4},
  pages={2252--2268},
  year={2019},
  publisher={IEEE}
}

@article{yang2026low,
  title={A Low-Complexity Architecture for Multi-access Coded Caching Systems with Arbitrary User-cache Access Topology},
  author={Yang, Ting and Cheng, Minquan and Yi, Xinping and Qiu, Robert Caiming and Caire, Giuseppe},
  journal={arXiv preprint arXiv:2601.10175},
  year={2026}
}

@article{huang2026placement,
  title={Placement Delivery Array for Cache-Aided MIMO Systems},
  author={Huang, Yifei and Wan, Kai and Cheng, Minquan and Wang, Jinyan and Caire, Giuseppe},
  journal={arXiv preprint arXiv:2601.10422},
  year={2026}
}

@inproceedings{wu2023coded,
  title={Coded Caching Design for Dynamic Networks with Reduced Subpacketizations},
  author={Wu, Xianzhang and Cheng, Minquan and Chen, Li and Li, Congduan},
  booktitle={2023 IEEE International Symposium on Information Theory (ISIT)},
  pages={430--435},
  year={2023},
  organization={IEEE}
}

@article{lampiris2024adapt,
  title={Adapt or wait: Quality adaptation for cache-aided channels},
  author={Lampiris, Eleftherios and Caire, Giuseppe},
  journal={IEEE Transactions on Communications},
  year={2024},
  publisher={IEEE}
}

@article{cheng2024asymptotically,
  title={Asymptotically Optimal Coded Distributed Computing via Combinatorial Designs},
  author={Cheng, Minquan and Wu, Youlong and Li, Xianxian and Wu, Dianhua},
  journal={IEEE/ACM Transactions on Networking},
  year={2024},
  publisher={IEEE}
}

@inproceedings{cheng2024coded,
  title={Coded Caching for MISO Systems with Linear Subpacketization},
  author={Cheng, Minquan and Tan, Rongqing and Wang, Jinyu and Wu, Youlong and Li, Xianxian},
  booktitle={2024 9th International Conference on Computer and Communication Systems (ICCCS)},
  pages={593--598},
  year={2024},
  organization={IEEE}
}

@inproceedings{wei2024novel,
  title={A Novel Construction of Coded Caching Schemes with Polynomial Subpacketizations via Projective Geometry},
  author={Wei, Huimei and Cheng, Minquan and Leung, Kahin},
  booktitle={2024 IEEE Information Theory Workshop (ITW)},
  pages={496--501},
  year={2024},
  organization={IEEE}
}

@article{cheng2025new,
  title={A New Construction Structure on Coded Caching with Linear Subpacketization: Non-Half-Sum Disjoint Packing},
  author={Cheng, Minquan and Wei, Huimei and Wan, Kai and Caire, Giuseppe},
  journal={arXiv preprint arXiv:2501.11855},
  year={2025}
}

@ARTICLE{9328826,
  author={Katyal, Digvijay and Muralidhar, Pooja Nayak and Rajan, B. Sundar},
  journal={IEEE Transactions on Communications}, 
  title={Multi-Access Coded Caching Schemes From Cross Resolvable Designs}, 
  year={2021},
  volume={69},
  number={5},
  pages={2997-3010},
  doi={10.1109/TCOMM.2021.3053048}}

@ARTICLE{9037309,
  author={Zhong, Xi and Cheng, Minquan and Jiang, Jing},
  journal={IEEE Communications Letters}, 
  title={Placement Delivery Array Based on Concatenating Construction}, 
  year={2020},
  volume={24},
  number={6},
  pages={1216-1220},
  doi={10.1109/LCOMM.2020.2981071}}

@INPROCEEDINGS{9148593,
  author={Salehi, MohammadJavad and Tolli, Antti and Shariatpanahi, Seyed Pooya},
  booktitle={ICC 2020 - 2020 IEEE International Conference on Communications (ICC)}, 
  title={A Multi-Antenna Coded Caching Scheme with Linear Subpacketization}, 
  year={2020},
  volume={},
  number={},
  pages={1-6},
  doi={10.1109/ICC40277.2020.9148593}}

@ARTICLE{9284439,
  author={Mingming, Zhang and Minquan, Cheng and Jinyu, Wang and Xi, Zhong and Chen, Yishan},
  journal={IEEE Access}, 
  title={Improving Placement Delivery Array Coded Caching Schemes With Coded Placement}, 
  year={2020},
  volume={8},
  number={},
  pages={217456-217462},
  doi={10.1109/ACCESS.2020.3042849}}

@INPROCEEDINGS{8613527,
  author={Krishnan, Prasad},
  booktitle={2018 IEEE Information Theory Workshop (ITW)}, 
  title={Coded Caching via Line Graphs of Bipartite Graphs}, 
  year={2018},
  volume={},
  number={},
  pages={1-5},
  doi={10.1109/ITW.2018.8613527}}

@ARTICLE{9536664,
  author={Cheng, Minquan and Wang, Jinyu and Zhong, Xi and Wang, Qiang},
  journal={IEEE Transactions on Information Theory}, 
  title={A Framework of Constructing Placement Delivery Arrays for Centralized Coded Caching}, 
  year={2021},
  volume={67},
  number={11},
  pages={7121-7131},
  doi={10.1109/TIT.2021.3112492}}

@INPROCEEDINGS{9834012,
  author={Zhang, Xiang and Ji, Mingyue},
  booktitle={2022 IEEE 23rd International Workshop on Signal Processing Advances in Wireless Communication (SPAWC)}, 
  title={Finite-length Analysis of {D2D} Coded Caching via Exploiting Asymmetry in Delivery}, 
  year={2022},
  volume={},
  number={},
  pages={1-5},
  doi={10.1109/SPAWC51304.2022.9834012}}

@ARTICLE{8651553,
  author={Cheng, Minquan and Jiang, Jing and Yan, Qifa and Tang, Xiaohu},
  journal={IEEE Transactions on Communications}, 
  title={Constructions of Coded Caching Schemes With Flexible Memory Size}, 
  year={2019},
  volume={67},
  number={6},
  pages={4166-4176},
  doi={10.1109/TCOMM.2019.2901686}}

@ARTICLE{8080217,
  author={Yan, Qifa and Tang, Xiaohu and Chen, Qingchun and Cheng, Minquan},
  journal={IEEE Communications Letters}, 
  title={Placement Delivery Array Design Through Strong Edge Coloring of Bipartite Graphs}, 
  year={2018},
  volume={22},
  number={2},
  pages={236-239},
  doi={10.1109/LCOMM.2017.2765629}}

@ARTICLE{7539576,
  author={Shanmugam, Karthikeyan and Ji, Mingyue and Tulino, Antonia M. and Llorca, Jaime and Dimakis., Alexandros G.},
  journal={IEEE Transactions on Information Theory}, 
  title={Finite-Length Analysis of Caching-Aided Coded Multicasting}, 
  year={2016},
  volume={62},
  number={10},
  pages={5524-5537},
  doi={10.1109/TIT.2016.2599110}}

@INPROCEEDINGS{8335418,
  author={Shanmugam, Karthikeyan and Dimakis, Alexandras G. and Llorca, Jaime and Tulino, Antonia M.},
  booktitle={2017 51st Asilomar Conference on Signals, Systems, and Computers}, 
  title={A unified {Ruzsa-Szem$\acute{\rm e}$redi} framework for finite-length coded caching}, 
  year={2017},
  volume={},
  number={},
  pages={631-635},
  doi={10.1109/ACSSC.2017.8335418}}

@ARTICLE{9913463,
  author={Li, Jian and Chang, Yanxun},
  journal={IEEE Communications Letters}, 
  title={New Constructions of {D2D} Placement Delivery Arrays}, 
  year={2023},
  volume={27},
  number={1},
  pages={85-89},
  doi={10.1109/LCOMM.2022.3212679}}

@INPROCEEDINGS{9839268,
  author={Malik, Adeel and Serbetci, Berksan and Elia, Petros},
  booktitle={ICC 2022 - IEEE International Conference on Communications}, 
  title={Stochastic Coded Caching with Optimized Shared-Cache Sizes and Reduced Subpacketization}, 
  year={2022},
  volume={},
  number={},
  pages={2918-2923},
  doi={10.1109/ICC45855.2022.9839268}}

@ARTICLE{9448271,
  author={Woolsey, Nicholas and Chen, Rong-Rong and Ji, Mingyue},
  journal={IEEE Transactions on Communications}, 
  title={A New Combinatorial Coded Design for Heterogeneous Distributed Computing}, 
  year={2021},
  volume={69},
  number={9},
  pages={5672-5685},
  doi={10.1109/TCOMM.2021.3087628}}

@ARTICLE{woolsey2020d2d,  author={N. {Woolsey} and R. {Chen} and M. {Ji}},  journal={IEEE Transactions on Communications},   title={Towards Finite File Packetizations in Wireless {Device}-to-{Device} Caching Networks},   year={2020},  volume={},  number={},  pages={1-1},}

@article{salehi2020low,
  title={Low-Complexity High-Performance Cyclic Caching for Large MISO Systems},
  author={Salehi, M. and Parrinello, E. and Shariatpanahi, Seyed P. and Elia, P. and T{\"o}lli, A.},
  journal={arXiv preprint arXiv:2009.12231},
  year={2020}
}

@INPROCEEDINGS{8437323,
  author={Woolsey, Nicholas and Chen, Rong-Rong and Ji, Mingyue},
  booktitle={2018 IEEE International Symposium on Information Theory (ISIT)}, 
  title={A New Combinatorial Design of Coded Distributed Computing}, 
  year={2018},
  volume={},
  number={},
  pages={726-730},
  doi={10.1109/ISIT.2018.8437323}}

@ARTICLE{shangguan2018hpyergraph,
author={C. Shangguan and Y. Zhang and G. Ge},
journal={IEEE Transactions on Information Theory},
title={Centralized coded caching schemes: A hypergraph theoretical approach},
year={2018},
volume={},
number={},
pages={1-1},
doi={10.1109/TIT.2018.2847679},
ISSN={0018-9448},
month={},}

@ARTICLE{tang2018subpacketization,
author={L. Tang and A. Ramamoorthy},
journal={IEEE Transactions on Information Theory},
title={Coded Caching Schemes With Reduced Subpacketization From Linear Block Codes},
year={2018},
volume={64},
number={4},
pages={3099-3120},
doi={10.1109/TIT.2018.2800059},
ISSN={0018-9448},
month={April},}

@ARTICLE{yan2017pda,
author={Q. Yan and M. Cheng and X. Tang and Q. Chen},
journal={IEEE Transactions on Information Theory},
title={On the Placement Delivery Array Design for Centralized Coded Caching Scheme},
year={2017},
volume={63},
number={9},
pages={5821-5833},
doi={10.1109/TIT.2017.2725272},
ISSN={0018-9448},
month={Sept},}

@ARTICLE{8950279,
  author={Tolli, Antti and Shariatpanahi, Seyed Pooya and Kaleva, Jarkko and Khalaj, Babak Hossein},
  journal={IEEE Transactions on Wireless Communications}, 
  title={Multi-Antenna Interference Management for Coded Caching}, 
  year={2020},
  volume={19},
  number={3},
  pages={2091-2106},
  doi={10.1109/TWC.2019.2962686}}

@INPROCEEDINGS{8815564,
  author={Lampiris, Eleftherios and Elia, Petros},
  booktitle={2019 IEEE 20th International Workshop on Signal Processing Advances in Wireless Communications (SPAWC)}, 
  title={Bridging two extremes: Multi-antenna Coded Caching with Reduced Subpacketization and {CSIT}}, 
  year={2019},
  volume={},
  number={},
  pages={1-5},
  doi={10.1109/SPAWC.2019.8815564}}

@INPROCEEDINGS{9739199,
  author={Salehi, Mohammad Javad and Bakhshzad Mahmoodi, Hamidreza and Toelli, Antti},
  booktitle={WSA 2021; 25th International ITG Workshop on Smart Antennas}, 
  title={A Low-Subpacketization High-Performance MIMO Coded Caching Scheme}, 
  year={2021},
  volume={},
  number={},
  pages={1-6},
  doi={}}

@INPROCEEDINGS{9348098,
  author={Salehi, MohammadJavad and Tolli, Antti},
  booktitle={GLOBECOM 2020 - 2020 IEEE Global Communications Conference}, 
  title={Diagonal Multi-Antenna Coded Caching for Reduced Subpacketization}, 
  year={2020},
  volume={},
  number={},
  pages={1-6},
  doi={10.1109/GLOBECOM42002.2020.9348098}}

@ARTICLE{9477627,
  author={Chittoor, Hari Hara Suthan and Krishnan, Prasad and Sree, K. V. Sushena and Mamillapalli, Bhavana},
  journal={IEEE Transactions on Information Theory}, 
  title={Subexponential and Linear Subpacketization Coded Caching via Projective Geometry}, 
  year={2021},
  volume={67},
  number={9},
  pages={6193-6222},
  doi={10.1109/TIT.2021.3095471}}

@INPROCEEDINGS{9739200,
  author={Zhao, Hui and Lampiris, Eleftherios and Caire, Giuseppe and Elia, Petros},
  booktitle={WSA 2021; 25th International ITG Workshop on Smart Antennas}, 
  title={Multi-Antenna Coded Caching Analysis in Finite {SNR} and Finite Subpacketization}, 
  year={2021},
  volume={},
  number={},
  pages={1-6},
  doi={}}

@ARTICLE{9103948,  author={Konstantinidis, Konstantinos and Ramamoorthy, Aditya},  journal={IEEE/ACM Transactions on Networking},   title={Resolvable Designs for Speeding Up Distributed Computing},   year={2020},  volume={28},  number={4},  pages={1657-1670},  doi={10.1109/TNET.2020.2992989}}

@ARTICLE{8620232,  author={Wang, Jinyu and Cheng, Minquan and Yan, Qifa and Tang, Xiaohu},  journal={IEEE Transactions on Communications},   title={Placement Delivery Array Design for Coded Caching Scheme in {D2D} Networks},   year={2019},  volume={67},  number={5},  pages={3388-3395},  doi={10.1109/TCOMM.2019.2893942}}

@article{wang2017placement,
  title={On the Placement Delivery Array Design for Coded Caching Scheme in D2D Networks},
  author={Wang, Jinyu and Cheng, Minquan and Yan, Qifa and Tang, Xiaohu},
  journal={arXiv preprint arXiv:1712.06212},
  year={2017}
}

@ARTICLE{9369971,  author={Zhang, Xiang and Woolsey, Nicholas and Ji, Mingyue},  journal={IEEE Transactions on Wireless Communications},   title={Cache-Aided Interference Management Using Hypercube Combinatorial Design With Reduced Subpacketizations and Order Optimal Sum-Degrees of Freedom},   year={2021},  volume={20},  number={8},  pages={4797-4810},  doi={10.1109/TWC.2021.3062264}}

@INPROCEEDINGS{9174215,  author={Zhang, Xiang and Yang, Xianfeng Terry and Ji, Mingyue},  booktitle={2020 IEEE International Symposium on Information Theory (ISIT)},   title={A New Design Framework on {D2D} Coded Caching with Optimal Rate and Less Subpacketizations},   year={2020},  volume={},  number={},  pages={1699-1704},  doi={10.1109/ISIT44484.2020.9174215}}

@article{yapar2019optimality,
  title={On the optimality of {D2D} coded caching with uncoded cache placement and one-shot delivery},
  author={Yapar, {\c{C}}a{\u{g}}kan and Wan, Kai and Schaefer, Rafael F and Caire, Giuseppe},
  journal={IEEE Transactions on Communications},
  volume={67},
  number={12},
  pages={8179--8192},
  year={2019},
  publisher={IEEE}
}

@article{li2017fundamental,
  title={A fundamental tradeoff between computation and communication in distributed computing},
  author={Li, Songze and Maddah-Ali, Mohammad Ali and Yu, Qian and Avestimehr, A Salman},
  journal={IEEE Transactions on Information Theory},
  volume={64},
  number={1},
  pages={109--128},
  year={2017},
  publisher={IEEE}
}

@article{maddah2014fundamental,
  title={Fundamental limits of caching},
  author={Maddah-Ali, M. A. and Niesen, U.},
  journal={Information Theory, IEEE Trans. on},
  volume={60},
  number={5},
  pages={2856--2867},
  year={2014},
  publisher={IEEE}
}

@ARTICLE{ji2016fundamental,
author={M. Ji and G. Caire and A. F. Molisch},
journal={IEEE Transactions on Information Theory},
title={Fundamental Limits of Caching in Wireless {D2D} Networks},
year={2016},
volume={62},
number={2},
pages={849-869},
doi={10.1109/TIT.2015.2504556},
ISSN={0018-9448},
month={Feb},}

@INPROCEEDINGS{8006726,
  author={Shanmugam, Karthikeyan and Tulino, Antonia M. and Dimakis, Alexandros G.},
  booktitle={2017 IEEE International Symposium on Information Theory (ISIT)}, 
  title={Coded caching with linear subpacketization is possible using {Ruzsa-Szem$\acute{\rm e}$redi} graphs}, 
  year={2017},
  volume={},
  number={},
  pages={1237-1241},
  doi={10.1109/ISIT.2017.8006726}}
